\documentclass[12pt]{article}
\usepackage[charter]{mathdesign}

\usepackage[T1]{fontenc}
\usepackage[latin9]{inputenc}
\usepackage[letterpaper]{geometry}
\usepackage{mathrsfs}
\usepackage{amsmath}
\usepackage{amsthm}
\usepackage{setspace}
\usepackage[authoryear]{natbib}
\usepackage{microtype}
\usepackage[unicode=true,
 bookmarks=false,
 breaklinks=false,pdfborder={0 0 1},backref=section,colorlinks=false]
 {hyperref}

\makeatletter
\theoremstyle{plain}
\newtheorem{assumption}{\protect\assumptionname}
\theoremstyle{plain}
\newtheorem{lem}{\protect\lemmaname}
\theoremstyle{plain}
\newtheorem{prop}{\protect\propositionname}
\theoremstyle{plain}
\newtheorem{cor}{\protect\corollaryname}
\theoremstyle{remark}
\newtheorem{rem}{\protect\remarkname}
\theoremstyle{plain}
\newtheorem{thm}{\protect\theoremname}

\usepackage{amsfonts}
\usepackage{eurosym}
\usepackage{ulem}
\usepackage{graphicx}
\usepackage{caption}
\usepackage{color}
\usepackage{sectsty}
\usepackage{comment}
\usepackage{footmisc}
\usepackage{caption}
\usepackage{pdflscape}
\usepackage{subfigure}
\usepackage{array}
\usepackage{upgreek}
\usepackage{bbm}

\usepackage{amsthm}
\usepackage{graphicx}
\usepackage{verbatim}
\usepackage{ulem}
\usepackage{textpos}
\usepackage{changepage}
\usepackage{url}

    \DeclareMathOperator*{\argmax}{\arg\!\max}

\newcolumntype{L}[1]{>{\raggedright\let\newline\\arraybackslash\hspace{0pt}}m{#1}}
\newcolumntype{C}[1]{>{\centering\let\newline\\arraybackslash\hspace{0pt}}m{#1}}
\newcolumntype{R}[1]{>{\raggedleft\let\newline\\arraybackslash\hspace{0pt}}m{#1}}

\makeatother

\providecommand{\assumptionname}{Assumption}
\providecommand{\corollaryname}{Corollary}
\providecommand{\lemmaname}{Lemma}
\providecommand{\propositionname}{Proposition}
\providecommand{\remarkname}{Remark}
\providecommand{\theoremname}{Theorem}

\begin{document}
\title{{\Large{}Nonparametric Identification of Production Function, Total
Factor Productivity, and Markup from Revenue Data}\thanks{Acknowledgement: We are grateful to comments made at seminars and
conferences with regard to an earlier version of this work. We thank
Zheng Han, Yoko Sakamoto and Makoto Tanaka for excellent research
assistance. Sugita acknowledges financial supports from JSPS KAKENHI
(grant numbers: 17H00986, 19H01477, and 19H00594).}{\Large{} }}
\author{Hiroyuki Kasahara\thanks{Department of Economics, University of British Columbia, Canada. (Email:
hkasahar@mail.ubc.ca)} \and Yoichi Sugita\thanks{Graduate School of Economics, Hitotsubashi University, Japan. (E-mail:
yoichi.sugita@r.hit-u.ac.jp)}}
\date{\today}
\maketitle
\begin{abstract}
Commonly used methods of production function and markup estimation
assume that a firm\textquoteright s output quantity can be observed
as data, but typical datasets contain only revenue, not output quantity.
We examine the nonparametric identification of production function
and markup from revenue data when a firm faces a general nonparametric
demand function under imperfect competition. Under standard assumptions,
we provide the constructive nonparametric identification of various
firm-level objects: gross production function, total factor productivity,
price markups over marginal costs, output prices, output quantities,
a demand system, and a representative consumer\textquoteright s utility
function.
\end{abstract}

\section{Introduction}

The estimation of production function and markup is a core tool used
in empirical analyses of market outcomes.\footnote{\citet{griliches_mairesse_1999} and \citet*{ACKERBERG20074171} provide
excellent surveys on production function estimations.} The residual of an estimated production function, total factor productivity
(TFP), is widely used to measure firm-level technological efficiency
(see \citet{bartelsman2000understanding} and \citet{syverson2011determines}
for recent surveys) and its contribution to aggregate efficiency (e.g.,
\citealp{olley1996dynamics}). Researchers often estimate the elasticity
of production functions to analyze technological changes (e.g., \citealp{van2003productivity};
\citealp{doraszelski2018measuring}) and price markups over marginal
costs (e.g., \citealp{hall1988relation}; \citealp{de2012markups}).
The estimation of firm-level markup via production function has been
widely applied in various topics and complements markup estimation
via demand function (e.g., \citealp*{blp1995ecta}) in economic analysis
of firm's market power.

Commonly used methods of production function and markup estimation
assume that a firm's output quantity can be observed as data. However,
typical firm-level datasets contain only revenue, not output quantity.
Therefore, in practice, many applications use revenue deflated by
an industry-level price deflator as output.\footnote{A few studies use firm-level datasets that include output quantity
(e.g., \citealp*{foster2008reallocation}; \citealp*{de2016prices};
\citealp{lu2015trade}; \citealp{nishioka2019measuring}). However,
those quantity datasets are available only for a limited number of
countries, industries, and years, and they are not easily accessible
to all researchers.} For production function estimation, this practice may be justified
under perfect competition where an output price is exogenous and identical
across firms. However, ever since \citet{ma44ecma}'s pioneering study,
several researchers have voiced cautions and suggested that the practice
may not be justified under imperfect competition; they show that using
revenue as output can significantly bias the identification of production
functions (e.g., \Citealp{klette1996jae}; \citealp*{de2011product})
and TFP (e.g., \citealp*{foster2008reallocation}; \citealp*{katayama2009firm};
\citealp*{de2011product}). Furthermore, as shown in \citealp*{bond2020some},
using revenue in place of output quantity may lead to serious biases
in estimation of firm's markups. Despite such criticism, the practice
of using revenue in place of output quantity persists in many applications
given a lack of output quantity data.

In the existing literature, it is not known whether identifying production
functions and markups from firm-level revenue data is possible without
imposing parametric assumptions. This paper contributes to the literature
of production function and markup estimation by establishing nonparametric
identification of production function, TFP, and markup from revenue
data. The proof is constructive and the required assumption is similar
to the standard assumption in the production function literature except
that we impose additional assumptions on firm's demand function.

Following \citet{ma44ecma}, \citet{klette1996jae} and \citet{de2011product},
we explicitly model a demand function that an individual firm faces
as a function of its output and observable characteristics that are
excluded from the production function.\footnote{\citet*{de2020rise} study an alternative approach using an exogenous
variable to remove output price variation from revenue data. } While each of these earlier studies examines a demand function with
a constant and identical demand elasticity\textemdash something that
implies identical markups across firms\textemdash we consider a general
nonparametric demand function that generates rich heterogeneity in
various firm-level outcomes, including markups; for this reason, we
can address the bias from markup heterogeneity across firms that the
literature has criticized. In other respects, our method requires
the standard assumptions and can be implemented using typical data
found in empirical applications.

We develop a three-step identification approach that combines the
control function approach developed by \citet{olley1996dynamics},
\citet{levinsohn2003estimating}, and \citet*{ackerberg2015identification}
and the first-order condition approach recently developed by \citet*{gandhi2020identification}.\footnote{These approaches assumed quantity data or perfect competition. \citet{gandhi2020identification}
also examined an imperfect competition with a constant elastic demand
as in \citet{klette1996jae} and \citet{de2011product} where markups
must be constant and identical across firms. } Following \citet{levinsohn2003estimating} and \citet{ackerberg2015identification},
the inverse function of a material demand function serves as a control
function for TFP. In the first step, we identify revenue as a function
of inputs and observable demand shifters by using the control function;
this first step corresponds to that of \citet{ackerberg2015identification}.
Our novel second step identifies the control function for TFP by applying
the nonparametric identification of transformation models (e.g., \citealp{horowitz1996semiparametric})
examined by \citet*{ekeland2004identification} and \citet*{chiappori2015nonparametric}.
By identifying the control function, TFP is identified (up to normalization)
from the dynamics of inputs, without output data. In the third step,
we identify a production function, markup, and a demand function,
using the first-order condition for the material and the control function
identified in the second step.

Our method identifies various objects from the revenue data. In our
main setting, markup and output elasticities are identified up to
scale; an output price, an output quantity, a gross production function,
and TFP are identified up to scale and location. Identification is
cross-sectional so that the identified objects can vary over time.
With an additional assumption of \emph{local} constant returns to
scale, we identify the levels of markup and output elasticities; we
may also identify an output price, an output quantity, a production
function, and TFP up to location.\footnote{\citet*{flynn2019measuring} used \emph{global} constant returns to
scale to identify a production function. In subsection \ref{subsec:LCRS},
we clarify \emph{local} and \emph{global} constant returns to scale.} Finally, if we are willing to assume monopolistic competition (without
imposing free entry), we further identify a demand system and a utility
function of a representative consumer\textemdash specifically, \citet{matsuyama2017beyond}'s
homothetic demand system with a single aggregator (HSA)\textemdash that
can be used for a counter-factual analysis and a welfare analysis.\footnote{One frequently sees within the literature an assumption of market
structure for the identification of demand and supply side objects.
For example, \citet*{blp1995ecta} identify firm-level marginal costs
by specifying oligopolistic competition; meanwhile, \citet*{ekeland2004identification}
and \citet*{heckman2010nonparametric} identify various demand and
supply side objects of a hedonic model by exploiting the properties
of perfect competition.}

The remainder of this paper is organized as follows. Section \ref{sec:Potential-Biases-When}
summarizes previous studies on how using revenue as output could bias
the identification of production function, TFP, and markup; readers
familiar with the literature can skip this section and proceed to
Section \ref{sec:Identification}. Subsection \ref{subsec:Setting}
explains our setting, and subsection \ref{subsec:example} demonstrates
our three-step approach by offering a parametric example. Subsection
\ref{subsec:Nonparametric} presents our nonparametric identification
results, and subsection \ref{subsec:Normalization} discusses additional
assumptions for fixing scale and location normalization. Subsection
\ref{subsec:demand} examines the identification of a demand system
and a representative consumer's utility function. Both subsection
\ref{subsec:Alternative-Settings} and the Appendix present identification
results in alternative settings, including endogenous labor input,
endogenous firm-level observable demand shifters, unobservable demand
shifters, and i.i.d. productivity shocks. Section \ref{sec:Concluding-Remarks}
provides concluding remarks.

\section{Biases from Using Revenue as Output Quantity\label{sec:Potential-Biases-When}}

This section summarizes possible biases in the identification of production
function, TFP, and markup when revenue is used as an output quantity.
We denote the logarithms of the price, output, and revenue of firm
$i$ at time $t$ as $p_{it}$, $y_{it}$, and $r_{it}:=p_{it}+y_{it}$,
respectively. Suppose that these variables are related via the inverse
demand function $p_{it}=\psi_{it}(y_{it})$ and the revenue function
$r_{it}=\varphi_{it}(y_{it}):=y_{it}+\psi_{it}(y_{it})$. Let $y_{it}=f_{t}(m_{it},k_{it},l_{it})+\omega_{it}$
be firm $i$'s production function where $\omega_{it}$ is TFP and
$x_{it}:=(m_{it},k_{it},l_{it})$ is a vector of the logarithms of
material, capital, and labor, respectively. To highlight the sources
of biases from using revenue as output, assume that TFP is identical
across firms within time $t$, with $\omega_{it}=\omega_{t}$ for
all $i$. This simplification eliminates an additional and well-known
source of bias, correlations between inputs and TFP. 

From the first-order condition for profit maximization, $P_{it}\left(1+\psi_{it}'(y_{it})\right)=MC_{it}$,
the elasticity of revenue with respect to output is equal to the inverse
of markup:
\begin{equation}
\frac{d\varphi_{it}(y_{it})}{dy}=\frac{MC_{it}}{P_{it}}.\label{eq:markup_foc}
\end{equation}
Under perfect competition where $P_{it}=MC_{it}$, the variation in
revenue across firms coincides with that of output. However, they
are generally different when markups vary across firms.

Suppose that, using revenue as output, a researcher identifies a true
relationship between revenue and inputs, $\tilde{\varphi}_{it}(x_{it}):=\varphi_{it}(f_{t}(x_{it})+\omega_{t})$
to use $\tilde{\varphi}_{it}(x_{it})$ as a proxy for $f_{t}(x_{it})$.
Prior studies show that the use of revenue as output could cause biases
in three forms. First, \citet{ma44ecma} and \citet{klette1996jae}
establish that, from (\ref{eq:markup_foc}), the elasticity of $\tilde{\varphi}_{it}(x_{it})$
relates to the true elasticity of $f_{t}(x_{it})$ via markup:
\begin{equation}
\frac{\partial\tilde{\varphi}_{it}(x_{it})}{\partial v_{it}}=\frac{MC_{it}}{P_{it}}\frac{\partial f_{t}(x_{it})}{\partial v_{it}}\text{ for }v_{it}\in\{m_{it},k_{it},l_{it}\}.\label{eq:elasticity_bias}
\end{equation}
Thus, output elasticities would be underestimated by the extent of
markup. 

Second, \citet{katayama2009firm} and \citet{de2011product} demonstrated
a bias in TFP estimates. Let $d\omega_{t}$ be a TFP change. Suppose
that a TFP change for firm $i$ is estimated as a change in revenue
with inputs being fixed, $d\tilde{\omega}_{it}=\left.d\tilde{\varphi}_{it}(x_{it})\right|_{dx_{it}=0}.$
From (\ref{eq:markup_foc}), we see that this TFP estimate relates
to the true TFP change via markup: 
\begin{align}
d\tilde{\omega}_{it} & =\frac{MC_{it}}{P_{it}}d\omega_{t}.\label{eq:TFP}
\end{align}
Therefore, TFP would be underestimated by the extent of markup.

Finally, \citet{bond2020some} show that markup estimates using the
method of \citet{hall1988relation} and \citet{de2012markups} are
generally biased when revenue elasticity is used in place of output
elasticity. Suppose a firm is a price-taker of flexible input $v$.
\citet{hall1988relation} and \citet{de2012markups} developed the
following equation relating to markup and output elasticity with respect
to $v$ as:
\begin{equation}
\frac{P_{it}}{MC_{it}}=\frac{\partial f_{t}(x_{it})/\partial v_{it}}{\alpha_{it}^{v}}\label{eq:DLW_markup}
\end{equation}
where $\alpha_{it}^{v}$ is the ratio of expenditure on input $v$
to revenue. If a researcher uses $\partial\tilde{\varphi}_{it}(x_{it})/\partial m_{it}$
instead of $\partial f_{t}(x_{it})/\partial m_{it}$ in markup equation
(\ref{eq:DLW_markup}), then from (\ref{eq:elasticity_bias}), the
estimated markup is 1:
\begin{equation}
\frac{\partial\tilde{\varphi}_{it}(x_{it})/\partial v_{it}}{\alpha_{it}^{v}}=\frac{\frac{MC_{it}}{P_{it}}\frac{\partial f_{t}(x_{it})}{\partial v_{it}}}{\alpha_{it}^{v}}=1.\label{eq:markup_one}
\end{equation}
In such a case, the markup would be underestimated.\footnote{Result (\ref{eq:markup_one}) by \citet{bond2020some} relies on the
assumption that a researcher can correctly identify $\tilde{\varphi}_{it}(x_{it})$.
In practice, misspecification of $\tilde{\varphi}_{it}(x_{it})$ could
derive markup estimates (\ref{eq:markup_one}) that contain some information
on true markups. For instance, \citeauthor{de2012markups} (2012,
Section VI) show that when $f$ is Cobb\textendash Douglas, it is
possible to identify the effect of firm-level variables (e.g., export)
on markups.}

\citet{klette1996jae} and \citet{de2011product} developed methods
by which to identify production functions from revenue data, by assuming
a constant elastic demand function with an identical elasticity.\footnote{\citet{katayama2009firm} also developed a method by which to identify
production functions from revenue data. Their method allows for markup
heterogeneity but requires the ability to estimate firm's marginal
costs from total costs. } However, with this specific demand function, markups must be constant
and identical across firms. Studies estimating markups from quantity
data report substantial heterogeneity in markups across firms (e.g.,
\citealp*{de2016prices}; \citealp{lu2015trade}; \citealp{nishioka2019measuring}).
To address the biases arising from firm-level markup heterogeneity,
we extend the approach of \citet{klette1996jae} and \citet{de2011product}
by incorporating a general nonparametric demand function that allows
for variable and heterogeneous markups.

\section{Identification\label{sec:Identification}}

\subsection{Setting\label{subsec:Setting}}

We denote the logarithm of physical output, material, capital, and
labor as $y_{it}$, $m_{it}$, $k_{it}$, and $l_{it}$, respectively,
with their respective supports denoted as $\mathcal{Y}$, $\mathcal{M}$,
$\mathcal{K}$, and $\mathcal{L}$. We collect the three inputs (material,
capital, and labor) into a vector as $x_{it}:=(m_{it},k_{it},l_{it})'\in\mathcal{X}:=\mathcal{M}\times\mathcal{K}\times\mathcal{L}$. 

At time $t$, output $y_{it}$ relates to inputs $x_{it}=(m_{it},k_{it},l_{it})'$
via the production function:
\begin{equation}
y_{it}=f_{t}(x_{it})+\omega_{it},\label{prod}
\end{equation}
where the firm's TFP $\omega_{it}$ follows an exogenous first-order
stationary Markov process given by 
\begin{align}
\omega_{it} & =h(\omega_{it-1})+\eta_{it},\label{omega}
\end{align}
where we assume that neither $h(\cdot)$ nor the marginal distribution
of $\eta_{it}$ change over time.\footnote{$h(\cdot)$ can include a firm's observable exogenous characteristics.}

The demand function for a firm's product is strictly decreasing in
its price, and its inverse demand function is given by 
\begin{equation}
p_{it}=\psi_{t}(y_{it},z_{it}),\label{eq:inverse_demand}
\end{equation}
where $p_{it}$ is the logarithm of output price and $z_{it}\in\mathcal{Z}$
is an observable firm characteristic that affects firm's demand (e.g.,
export status in \citet{de2012markups}). $z_{it}$ can be either
a continuous or discrete vector; in the main text below, $z_{it}$
is assumed to be continuous and exogenous\textemdash that is, $z_{it}\perp\eta_{it}$.
In subsection \ref{subsec:Alternative-Settings} and the Appendix,
we present the identification results when $z_{it}$ is discrete and/or
may correlate with $\eta_{t}$.

The inverse demand function (\ref{eq:inverse_demand}) generalizes
the constant elastic demand function examined by \citet{ma44ecma},
\citet{klette1996jae} and \citet{de2011product}. Although $\psi_{t}$
is nonparametric, (\ref{eq:inverse_demand}) implicitly makes two
assumptions. First, $\psi_{t}(\cdot)$ is a common function for all
firms once the observed characteristics $z_{it}$ are controlled for.
This implies that unobserved demand shifters must be common for all
firms\textemdash that is, $\psi_{t}$ can be written as $\psi_{t}(y_{it},z_{it},A_{t})$
where $A_{t}$ is a vector of unobserved variables and can include
an aggregate price/quantity index. In subsection \ref{subsec:Alternative-Settings},
we discuss the case where $\psi_{t}(\cdot)$ includes a firm-level
unobservable demand shifter such as quality. Second, $\psi_{t}(\cdot)$
represents a demand curve that an individual firm takes as given.
This is satisfied in the case of monopolistic competition (without
free entry) where each firm takes $A_{t}$ as given.

Let $\bar{r}_{it}$ and $\mathcal{\bar{R}}$ be the logarithm of (true)
revenue and its support, respectively. Revenue $r_{it}$ in the data
is observed with a measurement error $\varepsilon_{it}$, $r_{it}=\bar{r}_{it}+\varepsilon_{it}$.
Then, from (\ref{prod}), the observed revenue relates to output and
input as follows:
\begin{align}
r_{it} & =\varphi_{t}(y_{it},z_{it})+\varepsilon_{it}\nonumber \\
 & =\varphi_{t}(f_{t}(m_{it},k_{it},l_{it})+\omega_{it},z_{it})+\varepsilon_{it}\label{rev}
\end{align}
where $\varphi_{t}(y_{it},z_{it}):=\psi_{t}(y_{it},z_{it})+y_{it}.$

We assume that $l_{it}$ and $k_{it}$ are predetermined at the end
of the last period $t-1$, while $m_{it}$ is flexibly chosen after
observing $\omega_{it}$.\footnote{In subsection \ref{subsec:Alternative-Settings}, we present identification
when $l_{it}$ also correlates with $\omega_{it}$.} Specifically, $m_{it}=\mathbb{M}_{t}\left(\omega_{it},k_{it},l_{it},z_{it}\right)$
is chosen at time $t$ by:
\begin{equation}
\mathbb{M}_{t}\left(\omega_{it},k_{it},l_{it},z_{it}\right)\in\arg\max_{m}\exp(\varphi_{t}(f_{t}(m,k_{it},l_{it})+\omega_{it},z_{it}))-\exp(p_{t}^{m}+m),\label{eq:profit_maximization}
\end{equation}
where $p_{t}^{m}$ denotes the logarithm of the material input price
at time $t$, which is common to all firms. A firm is assumed to be
a price-taker for material input. 

Equation (\ref{rev}) highlights two identification issues raised
by \citet{ma44ecma}. First, $m_{it}$ correlates with the unobservable
$\omega_{it}$. Second, $r_{it}$ relates to $x_{it}=(m_{it},k_{it},l_{it})$
via two unknown nonlinear functions $\varphi_{t}(\cdot,z_{it})$ and
$f_{t}(\cdot)$, and two unobservables $\omega_{it}$ and $\varepsilon_{it}$.\footnote{In subsection \ref{subsec:Alternative-Settings} and the Appendix,
we present identification when a firm receives an i.i.d. shock $e_{it}$
to output and then, the firm's revenue includes a non-additive error,
$r_{it}=\varphi_{t}(f_{t}(x_{it})+\omega_{it}+e_{it},z_{it})$.} To address these issues via a control function and a transformation
model, we first make the following assumptions. 
\begin{assumption}
\label{A-1} (a) $f_{t}(\cdot)$ is continuously differentiable with
respect to $(m,k,l)$ on $\mathcal{M}\times\mathcal{K}\times\mathcal{L}$
and strictly increasing in $m$. (b) For every $z\in\mathcal{Z}$,
$\varphi_{t}(\cdot,z)$ is strictly increasing and invertible with
its inverse $\varphi_{t}^{-1}(\bar{r},z)$, which is continuously
differentiable with respect to $(\bar{r},z)$ on $\mathcal{\bar{R}}\times\mathcal{Z}$.
(c) For every $(k,l,z)\in\mathcal{K}\times\mathcal{L}\times\mathcal{Z}$,
$\mathbb{M}_{t}(\cdot,k,l,z)$ is strictly increasing and invertible
with its inverse $\mathbb{M}_{t}^{-1}(m,k,l,z)$, which is continuously
differentiable with respect to $(m,k,l,z)$ on $\mathcal{M}\times\mathcal{K}\times\mathcal{L\times\mathcal{Z}}$.
(d) $\varepsilon_{t}$ is mean independent of $x_{t}$ and $z_{t}$
with $E[\varepsilon_{t}|x_{t},z_{t}]=0$.
\end{assumption}
Assumptions \ref{A-1} (a) and (b) are standard assumptions about
smooth production and demand functions. Assumption \ref{A-1} (b)
$\partial\varphi_{t}(y,z)/\partial y>0$ is equivalent to that the
elasticity of demand with respect to price, $-\left(\partial\psi_{t}(y,z)/\partial y\right)^{-1}$,
is greater than 1; this necessarily holds under profit maximization.
Therefore, Assumption \ref{A-1} (b) is innocuous as long as we analyze
the outcomes of profit maximization. Assumption \ref{A-1} (c) is
a standard assumption in the control function approach that uses material
as a control function for TFP \citep{levinsohn2003estimating,ackerberg2015identification}.

The inverse function of the material demand function with respect
to TFP
\[
\omega_{it}=\mathbb{M}_{t}^{-1}(m_{it},k_{it},l_{it},z_{t})
\]
 is used as a control function for $\omega_{it}$. Since $\partial\varphi_{t}(y_{t},z_{t})/\partial y_{t}>0$,
there exists the inverse function $\varphi_{t}^{-1}(\cdot,z_{t})$
so that the revenue function $\bar{r}_{it}=\varphi_{t}(f_{t}(x_{itt})+\omega_{it})$
can be written as: 
\begin{align}
\varphi_{t}^{-1}\left(\bar{r}_{it},z_{it}\right)=f_{t}(x_{it})+\mathbb{M}_{t}^{-1}(x_{it},z_{it}).\label{eq:model_step2}
\end{align}

In the following, we identify $\varphi_{t}^{-1}\left(\cdot\right)$,
$f_{t}(\cdot)$, and $\mathbb{M}_{t}^{-1}(\cdot)$ from the distribution
of variables in the data. Let $v_{t}:=(k_{t},l_{t},z_{t},x_{t-1},z_{t-1})'\in\mathcal{V}:=\mathcal{K}\times\mathcal{L}\times\mathcal{Z}\times\mathcal{X}\times\mathcal{Z}$.
Data includes a random sample of firms $\{r_{it},v_{it}\}_{i=1}^{N}$
from the population. For instance, the variable $x_{it}$ of firm
$i$ is considered as a realization of the random variable $x_{t}$.
Given a sufficiently large $N$, an econometrician can recover their
joint distributions.
\begin{assumption}
\label{A-data} The following information at time $t$ is known: (a)
the conditional distribution $G_{m_{t}|v_{t}}(m_{t}|v_{t})$ of $m_{t}$
given $v_{t}$; (b) the conditional expectation $E[r_{t}|x_{t},z_{t}]$
of $r_{t}$ given $(x_{t},z_{t})$; (c) firm's expenditure on material
$\exp(p_{t}^{m}+m_{it})$.
\end{assumption}
Assumption \ref{A-data} (a) is required for the identification of
$\mathbb{M}_{t}^{-1}(\cdot)$. Assumptions \ref{A-data} (b) and (c)
are additionally required for the identification of $\varphi_{t}^{-1}\left(\cdot\right)$
and $f_{t}(\cdot)$. Typical production datasets include those variables
in Assumption \ref{A-data}.

Let $\{\varphi_{t}^{*-1}(\cdot),f_{t}^{*}(\cdot),\mathbb{M}_{t}^{*-1}(\cdot)\}$
be the true model structure that satisfies (\ref{eq:model_step2}).
Then, for any $(a_{1t},a_{2t},b_{t})\in\mathbb{R}^{2}\times\mathbb{R}_{++}$,
\begin{align}
\varphi_{t}^{-1}\left(\bar{r}_{t},z_{t}\right) & =(a_{1t}+a_{2t})+b_{t}\varphi_{t}^{*-1}\left(\bar{r}_{t},z_{t}\right),\ f_{t}(x_{t})=a_{1t}+b_{t}f_{t}^{*}(x_{t}),\nonumber \\
 & \text{{and}\ }\mathbb{M}_{t}^{-1}(x_{t},z_{t})=a_{2t}+b_{t}\mathbb{M}_{t}^{*-1}(x_{t},z_{t})\label{eq:equivalence}
\end{align}
also satisfy (\ref{eq:model_step2}), and the true structure $\{\varphi_{t}^{*-1}(\cdot),f_{t}^{*}(\cdot),\mathbb{M}_{t}^{*-1}(\cdot)\}$
is observationally equivalent to the structure (\ref{eq:equivalence}).
That is, the structure $\{\varphi_{t}^{-1}(\cdot),f_{t}(\cdot),\mathbb{M}_{t}^{-1}(\cdot)\}$
is identified only up to location and scale normalization $(a_{1t},a_{2t},b_{t})$
from restriction (\ref{eq:model_step2}).

Therefore, identification requires location and scale normalization.
We fix $(a_{1t},a_{2t},b_{t})$ in (\ref{eq:equivalence}) by fixing
the values of $\{\varphi_{t}^{-1}(\cdot),f_{t}(\cdot),\mathbb{M}_{t}^{-1}(\cdot)\}$
at some points. Specifically, choosing two points $(m_{t1}^{*},k_{t}^{*},l_{t}^{*},z_{t}^{*})$
and $(m_{t0}^{*},k_{t}^{*},l_{t}^{*},z_{t}^{*})$ on the support $\mathcal{X}\times\mathcal{Z}$
where $m_{t0}^{*}<m_{t1}^{*}$, we denote
\begin{equation}
c_{1t}:=f_{t}(m_{t0}^{*},k_{t}^{*},l_{t}^{*}),\ c_{2t}=\mathbb{M}_{t}^{-1}(m_{t0}^{*},k_{t}^{*},l_{t}^{*},z_{t}^{*}),\ \text{{and}\ }c_{3t}:=\mathbb{M}_{t}^{-1}(m_{t1}^{*},k_{t}^{*},l_{t}^{*},z_{t}^{*}).\label{eq:norm}
\end{equation}
Note that $\partial\mathbb{M}_{t}^{-1}/\partial m_{t}>0$ implies
that $c_{2t}<c_{3t}$. Then, there exists a unique one-to-one mapping
between $(c_{1t},c_{2t},c_{3t})$ in (\ref{eq:norm}) and $(a_{1t},a_{2t},b_{t})$
in (\ref{eq:equivalence}) such that $b_{t}=\left(c_{3t}-c_{2t}\right)/\left(\mathbb{M}_{t}^{*-1}(m_{t1}^{*},k_{t}^{*},l_{t}^{*},z_{t}^{*})-\mathbb{M}_{t}^{*-1}(m_{t0}^{*},k_{t}^{*},l_{t}^{*},z_{t}^{*})\right)$,
$a_{1t}=c_{1t}-b_{1t}f_{t}^{*}(m_{t0}^{*},k_{t}^{*},l_{t}^{*})$ and
$a_{2t}=c_{2t}-b_{1t}\mathbb{M}_{t}^{*-1}(m_{t0}^{*},k_{t}^{*},l_{t}^{*},z_{t}^{*})$.
Thus, we can fix the value of $(a_{1t},a_{2t},b_{t})$ by choosing
arbitrary values $(c_{1t},c_{2t},c_{3t})\in\mathbb{R}^{3}$ that satisfies
$c_{2t}<c_{3t}$. In particular, we impose the following normalization
that corresponds to (N2) in \citet{chiappori2015nonparametric}.
\begin{assumption}
(Normalization) \label{A-2} The support $\mathcal{X}\times\mathcal{Z}$
includes two points $(m_{t1}^{*},k_{t}^{*},l_{t}^{*},z_{t}^{*})$
and $(m_{t0}^{*},k_{t}^{*},l_{t}^{*},z_{t}^{*})$ such that $c_{1t}=c_{2t}=0$
and $c_{3t}=1$ in (\ref{eq:norm}).
\end{assumption}
As \citet{chiappori2015nonparametric} demonstrates, this choice of
normalization makes the identification proofs transparent.

\subsection{Identification in a Parametric Example\label{subsec:example}}

Before presenting the nonparametric identification results, we demonstrate
our identification approach by applying it to a simple parametric
example. Consider a monopolistically competitive market where each
firm $i$ faces the following constant elastic inverse demand function:

\begin{equation}
p_{it}=\alpha_{t}(z_{it})+(\rho(z_{it})-1)y_{it},\label{eq:demand}
\end{equation}
where $\alpha_{t}(z_{it})$ and $0<\rho(z_{it})\le1$ are unknown
parameters.\footnote{The demand function (\ref{eq:demand}) can be derived from a constant
elasticity of substitution (CES) utility function; $a_{t}(z_{t})$
implicitly includes aggregate expenditure and an aggregate price index.} The markup equals $1/\rho(z_{it})$ and depends on the exogenous
scalar $z_{it}\in\mathcal{Z}:=\{1,0\}$ such that $z_{it}\perp\eta_{it}$.
Firm $i$ has a Cobb\textendash Douglas production function and $\omega_{it}$
follows a first-order autoregressive (AR(1)) process:
\begin{align}
y_{it} & =\theta_{0}+\theta_{m}m_{it}+\theta_{k}k_{it}+\theta_{l}l_{it}+\omega_{it},\nonumber \\
\omega_{it} & =h_{0}+h_{1}\omega_{it-1}+\eta_{it},\label{eq:AR1}
\end{align}
where $\{\theta_{0},\theta_{m},\theta_{k},\theta_{l},h_{0},h_{1}\}$
are unknown parameters. The firm's revenue function is expressed as:
\begin{equation}
r_{it}=\alpha_{t}(z_{it})+\rho(z_{it})\theta_{0}+\rho(z_{it})\theta_{m}m_{it}+\rho(z_{it})\theta_{k}k_{it}+\rho(z_{it})\theta_{l}l_{it}+\rho(z_{it})\omega_{it}+\varepsilon_{it}.\label{eq:r_t}
\end{equation}
The first-order condition for (\ref{eq:profit_maximization}),
\begin{equation}
\rho(z_{it})\theta_{m}=\frac{\exp(p_{t}^{m}+m_{it})}{\exp(r_{it}-\varepsilon_{it})},\label{eq:foc}
\end{equation}
determines the control function for $\omega_{it}$ as
\begin{align}
\omega_{it} & =\mathbb{M}_{t}^{-1}(m_{it},k_{it},l_{it},z_{it})=\beta_{t}(z_{it})+\beta_{m}(z_{it})m_{it}+\beta_{k}k_{it}+\beta_{l}l_{it}\label{eq:omega}
\end{align}
where $\beta_{t}(z_{it}):=\left(p_{t}^{m}-\alpha_{t}(z_{it})-\theta_{0}-\ln\rho(z_{it})\theta_{m}\right)/\rho(z_{it})$,
$\beta_{m}(z_{it}):=\left(1-\rho(z_{it})\theta_{m}\right)/\rho(z_{it})>0$,
$\beta_{k}:=-\theta_{k}$ and $\beta_{l}:=-\theta_{l}$. 

For notational brevity, assume that the support $\mathcal{X}\times\mathcal{Z}$
includes two points $(m_{t1}^{*},k_{t}^{*},l_{t}^{*},z_{t}^{*})=(0,0,0,0)$
and $(m_{t0}^{*},k_{t}^{*},l_{t}^{*},z_{t}^{*})=(1,0,0,0)$. Following
Assumption \ref{A-2}, we fix the location and scale of $f_{t}(\cdot)$
and $\mathbb{M}_{t}^{-1}(\cdot)$ by imposing the following normalization:
\begin{align}
0 & =f_{t}(0,0,0)=\theta_{0},\,0=\mathbb{M}_{t}^{-1}(0,0,0,0)=\beta_{t}(0),\nonumber \\
1 & =\mathbb{M}_{t}^{-1}(1,0,0,0)=\beta_{t}(0)+\beta_{m}(0)\label{eq:normalization_example}
\end{align}
which implies $\theta_{0}=0$, $\beta_{t}(0)=0$, and $\beta_{m}(0)=1$.

Our identification approach follows three steps.

\paragraph{Step 1: Identification of Measurement Errors}

The first step removes the measurement error $\varepsilon_{it}$ in
the spirit of \citet{ackerberg2015identification}. Substituting (\ref{eq:omega})
into (\ref{eq:r_t}) and using $\theta_{0}=0$, we obtain two expressions
of $r_{it}$ as follows:
\begin{align}
r_{it}= & \left(\alpha_{t}(z_{it})+\rho(z_{it})\beta_{t}(z_{it})\right)+\rho(z_{it})\left(\theta_{m}+\beta_{m}(z_{it})\right)m_{it}\nonumber \\
 & +\rho(z_{it})\left(\theta_{k}+\beta_{k}\right)k_{it}+\rho(z_{it})\left(\theta_{l}+\beta_{l}\right)l_{it}+\varepsilon_{it}\label{eq:expr1}\\
= & \phi(z_{it})+m_{it}+\varepsilon_{it},\label{eq:expr2}
\end{align}
where $\phi(z_{it}):=\alpha_{t}(z_{it})+\rho(z_{it})\beta_{t}(z_{it})$.
Applying the conditional moment restriction $E[\varepsilon_{it}|m_{t},z_{t}]=0$
for the second expression (\ref{eq:expr2}), we identify $\phi(z_{it})$,
$\bar{r}_{it}$ and $\varepsilon_{it}$ by
\[
\phi(z_{t})=E[r_{it}-m_{it}|m_{t},z_{t}],\bar{r}_{it}=\phi(z_{it})\text{ and }\varepsilon_{it}=r_{it}-m_{it}-\phi(z_{it}).
\]

\paragraph{Step 2: Identification of Control Function and TFP}

The second step identifies the control function $\mathbb{M}_{t}^{-1}(\cdot)$.
Substituting (\ref{eq:omega}) into the AR(1) process (\ref{eq:AR1})
leads to
\begin{equation}
\mathbb{M}_{t}^{-1}(m_{it},k_{it},l_{it},z_{it})=h_{0}+h_{1}\mathbb{M}_{t-1}^{-1}(m_{it-1},k_{it-1},l_{it-1},z_{it-1})+\eta_{it}.\label{eq:model_ex1}
\end{equation}
Since $\mathbb{M}_{t}^{-1}(m_{it},k_{it},l_{it},z_{it})$ is linear
in $m_{it}$ from (\ref{eq:omega}), we can rearrange (\ref{eq:model_ex1})
as: 
\begin{align}
m_{it} & =\gamma(z_{it},z_{t-1})+\gamma_{k}(z_{it})k_{it}+\gamma_{l}(z_{it})l_{it}+\delta_{m}(z_{it},z_{it-1})m_{it-1}\nonumber \\
 & +\delta_{k}(z_{it})k_{it-1}+\delta_{l}(z_{it})l_{it-1}+\tilde{\eta}_{it},\label{eq:m_model}
\end{align}
where 
\begin{align}
\gamma(z_{it},z_{it-1}) & :=\frac{h_{0}-\beta_{t}(z_{it})+h_{1}\beta_{t-1}(z_{it-1})}{\beta_{m}(z_{it})},\,\gamma_{k}(z_{it}):=-\frac{\beta_{k}}{\beta_{m}(z_{it})},\gamma_{l}(z_{it}):=-\frac{\beta_{l}}{\beta_{m}(z_{it})},\nonumber \\
\delta_{m}(z_{it},z_{it-1}) & :=\frac{h_{1}\beta_{m}(z_{it-1})}{\beta_{m}(z_{it})},\,\delta_{k}(z_{it}):=\frac{h_{1}\beta_{k}}{\beta_{m}(z_{it})},\,\delta_{l}:=\frac{h_{1}\beta_{l}}{\beta_{m}(z_{it})},\tilde{\eta}_{it}:=\frac{\eta_{it}}{\beta_{m}(z_{it})}.\label{eq:parameters}
\end{align}
For a given $(z_{it},z_{it-1})$, (\ref{eq:m_model}) is a linear
model. Since $E\left[\left.\tilde{\eta}_{it}\right|v_{it}\right]=E\left[\left.\eta_{it}\right|v_{it}\right]/\beta_{m}(z_{it})=0$,
where $v_{it}:=(k_{it},l_{it},x_{it-1},z_{it},z_{it-1})$, we can
identify $\{\gamma(z_{it},z_{t-1})$, $\gamma_{k}(z_{it})$, $\gamma_{l}(z_{it})$,
$\delta_{m}(z_{it},z_{it-1})$, $\delta_{k}(z_{it})$, $\delta_{l}(z_{it})\}$
in (\ref{eq:m_model}) from the conditional moment restriction $E\left[\left.\tilde{\eta}_{it}\right|v_{it}\right]=0$. 

From (\ref{eq:normalization_example}) and (\ref{eq:parameters}),
we identify the parameters of the control function (under the normalization
(\ref{eq:normalization_example})) as:
\begin{align*}
\beta_{t}(1) & =\gamma(0,0)-\gamma(1,0)\frac{\gamma_{k}(0)}{\gamma_{k}(1)},\ \beta_{m}(1)=\frac{\gamma_{k}(0)}{\gamma_{k}(1)},\ \beta_{k}=\gamma_{k}(0)\text{ and }\beta_{l}=\gamma_{l}(0).
\end{align*}

\paragraph{Step 3: Identification of Production Function and Markup}

The final step identifies the parameters of the demand and production
functions. Comparing the two expressions of $r_{it}$ in (\ref{eq:expr1})
and (\ref{eq:expr2}), we obtain the following relationships:
\begin{align}
\alpha_{t}(z_{it})+\rho(z_{it})\beta_{t}(z_{it}) & =\phi(z_{it}),\,\rho(z_{it})\left(\theta_{m}+\beta_{m}(z_{it})\right)=1,\nonumber \\
\theta_{k}=-\beta_{k} & \text{ and}\,\theta_{l}=-\beta_{l}.\label{eq:restriction_2nd}
\end{align}
Given that $(\beta_{t}(z_{t}),\beta_{m}(z_{t}),\beta_{k},\beta_{l})$
are identified in step 2, the first line in (\ref{eq:restriction_2nd})
contains four equations (two equations for two values of $z_{it}\in\{0,1\}$)
and five parameters $(\alpha_{t}(0),\alpha_{t}(1),\rho(0),\rho(1),\theta_{m})$.
Therefore, to identify these parameters, we need a further restriction.

Following \citet{gandhi2020identification}, we use as an additional
restriction the first-order condition for material (\ref{eq:foc}).
The first-order condition (\ref{eq:foc}) implies that the revenue
share of material expenditure on the right hand side of (\ref{eq:foc})
is a function of $z_{it}$. Using $\varepsilon_{it}$, we obtain the
revenue share of material expenditure $\exp(p_{t}^{m}+m_{it})/\exp(r_{it}-\varepsilon_{it})$
and identify it as a function of $z_{it}$ by taking its expectation
conditional on $z_{it}$:
\[
s(z_{t}):=E\left[\left.\frac{\exp(p_{t}^{m}+m_{it})}{\exp(\bar{r}_{it})}\right|z_{t}\right].
\]
Then, we obtain an additional restriction on the parameters:
\begin{equation}
\rho(z_{it})\theta_{m}=s(z_{it}).\label{eq:foc3}
\end{equation}
From (\ref{eq:restriction_2nd}) and (\ref{eq:foc3}), we identify
the parameters of the demand and production functions as follows
\begin{align*}
\rho(0) & =1-s(0),\,\rho(1)=\frac{1-s(1)}{\beta_{m}(1)},\,\\
\alpha_{t}(0) & =\phi(0),\,\alpha_{t}(1)=\phi(1)-\rho(1)\beta_{t}(1),\\
\theta_{0} & =0,\,\theta_{m}=\frac{s(0)}{1-s(0)},\,\theta_{k}=-\beta_{k}\text{ and }\theta_{l}=-\beta_{l}.
\end{align*}

Note that the parameters are identified under the scale and location
normalization of $f_{t}(\cdot)$ and $\mathbb{M}_{t}^{-1}(\cdot)$
in (\ref{eq:normalization_example}). Let $\theta_{i}$ ($i=0,m,k,l$)
and $\beta_{j}(z_{t})$ ($j=t,m,k,l$) be those parameters identified
above and let $\theta_{j}^{*}$ and $\beta_{i}^{*}(z_{t})$ be the
true parameters. Then, there exist unknown normalization parameters
$(a,b)\in\mathbb{R}\times\mathbb{R}_{+}$ such that $\theta_{0}=a+b\theta_{0}^{*}$,
$\beta_{t}=a+b\beta_{t}^{*}$, $\theta_{i}=b\theta_{i}^{*}$, $\beta_{j}(z_{t})=b\beta_{j}^{*}(z_{t})$.
We can fix the normalization by imposing further restrictions. For
instance, if constant returns to scale $\theta_{m}^{*}+\theta_{k}^{*}+\theta_{l}^{*}=1$
are imposed, then the scale parameter $b$ can be identified as follows:
\[
b=b\left(\theta_{m}^{*}+\theta_{k}^{*}+\theta_{l}^{*}\right)=\theta_{m}+\theta_{k}+\theta_{l}=\frac{s(0)}{1-s(0)}-\beta_{k}-\beta_{l}.
\]
 We discuss in subsection \ref{subsec:Normalization} additional assumptions
for fixing normalization.

The above identification argument is illustrative, but it relies on
the linearity of $\mathbb{M}_{t}^{-1}(m_{it},k_{it},l_{it},z_{it})$
in $m_{it}$, which holds only\textbf{ }under restrictive parametric
assumptions. Extending the argument, the following subsection establishes
nonparametric identification.

\subsection{Nonparametric Identification\label{subsec:Nonparametric}}

\subsubsection{Step 1: Identification of Measurement Error}

The first step removes the measurement error $\varepsilon_{it}$.
Substituting the control function $\omega_{it}=\mathbb{M}_{t}^{-1}(m_{it},k_{it},l_{it},z_{it})$,
the revenue function (\ref{rev}) can be written as:
\begin{align*}
r_{it} & =\varphi_{t}\left(f_{t}(x_{it})+\mathbb{M}_{t}^{-1}(x_{it},z_{it}),z_{it}\right)+\varepsilon_{it}\\
 & =\phi_{t}(x_{it},z_{it})+\varepsilon_{it},
\end{align*}
where $\phi_{t}(x_{t},z_{t}):=\varphi_{t}\left(f(x_{t})+\mathbb{M}_{t}^{-1}\left(x_{t},z_{t}\right),z_{t}\right)$.
From Assumption \ref{A-1}, $\phi_{t}(\cdot)$ is continuously differentiable.
From $E\left[\varepsilon_{it}|x_{t},z_{t}\right]=0$, we can identify
$\phi_{t}(\cdot)$, $\bar{r}_{it}$, and $\varepsilon_{it}$ as:
\begin{align}
\phi_{t}(x_{t},z_{t}) & =E\left[r_{it}|x_{t},z_{t}\right],\,\bar{r}_{it}=\phi_{t}(x_{it},z_{it})\text{ and }\varepsilon_{it}=r_{it}-\phi_{t}(x_{it},z_{it}).\label{eq:phi_t}
\end{align}

\begin{lem}
\label{lem:step2} Suppose that Assumptions \ref{A-1}\textendash \ref{A-data}
hold. Then, we can identify $\phi_{t}(\cdot)$, $\bar{r}_{it}$, and
$\varepsilon_{it}$ as in (\ref{eq:phi_t}).
\end{lem}
Hereafter, $\phi_{t}(\cdot)$, $\bar{r}_{it}$, and $\varepsilon_{it}$
are assumed to be known.\footnote{As will be shown, $\omega_{it}$ is identified in step 2 independently
of step 1. Therefore, one can think of an alternative approach that
first identifies $\omega_{it}$ and then regresses $r_{it}$ on $(x_{it},z_{it},\omega_{it})$
to obtain $E\left[r_{it}|x_{it},z_{it},\omega_{it}\right]$ instead
of $E\left[r_{it}|x_{it},z_{it}\right]$. However, it is not possible
to identify $E\left[r_{it}|x_{it},z_{it},\omega_{it}\right]$ because
$\omega_{it}=\mathbb{M}_{t}^{-1}(x_{it},z_{it})$ is a deterministic
function of $(x_{it},z_{it}$). Once $(x_{it},z_{it})$ are conditioned,
there is no remaining source of variation in $\omega_{it}$.}

\subsubsection{Step 2: Identification of Control Function and TFP}

From (\ref{omega}), the control function $\omega_{it}=\mathbb{M}_{t}^{-1}(m_{it},k_{it},l_{it},z_{it})$
satisfies
\begin{align}
\mathbb{M}_{t}^{-1}(m_{it},k_{it},l_{it},z_{it}) & =\bar{h}_{t}\left(x_{it-1},z_{it-1}\right)+\eta_{it},\label{eq:model}
\end{align}
where $\bar{h}_{t}\left(x_{t-1},z_{t-1}\right):=h\left(\mathbb{M}_{t-1}^{-1}(m_{t-1},k_{t-1},l_{t-1},z_{t-1})\right)$.
As $\partial\mathbb{M}_{t}^{-1}/\partial m_{it}>0$, given the values
of $(k_{it},l_{it},z_{it})$, the dependent variable in (\ref{eq:model})
is a monotonic transformation of $m_{it}$. Therefore, the model (\ref{eq:model})
belongs to a class of transformation models, the identification of
which \citet{chiappori2015nonparametric} analyze. 

We make the following assumption, which corresponds to Assumptions
A1\textendash A3, A5, and A6 in \citet{chiappori2015nonparametric}.\footnote{Assumption \ref{A-1} (c) corresponds to Assumption A4 of \citet{chiappori2015nonparametric}.}
\begin{assumption}
\label{A-3} (a) The distribution $G_{\eta}(\cdot)$ of $\eta$ is
absolutely continuous with a density function $g_{\eta}(\cdot)$ that
is continuous on its support. (b) $\eta_{t}$ is independent of $v_{t}:=(k_{t},l_{t},z_{t},x_{t-1},z_{t-1})'\in\mathcal{V}:=\mathcal{K}\times\mathcal{L}\times\mathcal{Z}\times\mathcal{X}\times\mathcal{Z}$
with $E[\eta_{t}|v_{t}]=0$. (c) $v_{t}$ is continuously distributed
on $\mathcal{V}$. (d) Support $\varOmega$ of $\omega_{t}$ is an
interval $[\text{\ensuremath{\underbar{\ensuremath{\omega}}}},\bar{\omega}]\subset\mathbb{R}$
where $\text{\ensuremath{\underbar{\ensuremath{\omega}}}}<0$ and
$1<\bar{\omega}$. (e) $h(\cdot)$ is continuously differentiable
with respect to $\omega$ on $\Omega$. (f) The set $\mathcal{A}_{q_{t-1}}:=\{(x_{t-1},z_{t-1})\in\mathcal{X}\times\mathcal{Z}:\partial G_{m_{t}|v_{t}}(m_{t}|v_{t})/\partial q_{t-1}\neq0\text{ for all }(m_{t},k_{t},l_{t},z_{t})\in\mathcal{M}\times\mathcal{K}\times\mathcal{L}\times\mathcal{Z}\}$
is nonempty for some $q_{t-1}\in\{k_{t-1},l_{t-1},m_{t-1},z_{t-1}\}$. 
\end{assumption}
We can relax Assumption \ref{A-3}(b) by allowing $z_{t}$ and $l_{t}$
to correlate with $\eta_{t}$, which we discuss this in subsection
\ref{subsec:Alternative-Settings}. Assumption \ref{A-3}(d) holds
without loss of generality because we can choose any two points on
the support of $\omega_{t}$ without changing the essence of our argument.
Assumption \ref{A-3}(f) can be interpreted as a generalized rank
condition, thus implying that a given exogenous variable $q_{t-1}$
has a causal impact on $(m_{t},k_{t},l_{t},z_{t})$. Suppose $g_{\eta}\left(\eta\right)>0$
for all $\eta\in\mathbb{R}$. Then, as will be shown below (in (\ref{eq:dG3})),
Assumption \ref{A-3}(f) holds if and only if
\begin{align*}
\frac{\partial\bar{h}\left(\tilde{x}_{t-1},\tilde{z}_{t-1}\right)}{\partial q_{t-1}} & =h'\left(\mathbb{M}_{t-1}^{-1}(\tilde{x}_{t-1},\tilde{z}_{t-1})\right)\frac{\partial\mathbb{M}_{t-1}^{-1}(\tilde{x}_{t-1},\tilde{z}_{t-1})}{\partial q_{t-1}}\neq0
\end{align*}
for some $(\tilde{x}_{t-1},\tilde{z}_{t-1})$ and some $q_{t-1}\in\{k_{t-1},l_{t-1},m_{t-1},z_{t-1}\}$.
This condition is equivalent to (1) $\omega_{t-1}$ has a causal impact
on $\omega_{t}$ ($h'(\omega_{t-1})\neq0$) and (2) $q_{t-1}$ has
a causal impact on $m_{t-1}$, ($\partial\mathbb{M}_{t-1}/\partial q_{t-1}\neq0$).
These conditions must be satisfied for at least one exogenous variable
$q_{t-1}$ and some point $(\tilde{x}_{t-1},\tilde{z}_{t-1})$.

Proposition \ref{P-step1} shows that the control function is identified
from the distribution of $(m_{it},v_{it})$.
\begin{prop}
\label{P-step1} Suppose that Assumptions \ref{A-1}\textendash \ref{A-3}
hold. Then, we can identify $\mathbb{M}_{t}^{-1}(m_{t},k_{t},l_{t},z_{t})$
up to scale and location and $G_{\eta}(\cdot)$ up to the scale normalization
of $\eta_{t}$.
\end{prop}
\begin{proof}
The proof follows the proof of Theorem 1 in \citet{chiappori2015nonparametric}.
In view of equation (\ref{eq:model}), the conditional distribution
of $m_{t}$ given $v_{t}$ satisfies 
\begin{align*}
G_{m_{t}|v_{t}}(m_{t}|v_{t}) & =G_{\eta_{t}|v_{t}}\left(\mathbb{M}_{t}^{-1}(m_{t},k_{t},l_{t},z_{t})-\bar{h}_{t}\left(x_{t-1},z_{t-1}\right)|v_{t}\right)\\
 & =G_{\eta}\left(\mathbb{M}_{t}^{-1}(m_{t},k_{t},l_{t},z_{t})-\bar{h}_{t}\left(x_{t-1},z_{t-1}\right)\right),
\end{align*}
where the second equality follows from $\eta_{t}\perp v_{t}$ in Assumption
\ref{A-3}(b). Let $q_{t}\in\{m_{t},k_{t},l_{t},z_{t}\}$ and $q_{t-1}\in\{k_{t-1},l_{t-1},m_{t-1},z_{t-1}\}$.
The derivatives of $G_{m_{t}|v_{t}}(m_{t}|v_{t})$ are
\begin{align}
\frac{\partial G_{m_{t}|v_{t}}\left(m_{t}|v_{t}\right)}{\partial q_{t}} & =\frac{\partial\mathbb{M}_{t}^{-1}(m_{t},k_{t},l_{t},z_{t})}{\partial q_{t}}g_{\eta}\left(\mathbb{M}_{t}^{-1}(m_{t},k_{t},l_{t},z_{t})-\bar{h}_{t}\left(x_{t-1},z_{t-1}\right)\right),\label{eq:dG2}\\
\frac{\partial G_{m_{t}|v_{t}}\left(m_{t}|v_{t}\right)}{\partial q_{t-1}} & =-\frac{\partial\bar{h}\left(x_{t-1},z_{t-1}\right)}{\partial q_{t-1}}g_{\eta}\left(\mathbb{M}_{t}^{-1}(m_{t},k_{t},l_{t},z_{t})-\bar{h}_{t}\left(x_{t-1},z_{t-1}\right)\right).\label{eq:dG3}
\end{align}
Using Assumption \ref{A-3}(f), we can choose $q_{t-1}\in\{k_{t-1},l_{t-1},m_{t-1},z_{t-1}\}$
and $(\tilde{x}_{t-1},\tilde{z}_{t-1})\in\mathcal{A}_{q_{t-1}}$ such
that $\partial G_{m_{t}|v_{t}}\left(m_{t}|k_{t},l_{t},z_{t},\tilde{x}_{t-1},\tilde{z}_{t-1}\right)/\partial q_{t-1}\neq0$
for all $(m_{t},k_{t},l_{t},z_{t})\in\mathcal{M}\times\mathcal{K}\times\mathcal{L}\times\mathcal{Z}$. 

Dividing (\ref{eq:dG2}) by (\ref{eq:dG3}), we derive
\begin{align}
\frac{\partial\mathbb{M}_{t}^{-1}(m_{t},k_{t},l_{t},z_{t})}{\partial q_{t}} & =-\frac{\partial\bar{h}\left(\tilde{x}_{t-1},\tilde{z}_{t-1}\right)}{\partial q_{t-1}}\frac{\partial G_{m_{t}|v_{t}}\left(m_{t}|k_{t},l_{t},z_{t},\tilde{x}_{t-1},\tilde{z}_{t-1}\right)/\partial q_{t}}{\partial G_{m_{t}|v_{t}}\left(m_{t}|k_{t},l_{t},z_{t},\tilde{x}_{t-1},\tilde{z}_{t-1}\right)/\partial q_{t-1}}.\label{eq:dM/dq}
\end{align}
Then, from (\ref{eq:dM/dq}) for $q_{t}=m_{t}$ and the normalization
in Assumption \ref{A-2}, we obtain
\begin{align}
1 & =\mathbb{M}_{t}^{-1}(m_{t1}^{*},k_{t}^{*},l_{t}^{*},z_{t}^{*})-\mathbb{M}_{t}^{-1}(m_{t0}^{*},k_{t}^{*},l_{t}^{*},z_{t}^{*})\nonumber \\
 & =-\frac{1}{S_{q_{t-1}}}\frac{\partial\bar{h}\left(\tilde{x}_{t-1},\tilde{z}_{t-1}\right)}{\partial q_{t-1}},\label{eq:hbar}
\end{align}
where
\[
S_{q_{t-1}}:=\left(\int_{m_{t0}^{*}}^{m_{t1}^{*}}\frac{\partial G_{m_{t}|v_{t}}\left(m|k_{t}^{*},l_{t}^{*},z_{t}^{*},\tilde{x}_{t-1},\tilde{z}_{t-1}\right)/\partial m_{t}}{\partial G_{m_{t}|v_{t}}\left(m|k_{t}^{*},l_{t}^{*},z_{t}^{*},\tilde{x}_{t-1},\tilde{z}_{t-1}\right)/\partial q_{t-1}}dm\right)^{-1}.
\]
Then, we identify $\partial\bar{h}\left(\tilde{x}_{t-1},\tilde{z}_{t-1}\right)/\partial q_{t-1}=-S_{q_{t-1}}$.
Substituting this into (\ref{eq:dM/dq}), $\partial\mathbb{M}_{t}^{-1}(m_{t},k_{t},l_{t},z_{t})/\partial q_{t}$
for $q_{t}\in\{m_{t},k_{t},l_{t},z_{t}\}$ are identified as follows:
\begin{align}
\frac{\partial\mathbb{M}_{t}^{-1}(m_{t},k_{t},l_{t},z_{t})}{\partial q_{t}} & =S_{q_{t-1}}\frac{\partial G_{m_{t}|v_{t}}\left(m_{t}|k_{t},l_{t},z_{t},\tilde{x}_{t-1},\tilde{z}_{t-1}\right)/\partial q_{t}}{\partial G_{m_{t}|v_{t}}\left(m_{t}|k_{t},l_{t},z_{t},\tilde{x}_{t-1},\tilde{z}_{t-1}\right)/\partial q_{t-1}}.\label{eq:dM_dq_identified}
\end{align}
Integrating (\ref{eq:dM_dq_identified}) with respective to $q_{t}\in\{m_{t},k_{t},l_{t},z_{t}\}$
obtains 
\begin{align}
\mathbb{M}_{t}^{-1}(m_{t},k_{t},l_{t},z_{t})= & \mathbb{M}_{t}^{-1}(m_{t},k_{t},l_{t},z_{t})-\mathbb{M}_{t}^{-1}(m_{t0}^{*},k_{t},l_{t},z_{t})\nonumber \\
+ & \mathbb{M}_{t}^{-1}(m_{t0}^{*},k_{t},l_{t},z_{t})-\mathbb{M}_{t}^{-1}(m_{t0}^{*},k_{t}^{*},l_{t},z_{t})\nonumber \\
+ & \mathbb{M}_{t}^{-1}(m_{t0}^{*},k_{t}^{*},l_{t},z_{t})-\mathbb{M}_{t}^{-1}(m_{t0}^{*},k_{t}^{*},l_{t}^{*},z_{t})\nonumber \\
+ & \mathbb{M}_{t}^{-1}(m_{t0}^{*},k_{t}^{*},l_{t}^{*},z_{t})-\mathbb{M}_{t}^{-1}(m_{t0}^{*},k_{t}^{*},l_{t}^{*},z_{t}^{*})\nonumber \\
= & \int_{m_{t0}^{*}}^{m_{t}}\frac{\partial\mathbb{M}_{t}^{-1}(s,k_{t},l_{t},z_{t})}{\partial m_{t}}ds+\int_{k_{t}^{*}}^{k_{t}}\frac{\partial\mathbb{M}_{t}^{-1}(m_{t0}^{*},s,l_{t},z_{t})}{\partial k_{t}}ds\nonumber \\
+ & \int_{l_{t}^{*}}^{l_{t}}\frac{\partial\mathbb{M}_{t}^{-1}(m_{t0}^{*},k_{t}^{*},s,z_{t})}{\partial l_{t}}ds+\int_{z_{t}^{*}}^{z_{t}}\frac{\partial\mathbb{M}_{t}^{-1}(m_{t0}^{*},k_{t}^{*},l_{t}^{*},s)}{\partial z_{t}}ds,\label{eq:M_t_const}
\end{align}
where the first equality follows from $\mathbb{M}_{t}^{-1}(m_{t0}^{*},k_{t}^{*},l_{t}^{*},z_{t}^{*})=0$
in Assumption \ref{A-2}. Substituting the identified derivatives
of $\mathbb{M}_{t}^{-1}(\cdot)$ in (\ref{eq:dM_dq_identified}) into
(\ref{eq:M_t_const}), we can identify $\mathbb{M}_{t}^{-1}(m_{t},k_{t},l_{t},z_{t})$
for all $\left(m_{t},k_{t},l_{t},z_{t}\right)$.

Finally, from $\omega_{it}=\mathbb{M}_{t}^{-1}(m_{it},k_{it},l_{it},z_{it})$,
we can identify $\bar{h}_{t}(x_{t-1},z_{t-1})=E\left[\omega_{it}|x_{t-1},z_{t-1}\right]$
and $\eta_{it}=\omega_{it}-\bar{h}_{t}(x_{it-1},z_{it-1})$. Thus,
we can identify the distribution of $\eta_{t}$, $G_{\eta_{t}}(\eta)$. 
\end{proof}

\subsubsection{Step 3: Identification of Production Function and Markup}

The final step identifies production function, markup and other remaining
objects. From $\bar{r}=\phi_{t}(x_{t},z_{t})=\varphi_{t}(f_{t}(x_{t})+\mathbb{M}_{t}^{-1}\left(x_{t},z_{t}\right),z_{t})$
and the monotonicity of $\varphi_{t}$, differentiating $\varphi_{t}^{-1}(\phi(x_{t},z_{t}),z_{t})=f_{t}(x_{t})+\mathbb{M}_{t}^{-1}\left(x_{t},z_{t}\right)$
with respect to $q_{t}\in\{m_{t},k_{t},l_{t}\}$ and $z_{t}$ gives:
\begin{align}
\frac{\partial\varphi_{t}^{-1}(\bar{r}_{t},z_{t})}{\partial\bar{r}_{t}}\frac{\partial\phi_{t}(x_{t},z_{t})}{\partial q_{t}} & =\frac{\partial f_{t}(x_{t})}{\partial q_{t}}+\frac{\partial\mathbb{M}_{t}^{-1}\left(x_{t},z_{t}\right)}{\partial q_{t}},\label{eq:phi_qt}\\
\frac{\partial\varphi_{t}^{-1}(\bar{r}_{t},z_{t})}{\partial\bar{r}_{t}}\frac{\partial\phi_{t}(x_{t},z_{t})}{\partial z_{t}} & =\frac{\partial\mathbb{M}_{t}^{-1}\left(x_{t},z_{t}\right)}{\partial z_{t}}-\frac{\partial\varphi_{t}^{-1}(\bar{r}_{t},z_{t})}{\partial z_{t}}.\label{eq:phi_zt}
\end{align}
Note that $\partial\varphi_{t}^{-1}(\bar{r}_{t},z_{t})/\partial\bar{r}_{t}=\left(\partial\varphi_{t}(y_{t},z_{t})/\partial y_{t}\right)^{-1}$
represents the markup from (\ref{eq:markup_foc}). If the markup $\partial\varphi_{t}^{-1}(\bar{r}_{t},z_{t})/\partial\bar{r}_{t}$
were known, then equations (\ref{eq:phi_qt}) and (\ref{eq:phi_zt})
could identify $\partial f_{t}(x_{t})/\partial q_{t}$ and $\partial\varphi_{t}^{-1}(\bar{r}_{t},z_{t})/\partial z_{t}$
given that $\mathbb{M}_{t}^{-1}(x_{t},z_{t})$ is identified. However,
since the markup is unknown, identification requires further restriction.
Following \citet{gandhi2020identification}, we use the first-order
condition with respect to the material as an additional restriction.
\begin{assumption}
\label{A-FOC} The first-order condition with respect to material
for the profit maximization problem (\ref{eq:profit_maximization})
\begin{equation}
\frac{\partial f_{t}(x_{it})}{\partial m_{it}}=\frac{\partial\varphi_{t}^{-1}(\bar{r}_{it},z_{it})}{\partial\bar{r}_{t}}\text{\ensuremath{\frac{\exp(p_{t}^{m}+m_{it})}{\exp\left(\bar{r}_{it}\right)}}}\label{eq:foc_m}
\end{equation}
holds for all firms.
\end{assumption}
Rearranging the first-order condition, we obtain the Hall-De Loecker-Warzynski
markup equation:
\begin{equation}
\frac{\partial\varphi_{t}^{-1}(\bar{r}_{it},z_{it})}{\partial\bar{r}_{t}}=\frac{\partial f_{t}(x_{it})/\partial m_{it}}{\exp(p_{t}^{m}+m_{it})/\exp\left(\bar{r}_{it}\right)}.\label{eq:DLW}
\end{equation}

We establish the following proposition. 
\begin{prop}
\label{P-step3} Suppose that Assumptions \ref{A-1}\textendash \ref{A-FOC}
hold. Then, we can identify $\varphi_{t}^{-1}(\cdot)$ and $f_{t}(\cdot)$
up to scale and location and each firm's markup $\partial\varphi_{t}^{-1}(\bar{r}_{it},z_{it})/\partial\bar{r}_{t}$
up to scale. 
\end{prop}
\begin{proof}
From (\ref{eq:phi_qt}) and (\ref{eq:foc_m}), the markup $\partial\varphi_{t}^{-1}(\bar{r}_{it},z_{it})/\partial\bar{r}_{t}$
is identified as
\begin{equation}
\frac{\partial\varphi_{t}^{-1}(\bar{r}_{it},z_{it})}{\partial\bar{r}_{t}}=\frac{\partial\mathbb{M}_{t}^{-1}\left(x_{it},z_{it}\right)}{\partial m_{t}}\left(\frac{\partial\phi_{t}(x_{it},z_{it})}{\partial m_{t}}-\frac{\exp(p_{t}^{m}+m_{it})}{\exp\left(\bar{r}_{it}\right)}\right)^{-1}.\label{eq:markup_identified}
\end{equation}
From $\bar{r}_{t}=\phi_{t}(x_{t},z_{t})$ and (\ref{eq:markup_identified}),
the markup is also identified as a function of $(x_{t},z_{t})$ as
\begin{align}
\mu_{t}(x_{t},z_{t}) & :=\frac{\partial\varphi_{t}^{-1}(\phi_{t}(x_{t},z_{t}),z_{t})}{\partial r_{t}}\nonumber \\
 & =\frac{\partial\mathbb{M}_{t}^{-1}\left(x_{t},z_{t}\right)}{\partial m_{t}}\left(\frac{\partial\phi_{t}(x_{t},z_{t})}{\partial m_{t}}-\frac{\exp(p_{t}^{m}+m_{t})}{\exp\left(\phi_{t}(x_{t},z_{t})\right)}\right)^{-1}\label{eq:markup_function}
\end{align}
Substituting (\ref{eq:markup_function}) into (\ref{eq:phi_qt}),
we identify $\partial f_{t}(x_{t})/\partial q_{t}$ for $q_{t}\in\{m_{t},k_{t},l_{t}\}$
as follows:
\begin{align}
\frac{\partial f_{t}(x_{t})}{\partial q_{t}} & =\mu_{t}(x_{t},z_{t})\frac{\partial\phi_{t}(x_{t},z_{t})}{\partial q_{t}}-\frac{\partial\mathbb{M}_{t}^{-1}\left(x_{t},z_{t}\right)}{\partial q_{t}}.\label{eq:output_elasticities_pro2}
\end{align}
Using $f_{t}(m_{t0}^{*},k_{t}^{*},l_{t}^{*})=0$ in Assumption \ref{A-2},
we identify $f_{t}(x_{t})$ by integration:
\begin{align}
f_{t}(m_{t},k_{t},l_{t}) & =\int_{m_{t0}^{*}}^{m_{t}}\frac{\partial f_{t}(s,k_{t},l_{t})}{\partial m_{t}}ds+\int_{k_{t}^{*}}^{k_{t}}\frac{\partial f_{t}(m_{t0}^{*},s,l_{t})}{\partial k_{t}}ds\nonumber \\
 & +\int_{l_{t}^{*}}^{l_{t}}\frac{\partial f_{t}(m_{t0}^{*},k_{t}^{*},s)}{\partial l_{t}}ds.\label{eq:f_integral}
\end{align}

Let $\mathcal{\bar{R}}:=\{\bar{r}_{t}:\bar{r}_{t}=\phi_{t}(x_{t},z_{t})\ \text{{for}\ \text{{some}\ }}(x_{t},z_{t})\in\mathcal{X}\times\mathcal{Z}\}$
be the support of $\bar{r}_{t}$. For given $(\bar{r}_{t},z_{t})\in\mathcal{\bar{R}}\times\mathcal{Z}$,
$X_{t}(\bar{r}_{t},z_{t}):=\{x_{t}\in\mathcal{X}:\phi_{t}(x_{t},z_{t})=\bar{r}_{t}\}$
is non-empty by the construction of $\mathcal{\bar{R}}$. Then, because
$f_{t}(x_{t})$ and $\mathbb{M}_{t}^{-1}(x_{t},z_{t})$ are identified,
the output quantity $\varphi_{t}^{-1}(\bar{r}_{t},z_{t})$ for any
$(\bar{r}_{t},z_{t})\in\mathcal{\bar{R}}\times\mathcal{Z}$ is identified
by
\[
\varphi_{t}^{-1}(\bar{r}_{t},z_{t})=f_{t}(x_{t})+\mathbb{M}_{t}^{-1}(x_{t},z_{t})\text{ for }x_{t}\in X_{t}(\bar{r}_{t},z_{t}).
\]
\end{proof}
The output price for individual firms is identified as
\begin{align*}
p_{it}: & =\bar{r}_{it}-\varphi_{t}^{-1}(\bar{r}_{it},z_{it}).
\end{align*}

\begin{cor}
Suppose that Assumptions \ref{A-1}\textendash \ref{A-FOC} hold.
Then, the production function, output quantities, output prices, and
TFP are identified up to scale and location; markups and output elasticities
are identified up to scale.
\end{cor}
\begin{rem}
Examination of the proofs reveals that we have over-identifying restrictions.
In particular, the proof of Proposition \ref{P-step1} goes through
with any choice of $q_{t-1}\in\{k_{t-1},l_{t-1},m_{t-1},z_{t-1}\}$
in (\ref{eq:dM_dq_identified}). Furthermore, the proof of Proposition
\ref{P-step3} does not rely on the restriction in (\ref{eq:phi_zt})
for identifying $\varphi_{t}^{-1}(\cdot)$. These over-identifying
restrictions can be useful in developing a specification test for
the model as well as for efficiently estimating the model.
\end{rem}

\subsubsection{Comparison to Existing Identification Approaches}

Our approach follows the spirits of existing identification approaches,
but it does differ from them in terms of implementations. First, step
2 distinguishes our approach from the standard control function approach
(e.g., \citealp{ackerberg2015identification}). In step 2, we identify
the control function from the dynamics of the inputs, and without
using any output measure. To clarify why this approach is necessary,
consider an alternative approach that uses an output measure. Specifically,
in the second step, we substitute $\omega_{it}=\varphi_{t}^{-1}(\bar{r}_{it},z_{it})-f_{t}(x_{it})$
into (\ref{eq:model}) and obtain the alternative transformation model:
\[
\varphi_{t}^{-1}(\bar{r}_{it},z_{it})=f_{t}(x_{it})+\tilde{h}_{t}(\bar{r}_{it-1},x_{it-1},z_{t-1})+\eta_{it}
\]
where $\tilde{h}_{t}(\bar{r}_{t-1},x_{t-1},z_{t-1}):=h(\varphi_{t-1}^{-1}(\bar{r}_{t-1},z_{t-1})-f_{t}(x_{t-1}))$.
Since this model also belongs to the class of transformation models
examined by \citet{chiappori2015nonparametric}, one might think that
we could have identified $\varphi_{t}(\cdot)$ and $f_{t}(\cdot)$
from the conditional distribution function $G_{\bar{r}_{t}|w_{t}}(\bar{r}_{t}|w_{t})$
of $\bar{r}_{t}$ given $w_{t}:=(x_{t},z_{t},\bar{r}_{t-1},x_{t-1},z_{t-1})$.
This is not possible, however, because once $(x_{t},z_{t})$ is conditioned
on, $\bar{r}_{t}=\phi_{t}(x_{t},z_{t})$ loses all variations. Therefore,
the derivatives of $G_{\bar{r}_{t}|w_{t}}$ with respect to past variables
$(\bar{r}_{t-1},x_{t-1},z_{t-1})$ are always 0, which violates the
condition corresponding to Assumption \ref{A-3} (f). 

Second, \citet{ackerberg2015identification} identify a structural
value-added function, $y_{it}=v_{t}(k_{it},l_{it})+\omega_{it}$,
which under perfect competition derives from a Leontief production
function $y_{it}=\min\left\{ v_{t}(k_{it},l_{it})+\omega_{it},a+m_{it}\right\} $.
However, the structural value-added function is difficult to employ
under imperfect competition because $y_{it}<v_{t}(k_{it},l_{it})+\omega_{it}$
can occur. Note that the maximum output capacity $y_{it}^{*}:=v_{t}(k_{it},l_{it})+\omega_{it}$
is determined before a firm chooses $m_{it}$ and $y_{it}$. Therefore,
if $y_{it}^{*}$ is large\textemdash due, for example, to a large
shock on $\omega_{it}$\textemdash then the profit maximizing output
$y_{it}$ can be lower than $y_{it}^{*}$.\footnote{As \citet{ackerberg2015identification} explains, under perfect competition,
if $y_{it}<y_{it}^{*}$, then the optimal output is 0 since the output
becomes linear in material. Since firms in a dataset have positive
outputs, $y_{it}=y_{it}^{*}$ holds for firms observed in a dataset.
However, under imperfect competition, it is possible to have $y_{it}<y_{it}^{*}$
and the optimal output is strictly positive. } Intuitively speaking, when increases in TFP double, a firm can preclude
a price drop by increasing its output by less than double. 

Third, our approach uses the first-order condition for material in
a way different from that seen in \citet{gandhi2020identification},
whose step identifies the material elasticity $\partial f_{t}(x_{t})/\partial m_{t}$
from the first-order condition (\ref{eq:foc_m}):
\[
\ln\text{\ensuremath{\frac{\exp(p_{t}^{m}+m_{it})}{\exp\left(r_{it}\right)}}}=\ln\frac{\partial f_{t}(x_{it})}{\partial m_{it}}-\ln\frac{\partial\varphi_{t}^{-1}(r_{it}-\varepsilon_{it},z_{it})}{\partial r_{t}}-\varepsilon_{it}
\]
under the assumption of perfect competition where $\ln\partial\varphi_{t}^{-1}(r_{it}-\varepsilon_{it},z_{it})/\partial r_{it}=0$
for all $i$. Under imperfect competition, when the markup depends
on revenue $r_{it}-\varepsilon_{it}$, $\partial f_{t}(x_{t})/\partial m_{t}$
cannot be identified solely from the first-order condition.

\subsection{Fixing Normalization across Periods\label{subsec:Normalization}}

Let $(\varphi_{t}^{-1}(\cdot),f_{t}(\cdot),\mathbb{M}_{t}^{-1}(\cdot))$
be a model structure for period $t$ identified by using Propositions
\ref{P-step1} and \ref{P-step3} under the normalization in Assumption
\ref{A-2}. Let $(\varphi_{t}^{*-1}(\cdot),f_{t}^{*}(\cdot),\mathbb{M}_{t}^{*-1}(\cdot))$
denote the true model structure. Since the structure is identified
up to scale and location normalization, there exist period-specific
location and scale parameters $(a_{1t},a_{2t},b_{t})\in\mathbb{R}^{2}\times\mathbb{R}_{+}$
such as 
\begin{align}
\varphi_{t}^{-1}(r_{t},z_{t}) & =a_{1t}+a_{2t}+b_{t}\varphi_{t}^{*-1}(r_{t},z_{t}),\ f_{t}(x)=a_{1t}+b_{t}f_{t}^{*}(x_{t}),\ \nonumber \\
\mathbb{M}_{t}^{-1}(x_{t},z_{t}) & =a_{2t}+b_{t}\mathbb{M}_{t}^{*-1}(x_{t},z_{t}).\label{eq:norm-t}
\end{align}
Generally speaking, the location and scale normalization differ across
periods\textemdash that is, $(a_{1t},a_{2t},b_{t})\neq(a_{1t+1},a_{2t+1},b_{t+1})$.
For the identified objects to be comparable across periods, we need
to fix normalization across periods by assuming that some object in
the model is time-invariant. The subsection discusses these additional
assumptions.\footnote{\citet{klette1996jae} and \citet{de2011product} identify the levels
of markups and output elasticities from revenue data by using a functional
form property of a demand function. They consider a constant elastic
demand function leading to $\varphi_{t}(y_{it},z_{it})=\alpha y_{it}-(\alpha-1)z_{it}$
where $z_{it}$ is an aggregate demand shifter, which is an weighted
average of revenue across firms, and $\alpha$ is an unknown parameter.
This formulation implies $\varphi_{t}^{-1}(r_{it},z_{it})=(1/\alpha)r_{it}+(1-1/\alpha)z_{it}$
and imposes a linear restriction $\partial\varphi_{t}^{-1}(r_{it},z_{it})/\partial r_{it}+\partial\varphi_{t}^{-1}(r_{it},z_{it})/\partial z_{it}=1$,
which fixes the scale parameter $b_{t}$.}

\subsubsection{Scale Normalization}

From (\ref{eq:norm-t}), the ratio of identified markups across two
periods relates to the ratio of true markups as 
\[
\frac{\partial\varphi_{t+1}^{-1}(r,z)/\partial r}{\partial\varphi_{t}^{-1}(r,z)/\partial r}=\frac{b_{t+1}}{b_{t}}\frac{\partial\varphi_{t+1}^{*-1}(r,z)/\partial r}{\partial\varphi_{t}^{*-1}(r,z)/\partial r}.
\]
Therefore, the ability to identify how true markups change over two
periods requires identification of the ratio of scale parameters,
$b_{t+1}/b_{t}$. Similarly, the ratio of identified output elasticities
across periods and that of identified TFP deviation from the mean
are related to their true values via the ratio of scale parameters:
\[
\frac{\partial f_{t+1}(x)/\partial q}{\partial f_{t}(x)/\partial q}=\frac{b_{t+1}}{b_{t}}\frac{\partial f_{t+1}^{*}(x)/\partial q}{\partial f_{t}^{*}(x)/\partial q}\text{ and }\frac{\omega_{it+1}-E\left[\omega_{it+1}\right]}{\omega_{it}-E\left[\omega_{it}\right]}=\frac{b_{t+1}}{b_{t}}\left(\frac{\omega_{it+1}^{*}-E\left[\omega_{it+1}^{*}\right]}{\omega_{it}^{*}-E\left[\omega_{it}^{*}\right]}\right)
\]
for $q\in\{m,k,l\}$. 

To identify $b_{t+1}/b_{t}$, we consider the following assumptions.
\begin{assumption}
\label{A-period} At least one of the following conditions (a)\textendash (c)
holds. (a) The unconditional variance of $\eta_{it}$ does not change
over time. (b) For some known interval $\mathcal{B}$ of $\mathcal{X}$,
the output elasticity of one of the inputs does not change over time
for all $x\in\mathcal{B}$. (c) For some known interval $\mathcal{B}$
of $\mathcal{X}$, the sum of output elasticities of the three inputs
does not change over time for all $x\in\mathcal{B}$.
\end{assumption}
Assumption \ref{A-period} (a) holds, for example, if the productivity
shock $\omega_{it}$ follows a stationary process because stationarity
requires that the distribution of $\eta_{it}$ does not change over
time. Assumption \ref{A-period} (b) assumes that the elasticity of
output with respect to one input does not change over time for some
known interval; meanwhile, under Assumption \ref{A-period} (c), returns
to scale in production technology does not change for some known interval
of inputs. 
\begin{prop}
\label{P-3} Suppose that Assumptions \ref{A-1}\textendash \ref{A-period}
hold for time $t$ and $t+1$. Then, we can identify the ratio of
markups between two periods $t$ and $t+1$, the ratio of output elasticities
between $t$ and $t+1$, and the ratio of TFP deviation from the mean
between $t$ and $t+1$. 
\end{prop}
\begin{proof}
Suppose that Assumption \ref{A-period}(a) holds. Let $\text{var}(\eta_{t})$
and $\text{var}(\eta_{t+1})$ be the variance of $\eta_{t}$ and $\eta_{t+1}$
identified under the period-specific normalization in Assumption \ref{A-2}
for $t$ and $t+1$, respectively. From (\ref{eq:model}) and (\ref{eq:norm-t}),
$\text{var}(\eta_{t})=b_{t}^{2}\text{var}(\eta_{t}^{*})$ and $\text{var}(\eta_{t+1})=b_{t+1}^{2}\text{var}(\eta_{t+1}^{*})$.
From $\text{var}(\eta_{t}^{*})=\text{var}(\eta_{t+1}^{*})$, $b_{t+1}/b_{t}$
is identified as $b_{t+1}/b_{t}=\sqrt{{\text{var}(\eta_{t})}/{\text{var}(\eta_{t+1})}}$.

Let $\partial f_{t}(x)/\partial q$ and $\partial f_{t+1}(x)/\partial q$
be those elasticities identified under the period-specific normalization
in Assumption \ref{A-2} for $t$ and $t+1$, respectively, and $\partial f_{t}^{*}(x)/\partial q$
and $\partial f_{t+1}^{*}(x)/\partial q$ be the true elasticities.
From (\ref{eq:norm-t}), $\partial f_{t}(x)/\partial q=b_{t}\partial f_{t}^{*}(x)/\partial q$
and $\partial f_{t+1}(x)/\partial q=b_{t+1}\partial f_{t+1}^{*}(x)/\partial q$
hold.

Suppose that Assumption \ref{A-period}(b) holds. Then, $\partial f_{t}^{*}(x)/\partial q=\partial f_{t+1}^{*}(x)/\partial q$
for some input $q\in\{m,k,l\}$ and $x\in\mathcal{B}$. Then, $b_{t+1}/b_{t}$
is identified as $b_{t+1}/b_{t}=(\partial f_{t+1}(x)/\partial q)/(\partial f_{t}(x)/\partial m)$
for $x\in\mathcal{B}$. 

Suppose that Assumption \ref{A-period}(c) holds, implying
\[
1=\frac{\partial f_{t+1}^{*}(x)/\partial m+\partial f_{t+1}^{*}(x)/\partial k+\partial f_{t+1}^{*}(x)/\partial l}{\partial f_{t}^{*}(x)/\partial m+\partial f_{t}^{*}(x)/\partial k+\partial f_{t}^{*}(x)/\partial l}\,\text{for }\ensuremath{x\in\mathcal{B}}.
\]
Then, $b_{t+1}/b_{t}$ is identified as 
\[
\frac{b_{t+1}}{b_{t}}=\frac{\partial f_{t+1}(x)/\partial m+\partial f_{t+1}(x)/\partial k+\partial f_{t+1}(x)/\partial l}{\partial f_{t}(x)/\partial m+\partial f_{t}(x)/\partial k+\partial f_{t}(x)/\partial l}\quad\text{for \ensuremath{x\in\mathcal{B}}}.
\]
\end{proof}

\subsubsection{Local Constant Returns to Scale\label{subsec:LCRS}}

We consider the following local constant returns to scale that strengthens
Assumption \ref{A-period} (c).
\begin{assumption}
\label{A-LCRS} (Local Constant Returns to Scale) For some known interval
$\mathcal{B}$ of $\mathcal{X}$, the sum of the output elasticities
of the three inputs equals to 1 for all $x\in\mathcal{B}$. 
\end{assumption}
Assumption \ref{A-LCRS} is stronger than Assumption \ref{A-period}(c),
but it is weaker than the assumptions used in some other studies on
markups. Markup is sometimes estimated as the ratio of revenue $\exp(r_{it})$
to total costs $TC_{it}$ under the assumption that a cost function
is linear in output $TC_{it}=MC_{it}y_{it}$ with constant marginal
costs $MC_{it}$. The linear cost function requires the following
assumptions that are stronger than Assumption \ref{A-LCRS}: (1) constant
returns to scale \emph{globally} holds for all $x\in\mathcal{B}$;
(2) all three inputs are flexible and (3) a firm is a price taker
of all three inputs. Under Assumption \ref{A-LCRS}, marginal costs
may increase in output, especially in the short run, when dynamic
inputs such as capital require adjustment costs.

With Assumption \ref{A-LCRS}, the scale normalization parameter $b_{t}$
can be identified for all periods as follows. Let $f_{t}(x)$ be the
identified production function and $f_{t}^{*}(x)$ be the true one
where $f_{t}(x_{t})=a_{t}+b_{t}f_{t}^{*}(x_{t})$ from (\ref{eq:norm-t}).
For $x\in\mathcal{B}$, we have
\[
b_{t}=b_{t}\left(\frac{\partial f_{t}^{*}(x)}{\partial m}+\frac{\partial f_{t}^{*}(x)}{\partial k}+\frac{\partial f_{t}^{*}(x)}{\partial l}\right)=\frac{\partial f_{t}(x)}{\partial m}+\frac{\partial f_{t}(x)}{\partial k}+\frac{\partial f_{t}(x)}{\partial l}.
\]
Given that we have identified the scale parameter $b_{t}$ in (\ref{eq:norm-t}),
we have established the following proposition.
\begin{prop}
\label{P-CRS} Suppose that Assumptions \ref{A-1}\textendash \ref{A-FOC}
and \ref{A-LCRS} hold. Then, $\varphi_{t}(\cdot)$ and $f_{t}(\cdot)$
can be identified up to location. The levels of markup and output
elasticities can be identified. Output quantity, output price, and
TFP can be identified up to location.
\end{prop}

\subsubsection{Location Normalization}

Suppose that scale normalization $b_{t}$ is already identified\textemdash for
example, from Proposition \ref{P-CRS}. Define
\begin{align}
\tilde{\varphi}_{t}^{-1}(r_{t},z_{t}) & :=\varphi_{t}^{-1}(r_{t},z_{t})/b_{t},\ \tilde{f}_{t}(x):=f_{t}(x)/b_{t},\ \tilde{\omega}_{t}:=\omega_{t}/b_{t},\nonumber \\
\ \tilde{a}_{1t} & :=a_{1t}/b_{t},\text{\ {and}}\ \tilde{a}_{2t}:=a_{2t}/b_{t}.\label{eq:norm-tilde}
\end{align}
Then, (\ref{eq:norm-t}) is written as
\begin{align}
\tilde{\varphi}_{t}^{-1}(r_{t},z_{t}) & =\tilde{a}_{1t}+\tilde{a}_{2t}+\varphi_{t}^{*-1}(r_{t},z_{t}),\ \tilde{f}_{t}(x)=\tilde{a}_{1t}+f_{t}^{*}(x_{t}),\ \tilde{\omega}_{t}=\tilde{a}_{2t}+\omega_{t}^{*}.\label{eq:norm-t2}
\end{align}
From (\ref{eq:norm-t}), the growth rates (log differences) of the
identified output and TFP between $t$ and $t+1$ are related to their
true values as follows: 
\begin{align}
\tilde{\varphi}_{t+1}^{-1}(\bar{r}_{it+1},z_{it+1})-\tilde{\varphi}_{t}^{-1}(\bar{r}_{it},z_{it}) & =\tilde{a}_{1t+1}+\tilde{a}_{2t+1}-\tilde{a}_{1t}-\tilde{a}_{2t}+\varphi_{t+1}^{*-1}(\bar{r}_{it+1},z_{it+1})-\varphi_{t}^{*-1}(\bar{r}_{it},z_{it}),\nonumber \\
\tilde{f}_{t+1}(x_{t+1})-\tilde{f}_{t}(x_{t}) & =\tilde{a}_{1t+1}-\tilde{a}_{1t}+f^{*}(x_{t+1})-f_{t}^{*}(x_{t}),\nonumber \\
\tilde{\omega}_{it+1}-\tilde{\omega}_{it} & =\tilde{a}_{2t+1}-\tilde{a}_{2t}+\omega_{it+1}^{*}-\omega_{it}^{*}.\label{eq:growth}
\end{align}
Therefore, to identify the growth rates of output and TFP, we need
to identify the changes in the location parameters. To do so, we can
use an industry-level producer price index $P_{t}^{*}$, which is
often available as data, to identify the change in the location parameters.
Suppose that $P_{t}^{*}$ is a Laspeyres index
\begin{equation}
P_{t}^{*}:=\frac{\sum_{i\in\tilde{N}}\exp(p_{it}^{*}+y_{i0}^{*})}{\sum_{i\in\tilde{N}}\exp(p_{i0}^{*}+y_{i0}^{*})},\label{eq:PPI}
\end{equation}
where $\tilde{N}$ is a known set (or a random sample) of products.
$p_{i0}^{*}$ and $y_{i0}^{*}$ are firm $i$'s log true price and
log true output at the base period, respectively. The following argument
holds for forms of a price index (other than Laspeyres) as long as
the price index is a known function of prices that is homogenous of
degree 1; this condition is usually satisfied.
\begin{assumption}
\label{A-location}(a) The industry-level producer price index $P_{t}^{*}$
is known as data. (b) For some known point $\bar{x}\in\mathcal{X}$,
the true production functions of $t$ and $t+1$, $f_{t}^{*}(\cdot)$
and $f_{t+1}^{*}(\cdot)$, satisfy $f_{t}^{*}(\bar{x})=f_{t+1}^{*}(\bar{x})$. 
\end{assumption}
Assumption \ref{A-location}(b) is innocuous, implying that any output
change between $t$ and $t+1$ when inputs are fixed at $\bar{x}$
is attributed to a TFP change. 

Using the aggregate price index, we can identify the change in the
location parameters and identify the growth of TFP and output.
\begin{prop}
\label{P-location} Suppose Assumptions \ref{A-1}\textendash \ref{A-FOC},
\ref{A-LCRS}, and \ref{A-location} hold. Then, the true growth rate
of output $\varphi_{t+1}^{*-1}(\bar{r}_{it+1},z_{it+1})-\varphi_{t}^{*-1}(\bar{r}_{it},z_{it})$
and that of TFP $\omega_{it+1}^{*}-\omega_{it}^{*}$ can be identified
for each firm.
\end{prop}
\begin{proof}
Let $\tilde{p}_{it}:=\bar{r}_{it}-\tilde{\varphi}_{t}^{-1}(\bar{r}_{it},z_{it})$
and $\tilde{y}_{it}:=\tilde{\varphi}_{t}^{-1}(\bar{r}_{it},z_{it})$
be an output price and an output quantity identified under the normalization
in (\ref{eq:norm-tilde}) and Assumption \ref{A-2}, respectively.
Using these, we calculate an industry-level producer price index with
them:
\[
P_{t}:=\frac{\sum_{i\in\tilde{N}}\exp(\tilde{p}_{it}+\tilde{y}_{i0})}{\sum_{i\in\tilde{N}}\exp(\tilde{p}_{i0}+\tilde{y}_{i0})}.
\]
 From (\ref{eq:norm-t2}) and (\ref{eq:PPI}), $P_{t}$ is written
as
\begin{align*}
P_{t} & =\frac{\sum_{i\in\tilde{N}}\exp(-(\tilde{a}_{1t}+\tilde{a}_{2t})+p_{it}^{*}+\tilde{a}_{1,0}+\tilde{a}_{2,0}+y_{i0}^{*})}{\sum_{i\in\tilde{N}}\exp(p_{i0}^{*}+y_{i0}^{*})}\\
 & =\exp(\tilde{a}_{1,0}+\tilde{a}_{2,0}-(\tilde{a}_{1t}+\tilde{a}_{2t}))P_{t}^{*}.
\end{align*}
Therefore, $\tilde{a}_{1t+1}+\tilde{a}{}_{2t+1}-\tilde{a}_{1t}-\tilde{a}{}_{2t}$
is identified as:
\begin{equation}
\tilde{a}_{1t+1}+\tilde{a}{}_{2t+1}-\tilde{a}_{1t}-\tilde{a}{}_{2t}=\ln P_{t+1}^{*}-\ln P_{t+1}-\left(\ln P_{t}^{*}-\ln P_{t}\right)\label{eq:a_diff}
\end{equation}
From (\ref{eq:growth}), we identify the output growth rate $\varphi_{t+1}^{*-1}(\bar{r}_{it+1},z_{it+1})-\varphi_{t}^{*-1}(\bar{r}_{it},z_{it})$.

Evaluating the second equation in (\ref{eq:growth}) at $x_{t+1}=x_{t}=\bar{x}$
in Assumption \ref{A-location}(b), we identify $\tilde{a}_{1t+1}-\tilde{a}_{1t}$
as: 
\begin{align*}
\tilde{a}_{1t+1}-\tilde{a}_{1t} & =\tilde{a}_{1,t+1}+f_{t+1}^{*}(\bar{x})-\left(\tilde{a}_{1,t}+f_{t}^{*}(\bar{x})\right)\\
 & =\tilde{f}_{t+1}(\bar{x})-\tilde{f}_{t}(\bar{x}).
\end{align*}
From (\ref{eq:a_diff}), $\tilde{a}_{2t+1}-\tilde{a}_{2t}$ is also
identified as 
\[
\tilde{a}_{2t+1}-\tilde{a}_{2t}=\ln P_{t+1}^{*}-\ln P_{t+1}-\left(\ln P_{t}^{*}-\ln P_{t}\right)-\left(\tilde{f}_{t+1}(\bar{x})-\tilde{f}_{t}(\bar{x})\right).
\]
Therefore, from (\ref{eq:growth}), the true TFP growth rate $\omega_{it+1}^{*}-\omega_{it}^{*}$
is also identified.
\end{proof}

\subsection{Identification of Demand System and Utility Function \label{subsec:demand}}

Given that we have identified each firm's output price and quantity,
it is possible to identify with additional assumptions a system of
demand functions and a homothetic utility function of a representative
consumer. The identified demand system and the identified utility
function can be used to undertake counterfactual analysis and welfare
analysis. 

We consider an HSA system \citep{matsuyama2017beyond}, which can
be expressed as a system of direct demand functions or of inverse
demand functions. The two systems are self-dual in the sense that
either can be derived from the other. We consider a system of inverse
demand functions. Let $P_{it}:=\exp(p_{it})$ and $Y_{it}:=\exp(y_{it})$
be the levels of price and quantity of firm $i$'s output at time
$t$, respectively. Let $N_{t}$ be the set of firms in the industry
and $\Phi_{t}:=\sum_{i\in N_{t}}P_{it}Y_{it}$ be the industry expenditure.
The inverse demand function for product $i$ is given by 
\[
P_{it}=\frac{\Phi_{t}}{Y_{it}}S_{t}\left(\frac{Y_{it}}{A_{t}\left(\mathbf{Y}_{t},\mathbf{z}_{t}\right)},z_{it}\right).
\]
where $S_{t}(\cdot,z_{it})$ provides the budget share of product
$i$, $\mathbf{Y}_{t}:=(Y_{1t},...,Y_{Nt})\in\mathscr{\bar{Y}}:=\exp\left(\mathcal{Y}\right)^{N}$
is a vector of consumption, $\mathbf{z}_{t}:=(z_{1t},...,z_{Nt})$
is a vector of observable demand shifters and $A_{t}(\mathbf{Y}_{t},\mathbf{z}_{t})$
is the aggregate quantity index summarizing interactions across products.\footnote{If the utility function is CES $U_{t}(\mathbf{Y}_{t},\mathbf{z}_{t})=\left[\sum_{i=1}^{N}Y_{it}^{\rho(z_{it})}\right]^{1/\rho(z_{it})}$,
then the inverse demand function is given by $P_{it}=\frac{\Phi_{t}}{Y_{it}}\left(\frac{Y_{it}}{U_{t}(\mathbf{Y}_{t},\mathbf{z}_{t})}\right)^{\rho(z_{it})}$.
In this case, the quantity index is the same as the utility function,
but they are generally different.} Since $S_{t}(\cdot)$ is nonparametric, the HSA system can nest various
demand functions used in the literature such as the constant elastic
demand from the CES utility, the symmetric translog demand (\citealp{feenstra2003homothetic};
\citealp{feenstra2017globalization}), or the constant response demand
(Mrázová and Neary, \citeyear{mrazova2017not,MRAZOVA2019102561}).\footnote{A HSA version of the constant response demand (Mrázová and Neary,
\citeyear{mrazova2017not,MRAZOVA2019102561}) can be formulated as
for example, $P_{it}=\frac{\beta\Phi_{t}}{Y_{it}}\left[\left(\frac{Y_{it}}{A_{t}\left(\mathbf{Y}_{t},\mathbf{z}_{t}\right)}\right)^{\alpha}+\gamma z_{it}\right]^{\delta}$
where firm $i$'s markup is given by $\mu_{it}=\frac{1}{\alpha\beta}+\frac{\gamma z_{it}}{\alpha\beta\left(Y_{it}\right)^{\alpha}}$.
See \citet{matsuyama2017beyond} regarding how the HSA nests the translog
demand.}

For identification of a demand system, we make assumptions regarding
the market structure.
\begin{assumption}
\label{A-demand} The good market is monopolistically competitive
(without free entry)\textemdash that is, each firm takes the quantity
index $A_{t}(\mathbf{Y}_{t},\mathbf{z}_{t})$ as given.
\end{assumption}
The assumption of monopolistic competition follows \citet{klette1996jae}
and \citet{de2011product}, with the inverse demand function becoming
a symmetric function of the firm's own output, as in (\ref{eq:inverse_demand}). 

The demand elasticity equals $(\mu-1)/\mu$ when $\mu$ is markup.
If the markup is identified up to scale, then the demand elasticity
is not uniquely identified. Therefore, we need to fix the scale normalization
to identify the demand function.
\begin{assumption}
\label{A-demand-2} $\varphi_{t}^{-1}(\bar{r}_{t},z_{t})$ is identified
up to location.
\end{assumption}
Assumption \ref{A-demand-2} is satisfied when Proposition \ref{P-CRS}
holds. 

An HSA demand system can be constructed as follows. Suppose $\varphi_{t}^{-1}(\bar{r}_{t},z_{t})$
is identified from Proposition \ref{P-CRS}; taking its inverse function
obtains the revenue function $\varphi_{t}(y_{t},z_{t})$. Fixing a
realized data point of $\mathbf{Y}_{t}^{0}:=(Y_{1t}^{0},...,Y_{Nt}^{0})\in\mathscr{\mathscr{\bar{Y}}}$
and $\mathbf{z}_{t}^{0}:=(z_{1t}^{0},...,z_{Nt}^{0})\in\mathcal{Z}^{N}$,
we let $\Phi_{t}:=\sum_{i\in N_{t}}\exp\left(\varphi_{t}\left(\ln Y_{it}^{0},z_{it}^{0}\right)\right)$
be the consumer's budget, which is taken as given. For given $(\mathbf{Y}_{t},\mathbf{z}_{t})\in\mathscr{\bar{Y}}\times\mathcal{Z}^{N}$,
we define a vector of market shares $S_{t}(\mathbf{Y}_{t},\mathbf{z}_{t}):=\left(S_{t}(Y_{1t},z_{1t}),....,S_{t}(Y_{Nt},z_{Nt})\right)$
such that
\[
S_{t}(Y_{it},z_{it}):=\frac{\exp\left(\varphi_{t}\left(\ln Y_{it},z_{it}\right)\right)}{\Phi_{t}}.
\]
The quantity index $A_{t}(\mathbf{Y}_{t},\mathbf{z}_{t})$ is identified
as follows. First, since $\sum_{i\in N_{t}}S_{t}\left(Y_{it}^{0},z_{it}^{0}\right)=1$,
by construction, $A_{t}(\mathbf{Y}_{t}^{0},\mathbf{z}_{t}^{0})=1$
holds for the data point $(\mathbf{Y}_{t}^{0},\mathbf{z}_{t}^{0})$.
For other values $(\mathbf{Y}_{t},\mathbf{z}_{t})\in\mathscr{\bar{Y}}\times\mathcal{Z}^{N}$,
we can obtain $A_{t}(\mathbf{Y}_{t},\mathbf{z}_{t})$ by solving
\[
\sum_{i\in N_{t}}S_{t}\left(\frac{Y_{it}}{A_{t}(\mathbf{Y}{}_{t},\mathbf{z}{}_{t})},z{}_{it}\right)=1.
\]
Since $S_{t}(\cdot,z_{it})$ is continuous and strictly increasing,
$A_{t}(\mathbf{Y}_{t},\mathbf{z}_{t})$ is uniquely determined. Then,
we obtain the inverse demand function for all $(\mathbf{Y}_{t},\mathbf{z}_{t})\in\mathscr{\bar{Y}}\times\mathcal{Z}^{N}$:
\begin{equation}
P_{it}=\frac{\Phi_{t}}{Y_{it}}S_{t}\left(\frac{Y_{it}}{A_{t}\left(\mathbf{Y}_{t},\mathbf{z}_{t}\right)},z_{it}\right).\label{eq:HSA_system}
\end{equation}
Applying the result of \citeauthor{matsuyama2017beyond} (2017, Proposition
1 and Remark 3), the following proposition establishes that the HSA
demand system (\ref{eq:HSA_system}) constructed above can be derived
from a unique consumer preference, and that it is possible to identify
an associated utility function. Appendix \ref{subsec:Demand} supplies
the proof.
\begin{prop}
\label{P-demand} Suppose Assumption \ref{A-demand-2} holds. (a)
There exists a unique monotone, convex, and homothetic rational preference
$\succsim$ over $\mathscr{\mathscr{\bar{Y}}}$ that generates an
HSA demand system (\ref{eq:HSA_system}). (b) This preference $\succsim$
is represented by a homothetic utility function defined by
\[
\ln U_{t}(\mathbf{Y}_{t},\mathbf{z}_{t})=\ln A_{t}(\mathbf{Y}_{t},\mathbf{z}_{t})+\sum_{i=1}^{N}\int_{c_{i}(\mathbf{z}_{t})}^{Y_{it}/A_{t}(\mathbf{Y}_{t},\mathbf{z}_{t})}\frac{S_{t}\left(\xi,z_{it}\right)}{\xi}d\xi,
\]
where $c(\mathbf{z}_{t}):=(c_{1}(\mathbf{z}_{t}),...,c_{N}(\mathbf{z}_{t}))$
is defined by $U_{t}(c(\mathbf{z}_{t}),\mathbf{z}_{t})=1$. (c) The
identified demand system $S_{t}(\cdot)$ and preference $\succsim$
do not depend on the location normalization of $\varphi_{t}^{-1}(\bar{r}_{t},z_{t})$.
\end{prop}

\subsection{Identification in Alternative Settings\label{subsec:Alternative-Settings}}

\subsubsection{Endogenous Labor Input\label{subsec:Endogenous-Labor-Input}}

Identification is possible when $l_{t}$ correlates with $\eta_{t}$.
In the spirits of \citet{ackerberg2015identification} and the dynamic
generalized method of moment approach (e.g., \citealp{arellano1991some,arellano1995another,blundell1998initial,blundell2000gmm}),
we provide identification using lagged labor $l_{t-1}$ as an instrument
for $l_{t}$. Specifically, we follow the approach of \citet{ackerberg2015identification},
which assumes (1) $l_{t}$ correlates with $l_{t-1}$ and (2) the
firm's profit maximization problem regarding $m_{t}$ conditional
on $l_{t}$ is expressed by (\ref{eq:profit_maximization}), which
allows the material demand to be written as $m_{it}=\mathbb{M}_{t}(\omega_{it},k_{it},l_{it},z_{it})$.
This approach has the advantage of being consistent with various data
generating processes regarding the choice of $l_{t}$.\footnote{See \citet{ackerberg2015identification} for examples of such data-generating
processes. For example, $l_{t}$ can be chosen at time $t$ with adjustment
costs; a firm can face an auto-correlated firm-specific wage; or $l_{t}$
can be chosen at time $t-1$ or at an intermediate time between $t$
and $t-1$. } 

Identifying $\mathbb{M}_{t}^{-1}(x_{t},z_{t})$ using $l_{t-1}$ as
an instrument for $l_{t}$ is nontrivial because the model (\ref{eq:model})
includes $l_{t-1}$ in $\bar{h}_{t}\left(x_{t-1},z_{t-1}\right)$.
It is not possible to use the variation of $l_{t-1}$ simultaneously
for two purposes (i.e., identifying $\bar{h}_{t}\left(x_{t-1},z_{t-1}\right)$
and instrumenting $l_{t}$). Therefore, we proceed to identification
in two steps. We first identify $\bar{h}_{t}\left(x_{t-1},z_{t-1}\right)$
(up to location) and then use $l_{t-1}$ to identify $\mathbb{M}_{t}^{-1}(m_{t},k_{t},l_{t},z_{t})$.

\paragraph{Identification of $\bar{h}_{t}\left(x_{t-1},z_{t-1}\right)$.}
\begin{assumption}
\label{assu:labor} (i) Assumptions \ref{A-3} (a), (d), (e), and
(f) hold. (ii) $\eta_{t}$ is independent of $\tilde{w}_{t}:=(k_{t},z_{t},x_{t-1},z_{t-1})'\in\mathcal{\tilde{W}}:=\mathcal{K}\times\mathcal{Z}\times\mathcal{X}\times\mathcal{Z}$
with $E[\eta_{t}|\tilde{w}_{t}]=0$. $\tilde{w}_{t}$ is continuously
distributed on $\mathcal{\tilde{W}}$. (iii) For each $(x_{t-1},z_{t-1})\in\mathcal{X}\times\mathcal{Z}$,
$\mathscr{A}_{m_{t}}\left(x_{t-1},z_{t-1}\right)=\{(\tilde{x}_{t},\tilde{z}_{t})\in\mathcal{X}\times\mathcal{Z}|\partial G_{m_{t}|v_{t}}(\tilde{m}_{t}|\tilde{k}_{t},\tilde{l}_{t},\tilde{z}_{t},x_{t-1},z_{t-1})/\partial m_{t}>0\}$
is non-empty.
\end{assumption}
Assumptions \ref{assu:labor} (i) and (ii) simply modify Assumption
\ref{A-3} such that $l_{t}$ may correlate with $\eta_{t}$. Assumption
\ref{assu:labor} (iii) is innocuous because it is satisfied if the
firm's survival probability at time $t$ conditional on $(x_{t-1},z_{t-1})$
is not 0. 

The conditional distribution of $m_{t}$ given $v_{t}$ satisfies

\[
G_{m_{t}|v_{t}}(m_{t}|v_{t})=G_{\eta_{t}|l_{t}}\left(\mathbb{M}_{t}^{-1}(m_{t},k_{t},l_{t},z_{t})-\bar{h}_{t}\left(x_{t-1},z_{t-1}\right)|l_{t}\right).
\]
Taking the derivatives of both sides with respect to $q_{t}\in\{m_{t},k_{t},z_{t}\}$
and $q_{t-1}\in\{k_{t-1},l_{t-1},m_{t-1},z_{t-1}\}$ and their ratios,
we identify $\partial\mathbb{M}_{t}^{-1}(m,k_{t},l_{t},z_{t})/\partial q_{t}$
and $\partial\bar{h}(x_{t-1},z_{t-1})/\partial q_{t}$ as follows:
\begin{align}
\frac{\partial\mathbb{M}_{t}^{-1}(m_{t},k_{t},l_{t},z_{t})}{\partial q_{t}} & =-\frac{\partial\bar{h}\left(\tilde{x}_{t-1},\tilde{z}_{t-1}\right)}{\partial q_{t-1}}\frac{\partial G_{m_{t}|v_{t}}\left(m_{t}|k_{t},l_{t},z_{t},\tilde{x}_{t-1},\tilde{z}_{t-1}\right)/\partial q_{t}}{\partial G_{m_{t}|v_{t}}\left(m_{t}|k_{t},l_{t},z_{t},\tilde{x}_{t-1},\tilde{z}_{t-1}\right)/\partial q_{t-1}},\label{eq:identified_M}\\
\frac{\partial\bar{h}\left(x_{t-1},z_{t-1}\right)}{\partial q_{t-1}} & =-\frac{\partial\mathbb{M}_{t}^{-1}(\tilde{m}_{t},\tilde{k}_{t},\tilde{l}_{t},\tilde{z}_{t})}{\partial m_{t}}\frac{\partial G_{m_{t}|v_{t}}\left(\tilde{m}_{t}|\tilde{k}_{t},\tilde{l}_{t},\tilde{z}_{t},x_{t-1},z_{t-1}\right)/\partial q_{t-1}}{\partial G_{m_{t}|v_{t}}\left(\tilde{m}_{t}|\tilde{k}_{t},\tilde{l}_{t},\tilde{z}_{t},x_{t-1},z_{t-1}\right)/\partial m_{t}},\label{eq:identified_h}
\end{align}
where $\left(\tilde{x}_{t-1},\tilde{z}_{t-1}\right)\in\mathcal{A}_{q_{t-1}}$
and $\left(\tilde{x}_{t},\tilde{z}_{t}\right)\in\mathcal{A}_{m_{t}}(x_{t-1},z_{t-1})$.
Note that (\ref{eq:identified_M}) is the same as in (\ref{eq:dM/dq}).
Thus, following the same steps as those in the proof for Proposition
\ref{P-step1}, we identify $\partial\mathbb{M}_{t}^{-1}(m,k_{t},l_{t},z_{t})/\partial q_{t}$
up to scale, and then $\partial\bar{h}\left(x_{t-1},z_{t-1}\right)/\partial q_{t-1}$
up to scale from (\ref{eq:identified_h}). 

Define $d_{l}\left(l_{t}\right):=\mathbb{M}_{t}^{-1}(m_{t0}^{*},k_{t}^{*},l_{t},z_{t}^{*})$
for $(m_{t0}^{*},k_{t}^{*},z_{t}^{*})$ in (\ref{eq:norm}) and $d:=\bar{h}_{t}\left(x_{t-1}^{*},z_{t-1}^{*}\right)$
for some point $(x_{t-1}^{*},z_{t-1}^{*})\in\mathcal{X}\times\mathcal{Z}$.
Integrating the identified elasticities in (\ref{eq:identified_M})
and (\ref{eq:identified_h}), we obtain 
\begin{align}
\mathbb{M}_{t}^{-1}(m_{t},k_{t},l_{t},z_{t}) & =d_{l}(l_{t})+\Lambda_{lt}\left(x_{t},z_{t}\right),\label{eq:identified1}\\
\bar{h}\left(x_{t-1},z_{t-1}\right) & =d+\Lambda_{ht}\left(x_{t-1},z_{t-1}\right),\label{eq:identified2}
\end{align}
where function $d_{l}(l_{t})$ and constant $d$ are unknown objects
to be identified; $\Lambda_{lt}\left(x_{t},z_{t}\right)$ and $\Lambda_{ht}\left(x_{t-1},z_{t-1}\right)$
are identified and thus treated as known functions.\footnote{Specifically, $\Lambda_{lt}\left(x_{t},z_{t}\right)$ and $\Lambda_{ht}\left(x_{t-1},z_{t-1}\right)$
are given by 
\begin{align*}
\Lambda_{lt}\left(x_{t},z_{t}\right) & :=\int_{m_{t0}^{*}}^{m_{t}}\frac{\partial\mathbb{M}_{t}^{-1}(s,k_{t},l_{t},z_{t})}{\partial m_{t}}ds+\int_{k_{t}^{*}}^{k_{t}}\frac{\partial\mathbb{M}_{t}^{-1}(m_{t0}^{*},s,l_{t},z_{t})}{\partial k_{t}}ds+\int_{z_{t}^{*}}^{z_{t}}\frac{\partial\mathbb{M}_{t}^{-1}(m_{t0}^{*},k_{t}^{*},l_{t},s)}{\partial z_{t}}ds\\
\Lambda_{ht}\left(x_{t-1},z_{t-1}\right) & :=\int_{m_{t-1}^{*}}^{m_{t-1}}\frac{\partial\bar{h}_{t}(s,k_{t-1},l_{t-1},z_{t-1})}{\partial m_{t-1}}ds+\int_{k_{t-1}^{*}}^{k_{t-1}}\frac{\partial\bar{h}_{t}(m_{t-1}^{*},s,l_{t-1},z_{t-1})}{\partial k_{t-1}}ds\\
 & +\int_{l_{t-1}^{*}}^{l_{t-1}}\frac{\partial\bar{h}_{t}(m_{t-1}^{*},k_{t-1}^{*},s,z_{t-1})}{\partial l_{t-1}}ds+\int_{z_{t-1}^{*}}^{z_{t-1}}\frac{\partial\bar{h}_{t}(m_{t-1}^{*},k_{t-1}^{*},l_{t-1}^{*},s)}{\partial z_{t-1}}ds
\end{align*}
 } 

\paragraph{Identification of $\mathbb{M}_{t}^{-1}(m_{t},k_{t},l_{t},z_{t})$.}

Defining $H_{it}:=\Lambda_{lt}\left(x_{it},z_{it}\right)-\Lambda_{ht}\left(x_{it-1},z_{it-1}\right)$
as a known variable, we rewrite model (\ref{eq:model}) as
\[
H_{it}=d-d_{l}(l_{it})+\eta_{it}.
\]
From $l_{t-1}\bot\eta_{t}$, we obtain the following moment condition
for nonparametric instrument variable (IV) identification: 
\begin{equation}
E\left[H_{it}-d+d_{l}(l_{it})|l_{it-1}\right]=0.\label{eq:moment}
\end{equation}
For instance, if $f_{t}$ is Cobb-Douglas as in (\ref{eq:AR1}), then
$d_{l}(l_{t})=-\theta_{l}\left(l_{t}-l_{t}^{*}\right)$ from (\ref{eq:omega}),
and the moment condition (\ref{eq:moment}) becomes that for linear
IV regression:
\[
E\left[H_{it}-d-\theta_{l}(l_{it}-l_{it}^{*})|l_{it-1}\right]=0.
\]
A standard procedure of linear IV regression identifies $(d,\theta_{l})$
if $l_{it}$ sufficiently correlates with $l_{it-1}$. 

Following the literature on nonparametric IV (e.g, \citealp{newey2003instrumental}),
we assume that $l_{t-1}$ satisfies the following completeness condition. 
\begin{assumption}
\label{assu:completeness} For all functions $\delta(l_{t}):\mathcal{L}\rightarrow\mathbb{R}$
such that $E[\delta(l_{t})|l_{t-1}]<\infty$, $E[\delta(l_{t})|l_{t-1}]=0$
a.s. implies $\delta(l_{t})=0$ a.s..
\end{assumption}
With Assumption \ref{assu:completeness}, the moment condition (\ref{eq:moment})
uniquely identifies $\{d,d_{l}(l_{t})\}$.\footnote{The proof is as follows. Suppose $\{\tilde{d},\tilde{d}_{l}(l_{it})\}$
also satisfies the moment condition (\ref{eq:moment}). Then, it holds
that $E\left[\tilde{d}-d+\tilde{d}_{l}(l_{it})-d_{l}(l_{it})|l_{it-1}\right]=0$
a.s. The completeness condition implies $\tilde{d}-d+\tilde{d}_{l}(l_{it})-d_{l}(l_{it})=0$
a.s. Since $\tilde{d}_{l}(l_{t}^{*})=d_{t}(l_{t}^{*})=0$ from Assumption
\ref{A-2}, $\tilde{d}=d$ holds so that $\tilde{d}_{0}(l_{it})=d_{0}(l_{it})$.} Since $E[\varepsilon_{t}|x_{t},z_{t}]=0$, step 1 continues to identify
$\phi_{t}(\cdot)$. Therefore, once $\mathbb{M}_{t}^{-1}(m_{t},k_{t},l_{t},z_{t})$
is identified, step 3 identifies all the same objects as before. 
\begin{prop}
\label{P-IV2-1} Suppose that $l_{t}$ may correlate with $\eta_{t}$
and that Assumptions \ref{A-1}\textendash \ref{A-2}, \ref{A-FOC},
\ref{assu:labor} and \ref{assu:completeness} hold. Then, the production
function, output quantities, output prices and TFP are identified
up to scale and location; markups and output elasticities are identified
up to scale.
\end{prop}

\subsubsection{Endogenous Firm Characteristics}

Firm characteristics $z_{t}$ may correlate with $\eta_{t}$. For
simplicity, we again assume that $l_{t}$ is exogenous. We show that
even in the absence of any IV for $z_{t}$, we can identify the markup
and the production function. If valid IVs for $z_{t}$ are available,
all the same objects can be identified as before. 

We modify Assumption \ref{A-3} so that $z_{t}$ may correlate with
$\eta_{t}$.
\begin{assumption}
\label{assu:z_t} (i) Assumptions \ref{A-3} (a), (d), (e), and (f)
hold. (ii) $\eta_{t}$ is independent of $\bar{w}_{t}:=(k_{t},l_{t},x_{t-1},z_{t-1})'\in\mathcal{\bar{W}}:=\mathcal{K}\times\mathcal{L}\times\mathcal{X}\times\mathcal{Z}$.
$\bar{w}$ is continuously distributed on $\mathcal{\bar{W}}$. (iii)
For each $(x_{t-1},z_{t-1})\in\mathcal{X}\times\mathcal{Z}$, $\mathscr{A}_{m_{t}}\left(x_{t-1},z_{t-1}\right)=\{(\tilde{x}_{t},\tilde{z}_{t})\in\mathcal{X}\times\mathcal{Z}|\partial G_{m_{t}|v_{t}}(\tilde{m}_{t}|\tilde{k}_{t},\tilde{l}_{t},\tilde{z}_{t},x_{t-1},z_{t-1})/\partial m_{t}>0\}$
is non-empty.
\end{assumption}

\paragraph{Identification without Instrument Variables.}

The conditional distribution of $m_{t}$ given $v_{t}$ satisfies

\[
G_{m_{t}|v_{t}}(m|v_{t})=G_{\eta_{t}|z_{t}}\left(\mathbb{M}_{t}^{-1}(m,k_{t},l_{t},z_{t})-\bar{h}_{t}\left(x_{t-1},z_{t-1}\right)|z_{t}\right).
\]
Taking the derivatives of both sides with respect to $m$, $q_{t}\in\{m_{t},k_{t},l_{t}\}$
and $q_{t-1}\in\{k_{t-1},l_{t-1},m_{t-1},z_{t-1}\}$, we obtain (\ref{eq:identified_M})
and (\ref{eq:identified_h}). Following the same steps as in subsection
\ref{subsec:Endogenous-Labor-Input}, we identify $\partial\mathbb{M}_{t}^{-1}(m,k_{t},l_{t},z_{t})/\partial q_{t}$
and $\partial\bar{h}\left(x_{t-1},z_{t-1}\right)/\partial q_{t-1}$
up to scale. 

Since $E[\varepsilon_{t}|x_{t},z_{t}]=0$, Lemma \ref{lem:step2}
continues to hold and $\phi_{t}(\cdot)$ is identified. Therefore,
using (\ref{eq:phi_qt}) and the first-order condition (\ref{eq:foc_m})
with the identified derivatives of $\mathbb{M}_{t}^{-1}(\cdot)$,
it is possible to identify markup (\ref{eq:markup_identified}) and
output elasticities (\ref{eq:output_elasticities_pro2}) up to scale.
Integrating the output elasticities, we can identify the production
function, following (\ref{eq:f_integral}).
\begin{prop}
\label{prop:withoutIV} Suppose that $z_{t}$ may correlate with $\eta_{t}$
and that Assumptions \ref{A-data}, \ref{A-2}, \ref{A-FOC}, and
\ref{assu:z_t} hold. Then, we can identify the markup $\partial\varphi_{t}^{-1}(\bar{r}_{it},z_{it})/\partial r_{t}$
of each firm up to scale and the production function $f_{t}(\cdot)$
up to scale and location. 
\end{prop}
Applying Propositions \ref{P-3} and Proposition \ref{P-CRS}, it
is possible to identify the changes in markup and output elasticities
overtime and the levels of markup and elasticities, respectively. 

\paragraph{Identification with Instrument Variables.}

To identify $\varphi_{t}^{-1}(\cdot)$ and $\mathbb{M}_{t}^{-1}(\cdot)$,
we need a set of IVs $\zeta_{t}$ for $z_{t}$. A candidate for $\zeta_{t}$
is $z_{t-1}$ if $z_{t-1}$ correlates with $z_{t}$. 
\begin{assumption}
\label{A-IV_z} (a) There exits a set of instruments $\zeta_{t}$
such that $E[\eta_{t}|\zeta_{t}]=0$ a.s. (b) For all functions $\delta(z_{t}):\mathcal{Z}\rightarrow\mathbb{R}$
such that $E[\delta(z_{t})|\zeta_{t}]<\infty$, $E[\delta(z_{t})|\zeta_{t}]=0$
a.s. implies $\delta(z_{t})=0$ a.s.
\end{assumption}
Following similar steps by which to derive (\ref{eq:identified2}),
we obtain 
\begin{align*}
\mathbb{M}_{t}^{-1}(m_{t},k_{t},l_{t},z_{t}) & =d_{z}(z_{t})+\Lambda_{zt}\left(x_{t},z_{t}\right)
\end{align*}
and (\ref{eq:identified2}), where $d_{z}(z_{t}):=\mathbb{M}_{t}^{-1}(m_{t0}^{*},k_{t}^{*},l_{t}^{*},z_{t})$
is an unknown function to be identified; $\Lambda_{zt}\left(x_{t},z_{t}\right)$
is identified and treated as a known function.\footnote{Specifically, $\Lambda_{zt}\left(x_{t},z_{t}\right)$ is given by
\begin{align*}
\Lambda_{zt}\left(x_{t},z_{t}\right) & :=\int_{m_{t0}^{*}}^{m_{t}}\frac{\partial\mathbb{M}_{t}^{-1}(s,k_{t},l_{t},z_{t})}{\partial m_{t}}ds+\int_{k_{t}^{*}}^{k_{t}}\frac{\partial\mathbb{M}_{t}^{-1}(m_{t0}^{*},s,l_{t},z_{t})}{\partial k_{t}}ds+\int_{l_{t}^{*}}^{l_{t}}\frac{\partial\mathbb{M}_{t}^{-1}(m_{t0}^{*},k_{t}^{*},s,z_{t})}{\partial l_{t}}ds.
\end{align*}
} Defining $H_{it}^{zh}:=\Lambda_{zt}\left(x_{it},z_{it}\right)-\Lambda_{ht}\left(x_{it-1},z_{it-1}\right)$
as a known variable, we rewrite model (\ref{eq:model}) as
\[
H_{it}^{zh}=d-d_{z}(z_{it})+\eta_{it}.
\]
From Assumption \ref{A-IV_z}, the moment condition, $E[\eta_{it}|\zeta_{it}]=E\left[H_{it}^{zh}-d+d_{z}(z_{it})|\zeta_{it}\right]=0$,
identifies $\{d,d_{z}(z_{t})\}$.
\begin{prop}
\label{P-IV2} Suppose that $z_{t}$ may correlate with $\eta_{t}$
and a set of IVs $\zeta_{t}$ satisfies Assumption \ref{A-IV_z}.
Suppose Assumptions \ref{A-data}, \ref{A-2}, \ref{A-FOC}, and \ref{assu:z_t}
hold. Then, we can identify $\varphi_{t}^{-1}(\cdot)$ and $\mathbb{M}_{t}^{-1}(\cdot)$
up to scale and location and identify $G_{\eta}(\cdot)$ up to scale.
That is, output quantities, output prices and TFP are identified up
to scale and location.
\end{prop}

\subsubsection{Alternative Settings}

The Appendix presents the identification results in three alternative
settings. The identification argument remains the same but requires
some additional steps. 

\paragraph{Discrete Firm Characteristics.}

Observable firm characteristics $z_{t}$ may constitute a discrete
variable. Appendix \ref{subsec:Discrete_z} provides a proof.

\paragraph{Unobservable Firm-Level Demand-Shifter.}

The identification can incorporate an unobserved demand shifter $\xi_{it}$,
which can be called quality. Let $y_{it}^{\dagger}:=y_{it}+\xi_{it}$
and $p_{it}^{\dagger}:=y_{it}-\xi_{it}$ be the quality-adjusted output
and the quality-adjusted price, respectively. We consider the following
inverse function and revenue function:
\begin{align}
p_{it}^{\dagger} & =\psi_{t}\left(y_{it}^{\dagger},z_{it}\right),\nonumber \\
\bar{r}_{it} & =\varphi_{t}\left(y_{it}^{\dagger},z_{it}\right)=\varphi_{t}\left(f_{t}(x_{t})+\omega_{it}^{\dagger},z_{it}\right)\label{eq:quality}
\end{align}
where $\omega_{it}^{\dagger}\equiv\omega_{it}+\xi_{it}$ is a composite
of TFP and quality. In Appendix \ref{subsec:quality}, we show that
(\ref{eq:quality}) derives from a representative consumer's maximization
problem where $\exp(\xi_{it})$ enters the utility function in a multiplicative
manner with quantity. In (\ref{eq:quality}), higher quality allows
a firm to earn more revenue for a given output. We assume that $\tilde{\omega}_{it}$
follows a first-order Markov process $\omega_{it}^{\dagger}=h\left(\omega_{it-1}^{\dagger}\right)+\eta_{it}$. 

Under the current setting, the model structure becomes identical to
the main model where $(p_{it},y_{it},\omega_{it})$ are replaced with
$(p_{it}^{\dagger},y_{it}^{\dagger},\omega_{it}^{\dagger})$. Therefore,
applying precisely the same steps, we can identify all functions identified
in Section 3 and the quality-adjusted variables $(p_{it}^{\dagger},y_{it}^{\dagger},\omega_{it}^{\dagger})$.

\paragraph{IID Productivity Shock.}

As an alternative error structure, we consider an i.i.d. production
shock $e_{it}$ to output instead of a measurement error $\varepsilon_{it}$.
Then, the firm's observed revenue $r_{it}$ and inputs $x_{it}$ are
related as follows:
\begin{align}
r_{it} & =\varphi_{t}(f_{t}(x_{it})+\omega_{it}+e_{it},z_{it}).\label{eq:rev_iid}
\end{align}
A firm chooses $m_{it}$ at time $t$ by maximizing the expected profit:
\begin{align*}
m_{it} & =\mathbb{M}_{t}(\omega_{it},k_{it},l_{it},z_{it})\\
 & :=\argmax_{m\in\mathcal{M}}\ E\left[\exp\left(\varphi_{t}(f_{t}(m,k_{it},l_{it})+\omega_{it}+e_{it},z_{it})\right)|\mathcal{I}_{it}\right]-\exp(p_{t}^{m}+m).
\end{align*}
where $\mathcal{I}_{it}$ is the set of information for the firm that
includes all past variables and all time$-t$ variables except $e_{it}$.
The identification of the control function $\omega_{it}=\mathbb{M}_{t}^{-1}\left(m_{it},k_{it},l_{it},z_{it}\right)$
remains the same because $\mathbb{M}_{t}^{-1}(\cdot)$ continues to
be a function of the same variables.

In the second step, the revenue function (\ref{eq:rev_iid}) is written
as: 
\begin{align}
\varphi_{t}^{-1}(r_{it},z_{it}) & =f_{t}(x_{it})+\mathbb{M}_{t}^{-1}(x_{it},z_{it})+e_{it}.\label{eq:model2-IID}
\end{align}
Model (\ref{eq:model2-IID}) also belongs to the class of transformation
models studied by \citet{chiappori2015nonparametric}. Therefore,
by applying the nonparametric identification of a transformation model
and using the first-order condition for the material, we can identify
$\varphi_{t}(\cdot)$ and $f_{t}(\cdot)$ up to scale and location
from the conditional distribution of $r_{it}$ given $(x_{it},z_{it})$
under the assumptions similar to those for Proposition \ref{P-step3}.
As an additional complication, the first-order condition includes
expectation with respect to $e_{t}$. Therefore, we first identify
the distribution of $e_{t}$ to derive the first-order condition.
Appendix \ref{subsec:IID-Productivity-Shock} provides a proof. 

Because of the i.i.d. shock $e_{it}$, the realized value of $\partial\varphi_{t}^{-1}(r_{it},z_{it})/\partial r_{t}$
no longer equals the markup. We identify the markup from the cost
minimization, following \citet{hall1988relation} and \citet{de2012markups}.
As shown in Appendix \ref{subsec:IID-Productivity-Shock}, the equation
for the markup $\mu_{it}$ becomes 
\[
\mu_{it}=\frac{\partial f_{t}(x_{it})/\partial m_{it}}{\exp(p_{t}^{m}+m_{it})/\exp\left(r_{it}-e_{it}\right)}.
\]
The difference from the original Hall-De Loecker-Warzynski markup
equation (\ref{eq:DLW}) is $\exp\left(r_{it}-e_{it}\right)$ instead
of $\exp\left(\bar{r}_{it}\right)=\exp\left(r_{it}-\varepsilon_{it}\right)$.
While $\bar{r}_{t}=\phi_{t}(x_{t},z_{t})$ in (\ref{eq:DLW}) is a
deterministic function of $(x_{t},z_{t})$, $r_{t}-e_{t}$ is generally
not. Therefore, the markups are different across firms even after
being conditioned on $(x_{t},z_{t})$. 

\section{Concluding Remarks\label{sec:Concluding-Remarks}}

The current study developes constructive nonparametric identification
of production function and markup from revenue data. Our method simultaneously
addresses two fundamental identification issues raised in the literature
of production function estimation since \citet{ma44ecma}\textemdash namely,
correlations between inputs and TFP, and biases from markup heterogeneity
when revenue is used as output. Under standard assumptions, when revenue
is modeled as a function of output (rather than a mere proxy for output)
and firm's observed characteristics, various economic objects of interest
can be identified from revenue data. In an ongoing follow-up research,
we provide an estimation procedure and plan to estimate these objects
from an actual dataset.

\bibliographystyle{asa}
\bibliography{./markup_new}

\begin{thebibliography}{38}
\newcommand{\enquote}[1]{``#1''}
\expandafter\ifx\csname natexlab\endcsname\relax\def\natexlab#1{#1}\fi

\bibitem[{Ackerberg et~al.(2007)Ackerberg, Benkard, Berry, and
  Pakes}]{ACKERBERG20074171}
Ackerberg, D., Benkard, C.~L., Berry, S., and Pakes, A. (2007),
  \enquote{Chapter 63 Econometric Tools for Analyzing Market Outcomes,}
  Elsevier, vol.~6 of \textit{Handbook of Econometrics}, pp. 4171 -- 4276.

\bibitem[{Ackerberg et~al.(2015)Ackerberg, Caves, and
  Frazer}]{ackerberg2015identification}
Ackerberg, D.~A., Caves, K., and Frazer, G. (2015), \enquote{{I}dentification
  {P}roperties of {R}ecent {P}roduction {F}unction {E}stimators,}
  \textit{Econometrica}, 83, 2411--2451.

\bibitem[{Arellano and Bond(1991)}]{arellano1991some}
Arellano, M. and Bond, S. (1991), \enquote{Some tests of specification for
  panel data: Monte Carlo evidence and an application to employment equations,}
  \textit{Review of Economic Studies}, 58, 277--297.

\bibitem[{Arellano and Bover(1995)}]{arellano1995another}
Arellano, M. and Bover, O. (1995), \enquote{Another look at the instrumental
  variable estimation of error-components models,} \textit{Journal of
  Econometrics}, 68, 29--51.

\bibitem[{Bartelsman and Doms(2000)}]{bartelsman2000understanding}
Bartelsman, E.~J. and Doms, M. (2000), \enquote{Understanding productivity:
  Lessons from longitudinal microdata,} \textit{Journal of Economic
  literature}, 38, 569--594.

\bibitem[{Berry et~al.(1995)Berry, Levinsohn, and Pakes}]{blp1995ecta}
Berry, S., Levinsohn, J., and Pakes, A. (1995), \enquote{{A}utomobile {P}rices
  in {M}arket {E}quilibrium,} \textit{Econometrica}, 63, 841--890.

\bibitem[{Blundell and Bond(1998)}]{blundell1998initial}
Blundell, R. and Bond, S. (1998), \enquote{Initial conditions and moment
  restrictions in dynamic panel data models,} \textit{Journal of Econometrics},
  87, 115--143.

\bibitem[{Blundell and Bond(2000)}]{blundell2000gmm}
--- (2000), \enquote{GMM estimation with persistent panel data: an application
  to production functions,} \textit{Econometric Reviews}, 19, 321--340.

\bibitem[{Bond et~al.(2020)Bond, Hashemi, Kaplan, and Zoch}]{bond2020some}
Bond, S., Hashemi, A., Kaplan, G., and Zoch, P. (2020), \enquote{{S}ome
  {U}npleasant {M}arkup {A}rithmetic: {P}roduction {F}unction {E}lasticities
  and {T}heir {E}stimation from {P}roduction {D}ata,} \normalfont{NBER Working
  Paper} w27002.

\bibitem[{Chiappori et~al.(2015)Chiappori, Komunjer, and
  Kristensen}]{chiappori2015nonparametric}
Chiappori, P.-A., Komunjer, I., and Kristensen, D. (2015),
  \enquote{{N}onparametric {I}dentification and {E}stimation of
  {T}ransformation {M}odels,} \textit{Journal of Econometrics}, 188, 22--39.

\bibitem[{De~Loecker(2011)}]{de2011product}
De~Loecker, J. (2011), \enquote{{P}roduct {D}ifferentiation, {M}ultiproduct
  {F}irms, and {E}stimating the {I}mpact of {T}rade {L}iberalization on
  {P}roductivity,} \textit{Econometrica}, 79, 1407--1451.

\bibitem[{De~Loecker et~al.(2020)De~Loecker, Eeckhout, and Unger}]{de2020rise}
De~Loecker, J., Eeckhout, J., and Unger, G. (2020), \enquote{{T}he {R}ise of
  {M}arket {P}ower and the {M}acroeconomic {I}mplications,} \textit{Quarterly
  Journal of Economics}, 135, 561--644.

\bibitem[{De~Loecker et~al.(2016)De~Loecker, Goldberg, Khandelwal, and
  Pavcnik}]{de2016prices}
De~Loecker, J., Goldberg, P.~K., Khandelwal, A.~K., and Pavcnik, N. (2016),
  \enquote{{P}rices, {M}arkups, and {T}rade {R}eform,} \textit{Econometrica},
  84, 445--510.

\bibitem[{De~Loecker and Warzynski(2012)}]{de2012markups}
De~Loecker, J. and Warzynski, F. (2012), \enquote{{M}arkups and {F}irm-{L}evel
  {E}xport {S}tatus,} \textit{American Economic Review}, 102, 2437--71.

\bibitem[{Doraszelski and Jaumandreu(2018)}]{doraszelski2018measuring}
Doraszelski, U. and Jaumandreu, J. (2018), \enquote{Measuring the Bias of
  Technological Change,} \textit{Journal of Political Economy}, 126,
  1027--1084.

\bibitem[{Ekeland et~al.(2004)Ekeland, Heckman, and
  Nesheim}]{ekeland2004identification}
Ekeland, I., Heckman, J.~J., and Nesheim, L. (2004), \enquote{{I}dentification
  and {E}stimation of {H}edonic {M}odels,} \textit{Journal of Political
  Economy}, 112, S60--S109.

\bibitem[{Feenstra(2003)}]{feenstra2003homothetic}
Feenstra, R.~C. (2003), \enquote{A homothetic utility function for monopolistic
  competition models, without constant price elasticity,} \textit{Economics
  Letters}, 78, 79--86.

\bibitem[{Feenstra and Weinstein(2017)}]{feenstra2017globalization}
Feenstra, R.~C. and Weinstein, D.~E. (2017), \enquote{Globalization, markups,
  and US welfare,} \textit{Journal of Political Economy}, 125, 1040--1074.

\bibitem[{Flynn et~al.(2019)Flynn, Gandhi, and Traina}]{flynn2019measuring}
Flynn, Z., Gandhi, A., and Traina, J. (2019), \enquote{{M}easuring {M}arkups
  with {P}roduction {D}ata,} Unpublished.

\bibitem[{Foster et~al.(2008)Foster, Haltiwanger, and
  Syverson}]{foster2008reallocation}
Foster, L., Haltiwanger, J., and Syverson, C. (2008), \enquote{{R}eallocation,
  {F}irm {T}urnover, and {E}fficiency: {S}election on {P}roductivity or
  {P}rofitability?} \textit{American Economic Review}, 98, 394--425.

\bibitem[{Gandhi et~al.(2020)Gandhi, Navarro, and
  Rivers}]{gandhi2020identification}
Gandhi, A., Navarro, S., and Rivers, D.~A. (2020), \enquote{{O}n the
  {I}dentification of {G}ross {O}utput {P}roduction {F}unctions,}
  \textit{Journal of Political Economy}, forthcoming.

\bibitem[{Griliches and Mairesse(1999)}]{griliches_mairesse_1999}
Griliches, Z. and Mairesse, J. (1999), \enquote{Production Functions: The
  Search for Identification,} in \textit{Econometrics and Economic Theory in
  the 20th Century: The Ragnar Frisch Centennial Symposium}, ed. Str{\o}m, S.,
  Cambridge University Press, Econometric Society Monographs, pp. 169--203.

\bibitem[{Hall(1988)}]{hall1988relation}
Hall, R.~E. (1988), \enquote{{T}he {R}elation between {P}rice and {M}arginal
  {C}ost in {U}{S} {I}ndustry,} \textit{Journal of Political Economy}, 96,
  921--947.

\bibitem[{Heckman et~al.(2010)Heckman, Matzkin, and
  Nesheim}]{heckman2010nonparametric}
Heckman, J.~J., Matzkin, R.~L., and Nesheim, L. (2010),
  \enquote{{N}onparametric {I}dentification and {E}stimation of {N}onadditive
  {H}edonic {M}odels,} \textit{Econometrica}, 78, 1569--1591.

\bibitem[{Horowitz(1996)}]{horowitz1996semiparametric}
Horowitz, J.~L. (1996), \enquote{{S}emiparametric {E}stimation of a
  {R}egression {M}odel with an {U}nknown {T}ransformation of the {D}ependent
  {V}ariable,} \textit{Econometrica}, 103--137.

\bibitem[{Katayama et~al.(2009)Katayama, Lu, and Tybout}]{katayama2009firm}
Katayama, H., Lu, S., and Tybout, J.~R. (2009), \enquote{{F}irm-{L}evel
  {P}roductivity {S}tudies: {I}llusions and a {S}olution,}
  \textit{International Journal of Industrial Organization}, 27, 403--413.

\bibitem[{Klette and Griliches(1996)}]{klette1996jae}
Klette, T.~J. and Griliches, Z. (1996), \enquote{{T}he {I}nconsistency of
  {C}ommon {S}cale {E}stimators {W}hen {O}utput {P}rices {A}re {U}nobserved and
  {E}ndogenous,} \textit{Journal of Applied Econometrics}, 11, 343--361.

\bibitem[{Levinsohn and Petrin(2003)}]{levinsohn2003estimating}
Levinsohn, J. and Petrin, A. (2003), \enquote{{E}stimating {P}roduction
  {F}unctions {U}sing {I}nputs to {C}ontrol for {U}nobservables,}
  \textit{Review of Economic Studies}, 317--341.

\bibitem[{Lu and Yu(2015)}]{lu2015trade}
Lu, Y. and Yu, L. (2015), \enquote{Trade liberalization and markup dispersion:
  evidence from China's WTO accession,} \textit{American Economic Journal:
  Applied Economics}, 7, 221--53.

\bibitem[{Marschak and Andrews(1944)}]{ma44ecma}
Marschak, J. and Andrews, W. (1944), \enquote{{R}andom {S}imultaneous
  {E}quations and the {T}heory of {P}roduction,} \textit{Econometrica}, 12,
  143--205.

\bibitem[{Matsuyama and Ushchev(2017)}]{matsuyama2017beyond}
Matsuyama, K. and Ushchev, P. (2017), \enquote{{B}eyond CES: {T}hree
  {A}lternative {C}ases of {F}lexible {H}omothetic {D}emand {S}ystems,}
  \normalfont{Buffett Institute Global Poverty Research Lab Working Paper}
  17-109.

\bibitem[{Mr{\'a}zov{\'a} and Neary(2017)}]{mrazova2017not}
Mr{\'a}zov{\'a}, M. and Neary, J.~P. (2017), \enquote{{N}ot {S}o {D}emanding:
  {D}emand {S}tructure and {F}irm {B}ehavior,} \textit{American Economic
  Review}, 107, 3835--74.

\bibitem[{Mr{\'a}zov{\'a} and Neary(2019)}]{MRAZOVA2019102561}
--- (2019), \enquote{IO for {E}xports(s),} \textit{International Journal of
  Industrial Organization}, 102561.

\bibitem[{Newey and Powell(2003)}]{newey2003instrumental}
Newey, W.~K. and Powell, J.~L. (2003), \enquote{Instrumental variable
  estimation of nonparametric models,} \textit{Econometrica}, 71, 1565--1578.

\bibitem[{Nishioka and Tanaka(2019)}]{nishioka2019measuring}
Nishioka, S. and Tanaka, M. (2019), \enquote{Measuring Markups from Revenue and
  Total Cost: An Application to Japanese Plant-Product Matched Data,} Rieti
  {D}iscussion {P}aper {S}eries 19-{E}-018.

\bibitem[{Olley and Pakes(1996)}]{olley1996dynamics}
Olley, G.~S. and Pakes, A. (1996), \enquote{{T}he {D}ynamics of {P}roductivity
  in the {T}elecommunications {E}quipment {I}ndustry,} \textit{Econometrica},
  1263--1297.

\bibitem[{Syverson(2011)}]{syverson2011determines}
Syverson, C. (2011), \enquote{What Determines Productivity?} \textit{Journal of
  Economic Literature}, 49, 326--65.

\bibitem[{Van~Biesebroeck(2003)}]{van2003productivity}
Van~Biesebroeck, J. (2003), \enquote{Productivity dynamics with technology
  choice: An application to automobile assembly,} \textit{Review of Economic
  Studies}, 70, 167--198.

\end{thebibliography}

\newpage{}

\appendix

\section{Online Appendix (Not for Publication)}

\setcounter{page}{1}
\renewcommand{\thepage}{\Alph{section}.\arabic{page}}
\setcounter{equation}{0}
\renewcommand{\theequation}{\Alph{section}.\arabic{equation}}
\setcounter{assumption}{0}
\renewcommand{\theassumption}{\Alph{section}.\arabic{assumption}}
\setcounter{thm}{0}
\renewcommand{\thethm}{\Alph{section}.\arabic{thm}}
\setcounter{prop}{0}
\renewcommand{\theprop}{\Alph{section}.\arabic{prop}}
\setcounter{lem}{0}
\renewcommand{\thelem}{\Alph{section}.\arabic{lem}}

\subsection{Identification of Demand Function \label{subsec:Demand}}

\subsubsection{Proof for Proposition \ref{P-demand}}

The proof for Proposition \ref{P-demand} uses the following result
of \citeauthor{matsuyama2017beyond} (2017).
\begin{thm}
\label{thm:1}(\citeauthor{matsuyama2017beyond}, 2017, Remark 3 and
Proposition 1). Consider a mapping $\mathbf{s(Y)}:=(s_{1}(Y_{1}),...,s_{N}(Y_{N}))'$
from $\mathbb{R}_{+}^{N}$ to $\mathbb{R}_{+}^{N}$, which is differentiable
almost everywhere, is normalized by 
\begin{equation}
\sum_{i=1}^{N}s_{i}(Y_{i}^{*})=1,\label{eq:norm_s}
\end{equation}
for some point $\mathbf{Y}^{*}:=(Y_{1}^{*},...,Y_{N}^{*})$ and satisfies
the following conditions 
\begin{align}
s_{i}'(Y_{i})Y_{i} & <s_{i}(Y_{i})\text{ for }i=1,...,N,\nonumber \\
s_{i}'(Y_{i}) & s_{j}'(Y_{j})\ge0\text{ for }i,j=1,...,N,\label{eq:HSAcondition}
\end{align}
for all $\mathbf{Y}$ such that $\sum_{i=1}^{N}s_{i}(Y_{i})=1$. Then,
(1) for any such mapping, there exists a unique monotone, convex,
continuous, and homothetic rational preference that generates the
HSA demand system described by 
\begin{align*}
P_{i} & =\frac{\Phi}{Y_{i}}s_{i}\left(\frac{Y_{i}}{A\left(\mathbf{Y}\right)}\right)\text{ for }i=1,..,N,
\end{align*}
where $\Phi:=\sum_{i=1}^{N}P_{i}Y_{i}$ and $A(\mathbf{Y})$ is obtained
by solving
\[
\sum_{i=1}^{N}s_{i}\left(\frac{Y_{i}}{A\left(\mathbf{Y}\right)}\right)=1.
\]
 (2) This homothetic preference is described by a utility function
$U$ which is defined by
\begin{equation}
\ln U(\mathbf{Y})=\ln A(\mathbf{Y})+\sum_{i=1}^{N}\int_{c}^{Y_{i}/A(Y)}\frac{s_{i}\left(\xi\right)}{\xi}d\xi,\label{eq:utility}
\end{equation}
where $c$ is a constant.
\end{thm}
\citet{matsuyama2017beyond} proved (1) from the Antonelli's integrability
theorem. See their paper for the proof. \citet{matsuyama2017beyond}
provides a proof for (2) for the case of direct demand functions instead
of inverse demand functions considered here. So we will provide the
proof for (2) in the following proof for Proposition \ref{P-demand}
(b).

\paragraph{Proof for Proposition \ref{P-demand}}
\begin{proof}
(a) We construct $S_{t}(Y_{i}/A_{t}(\mathbf{Y},z_{t}),z_{it})$ and
$A_{t}(\mathbf{Y}_{t},\mathbf{z}_{t})$ as is explained in the main
text. Fix $\mathbf{z}_{t}:=(z_{1t},...,z_{Nt})$ and time $t$. For
$\mathbf{Y}\in\mathscr{\mathscr{\bar{Y}}}$, define $A(\mathbf{Y}):=A_{t}(\mathbf{Y},\mathbf{z}_{t})$
and $s(\mathbf{Y}):=\left(s_{1}(Y_{1}),...,s_{N}(Y_{N})\right)$ such
that $s_{i}(Y_{i})=S_{t}(Y_{i},z_{it})$. 

Define $\mathscr{\mathscr{\bar{Y}}}_{A}:=\{\mathbf{Y}/A(\mathbf{Y}):\mathbf{Y}\in\mathscr{\mathscr{\bar{Y}}}\}.$
Then, for all $\mathbf{Y}\in\mathscr{\bar{Y}}_{A}$, $\sum_{i=1}^{N}s_{i}(Y_{i})=1$
holds by construction of $A_{t}(\cdot)$. At the same time, for all
$\mathbf{Y}$ that satisfies $\sum_{i=1}^{N}s_{i}(Y_{i})=1$, $A(\mathbf{Y})=1$
holds so that $\mathbf{Y}\in\mathscr{\bar{Y}}_{A}$. Therefore, $\mathscr{\bar{Y}}_{A}=\{\mathbf{Y}\in\mathscr{\bar{Y}}:\sum_{i=1}^{N}s_{i}(Y_{i})=1\}$. 

Consider $\mathbf{Y}\in\mathscr{Y}_{A}$. From Assumption \ref{A-1}
(b) and $y:=\ln Y$,
\[
0<\frac{\partial\varphi_{t}\left(\ln Y,z\right)}{\partial\ln Y}=1+\frac{\partial\psi_{t}\left(\ln Y,z\right)}{\partial\ln Y}<1
\]
holds. The above inequality implies 
\[
s_{i}'(Y)>0\text{ and }s_{i}'(Y)Y<s_{i}(Y)\text{ for all }i\text{ and }Y
\]
because 
\begin{align*}
s_{i}'(Y)Y & =\frac{\exp\left(\varphi_{t}\left(\ln Y,z_{it}\right)\right)}{\Phi_{t}}\frac{\partial\varphi_{t}\left(\ln Y,z_{it}\right)}{\partial\ln Y}\\
 & =s_{i}(Y)\frac{\partial\varphi_{t}\left(\ln Y,z_{it}\right)}{\partial\ln\tilde{Y}}.
\end{align*}
Therefore, $\mathbf{s(Y)}$ satisfies the inequalities in (\ref{eq:HSAcondition})
for all $\mathbf{Y}$ satisfying $\sum_{i=1}^{N}s_{i}(Y_{i})=1$.
From Theorem \ref{thm:1} (1), there exists a unique monotone, convex,
continuous, and homothetic rational preference that generates 
\begin{align*}
P_{it} & =\frac{\Phi_{t}}{Y_{it}}s_{i}\left(\frac{Y_{it}}{A\left(\mathbf{Y}_{t}\right)}\right)\\
 & =\frac{\Phi_{t}}{Y_{it}}S_{t}\left(\frac{Y_{it}}{A\left(\mathbf{Y}_{t},\mathbf{z}_{t}\right)},z_{it}\right),
\end{align*}
where $\Phi_{t}$ is the consumer's budget.  

(b) The following derivation of the utility function follows the steps
in \citet{matsuyama2017beyond}. Let $U_{t}(\mathbf{Y}_{t},\mathbf{z}_{t})$
be the utility function that is homogenous of degree one with respect
to $\mathbf{Y}_{t}$. Then, the indirect utility is linear in income
$\Phi_{t}$:
\begin{equation}
V_{t}(\mathbf{P}_{t},\Phi_{t})=\max_{\mathbf{Y}_{t}}\{U_{t}(\mathbf{Y}_{t},\mathbf{z}_{t})|\mathbf{P}_{t}'\mathbf{Y}_{t}\le\Phi_{t}\}=\frac{\Phi_{t}}{\Pi_{t}(\mathbf{P}_{t})},\label{eq:utility_max}
\end{equation}
where $\Pi_{t}(\mathbf{P}_{t})$ is the ideal price index. The first-order
condition is given by
\[
\frac{\partial U_{t}\left(\mathbf{Y}_{t},\mathbf{z}_{t}\right)}{\partial Y_{it}}=\lambda_{t}P_{it},
\]
where $\lambda_{t}=1/\Pi_{t}(\mathbf{P}_{t})$ is the Lagrange multiplier.
The Roy's identity derives the demand for firm $i$ as
\begin{equation}
Y_{it}=-\frac{\partial V_{t}/\partial P_{it}}{\partial V_{t}/\partial\Phi_{t}}=\frac{\Phi_{t}}{P_{it}}\left(\frac{\partial\Pi_{t}}{\partial P_{it}}\frac{P_{it}}{\Pi_{t}}\right).\label{Roy}
\end{equation}
From (\ref{eq:utility_max}), the expenditure function is written
as $e_{t}(\mathbf{P}_{t},U_{t})=\Pi_{t}(\mathbf{P}_{t})U_{t}.$ Applying
the Shepard's lemma derives the demand for firm $i$ as
\begin{equation}
Y_{it}=\frac{\partial e_{t}(\mathbf{P}_{t},U_{t})}{\partial P_{it}}=\frac{\partial\Pi_{t}}{\partial P_{it}}U_{t}.\label{eq:Shepard}
\end{equation}
Using (\ref{eq:Shepard}), $\lambda_{t}=1/\Pi_{t}$ and the first-order
condition, we obtain
\[
\frac{\partial\Pi_{t}}{\partial P_{it}}\frac{P_{it}}{\Pi_{t}}=\frac{Y_{it}}{U_{t}}\frac{P_{it}}{\Pi_{t}}=\frac{Y_{it}}{U_{t}}\lambda_{t}P_{it}=\frac{\partial U_{t}}{\partial Y_{it}}\frac{Y_{it}}{U_{t}}.
\]
Therefore, from (\ref{Roy}), we have 
\[
S_{t}\left(\frac{Y_{it}}{A\left(\mathbf{Y}_{t},\mathbf{z}_{t}\right)},z_{it}\right)=\frac{P_{it}Y_{it}}{\Phi_{t}}=\frac{\partial\Pi_{t}}{\partial P_{it}}\frac{P_{it}}{\Pi_{t}}=\frac{\partial U_{t}}{\partial Y_{it}}\frac{Y_{it}}{U_{t}},
\]
which can be written as 
\begin{equation}
\frac{\partial\ln U_{t}(\mathbf{Y}_{t},\mathbf{z}_{t})}{\partial Y_{it}}=\frac{1}{Y_{it}}S_{t}\left(\frac{Y_{it}}{A_{t}\left(\mathbf{Y}_{t},\mathbf{z}_{t}\right)},z_{it}\right).\label{eq:dlnU/dY}
\end{equation}
Let $A_{t}=A_{t}\left(\mathbf{Y}_{t},\mathbf{z}_{t}\right)$. Since
$U_{t}(\mathbf{Y}_{t},\mathbf{z}_{t})$ is homogeneous of degree one
with respect to $\mathbf{Y}_{t}$, $\partial U_{t}(\mathbf{Y}_{t},\mathbf{z}_{t})/\partial Y_{it}$
is homogenous of degree zero with respect to $\mathbf{Y}_{t}$. Therefore,
it holds 
\begin{align*}
\frac{\partial\ln U_{t}(\mathbf{Y}_{t}/A_{t},\mathbf{z}_{t})}{\partial Y_{it}} & =\frac{\partial U_{t}(\mathbf{Y}_{t}/A_{t},\mathbf{z}_{t})}{\partial Y_{it}}\frac{1}{U_{t}(\mathbf{Y}_{t}/A_{t},\mathbf{z}_{t})}\\
 & =\frac{\partial U_{t}(\mathbf{Y}_{t},\mathbf{z}_{t})}{\partial Y_{it}}\frac{A_{t}}{U_{t}(\mathbf{Y}_{t},\mathbf{z}_{t})}\\
 & =A_{t}\frac{\partial\ln U_{t}(\mathbf{Y}_{t},\mathbf{z}_{t})}{\partial Y_{it}}.
\end{align*}
Then, (\ref{eq:dlnU/dY}) becomes simplified as 
\begin{align}
\frac{\partial\ln U_{t}(\mathbf{Y}_{t},\mathbf{z}_{t})}{\partial Y_{it}} & =\frac{1}{Y_{it}}S_{t}\left(\frac{Y_{it}}{A_{t}},z_{it}\right)\nonumber \\
\Leftrightarrow\frac{\partial\ln U_{t}(\mathbf{Y}_{t}/A_{t},\mathbf{z}_{t})}{\partial Y_{it}} & =\frac{A_{t}}{Y_{it}}S_{t}\left(\frac{Y_{it}}{A_{t}},z_{it}\right)\nonumber \\
\Leftrightarrow\frac{\partial\ln U_{t}(\tilde{\mathbf{Y}}_{t},\mathbf{z}_{t})}{\partial\tilde{Y}_{it}} & =\frac{S_{t}\left(\tilde{Y}_{it},z_{it}\right)}{\tilde{Y}_{it}},\label{eq:lnU}
\end{align}
where $\tilde{Y}_{it}:=Y_{it}/A_{t}$ and $\tilde{\mathbf{Y}}_{t}:=\left(\tilde{Y}_{1t},...,\tilde{Y}_{Nt}\right)$.
Let $c_{t}(\mathbf{z}_{t}):=\left(c_{1t}(\mathbf{z}_{t}),...,c_{Nt}(\mathbf{z}_{t})\right)$
be defined by $U_{t}(c_{t}(\mathbf{z}_{t}),\mathbf{z}_{t})=1$. Then,
integration of (\ref{eq:lnU}) leads to 
\[
\ln U_{t}(\tilde{\mathbf{Y}}_{t},\mathbf{z}_{t})=\sum_{i=1}^{N}\int_{c_{it}(\mathbf{z}_{t})}^{\tilde{Y}_{it}}\frac{S_{t}\left(\xi,z_{it}\right)}{\xi}d\xi.
\]
Since $\ln U_{t}(\tilde{\mathbf{Y}}_{t},\mathbf{z}_{t})=\ln U_{t}(\mathbf{Y}_{t}/A_{t},\mathbf{z}_{t})=\ln U_{t}(\mathbf{Y}_{t},\mathbf{z}_{t})-\ln A_{t}$,
we obtain the utility function stated in the proposition as follows:
\[
\ln U_{t}(\mathbf{Y}_{t},\mathbf{z}_{t})=\ln A_{t}\left(\mathbf{Y}_{t},\mathbf{z}_{t}\right)+\sum_{i=1}^{N}\int_{c_{it}(\mathbf{z}_{t})}^{Y_{it}/A_{t}\left(\mathbf{Y}_{t},\mathbf{z}_{t}\right)}\frac{S_{t}\left(\xi,z_{it}\right)}{\xi}d\xi.
\]

(c) The homothetic preference implies that the market share $P_{it}Y_{it}/\Phi_{t}$
depends only on a price vector and is independent of income. This
property requires $A_{t}\left(\mathbf{Y}_{t},\mathbf{z}_{t}\right)$
to be homogenous of degree one with respect to $\mathbf{Y}_{t}$ so
that for any $k>0$, it
\[
S_{t}\left(\frac{kY_{it}}{A_{t}\left(k\mathbf{Y}_{t},\mathbf{z}_{t}\right)},z_{it}\right)=S_{t}\left(\frac{kY_{it}}{kA_{t}\left(\mathbf{Y}_{t},\mathbf{z}_{t}\right)},z_{it}\right)=S_{t}\left(\frac{Y_{it}}{A_{t}\left(\mathbf{Y}_{t},\mathbf{z}_{t}\right)},z_{it}\right).
\]

Let $\varphi_{t}^{-1}(\bar{r}_{it},z_{it})$ be the identified log
output and $\varphi_{t}^{*-1}(\bar{r}_{it},z_{it})$ be its true value.
Since $\varphi_{t}^{-1}(\bar{r}_{it},z_{it})$ is identified up to
location, there is $a\in\mathbb{R}$ such that $\varphi_{t}^{-1}(\bar{r}_{it},z_{it})=a+\varphi_{t}^{*-1}(\bar{r}_{it},z_{it})$.

The identified output $Y_{it}$ and the true output $Y_{it}^{*}$
are related as follows:
\begin{align*}
Y_{it} & =\exp(\varphi_{t}^{-1}(\bar{r}_{it},z_{it}))\\
 & =\exp(a+\varphi_{t}^{*-1}(\bar{r}_{it},z_{it}))\\
 & =\exp(a)Y_{it}^{*}.
\end{align*}
Since $\varphi_{t}(y_{t},z_{t})=\varphi_{t}^{*}(y_{t}-a,z_{t})$ for
all $y_{t}$ and $z_{t}$,
\[
\varphi_{t}\left(\ln Y_{it},z_{it}\right)=\varphi_{t}^{*}\left(\ln Y_{it}-a,z_{it}\right)=\varphi_{t}^{*}\left(\ln Y_{it}^{*},z_{it}\right).
\]
Then, the market share function $S_{t}(Y_{it},z_{it}):=\exp\left(\varphi_{t}\left(\ln Y_{it},z_{it}\right)\right)/\Phi_{t}$
constructed from the identified outputs agrees with the market share
function $S_{t}^{*}(Y_{it}^{*},z_{it}):=\exp\left(\varphi_{t}^{*}\left(\ln Y_{it}^{*},z_{it}\right)\right)/\Phi_{t}$
constructed from the true outputs: 
\[
S_{t}(Y_{it},z_{it})=\frac{\exp\left(\varphi_{t}\left(\ln Y_{it},z_{it}\right)\right)}{\Phi_{t}}=\frac{\exp\left(\varphi_{t}^{*}\left(\ln Y_{it}^{*},z_{it}\right)\right)}{\Phi_{t}}=S_{t}^{*}(Y_{it}^{*},z_{it}).
\]
Thus, the identified demand system does not depend on the location
normalization of $\varphi^{-1}(\cdot)$.

Since the quantity index $A_{t}(\mathbf{Y}_{t},\mathbf{z}_{t})$ is
homogenous of degree one with respect to $\mathbf{Y}_{t}$, 
\begin{align*}
\frac{Y_{it}}{A(\mathbf{Y}_{t},\mathbf{z}_{t})} & =\frac{\exp(a)Y_{it}^{*}}{A(\exp(a)\mathbf{Y}_{t}^{*},\mathbf{z}_{t})}=\frac{\exp(a)Y_{it}^{*}}{\exp(a)A(\mathbf{Y}_{t}^{*},\mathbf{z}_{t})}=\frac{Y_{it}^{*}}{A(\mathbf{Y}_{t}^{*},\mathbf{z}_{t})}
\end{align*}
Let $U_{t}(\mathbf{Y}_{t},\mathbf{z}_{t})$ be the identified utility
and $U_{t}^{*}(\mathbf{Y}_{t}^{*},\mathbf{z}_{t})$ be the true utility.
Then, they are related as
\begin{align*}
\ln U_{t}(\mathbf{Y}_{t},\mathbf{z}_{t}) & =\ln A_{t}(\mathbf{Y}_{t},\mathbf{z}_{t})+\sum_{i=1}^{N}\int_{c_{i}(\mathbf{z}_{t})}^{Y_{it}/A_{t}(\mathbf{Y}_{t},\mathbf{z}_{t})}\frac{S_{t}\left(\xi,z_{it}\right)}{\xi}d\xi,\\
 & =a+\ln A_{t}\left(\mathbf{Y}_{t}^{*},\mathbf{z}_{t}\right)+\sum_{i=1}^{N}\int_{c_{i}^{*}(\mathbf{z}_{t})}^{Y_{it}^{*}/A_{t}\left(\mathbf{Y}_{t}^{*},\mathbf{z}_{t}\right)}\frac{S_{t}^{*}\left(\xi,z_{it}\right)}{\xi}d\xi\\
 & =a+\ln U_{t}^{*}(\mathbf{Y}_{t}^{*},\mathbf{z}_{t}),
\end{align*}
where $c_{t}^{*}(\mathbf{z}_{t}):=\left(c_{1t}^{*}(\mathbf{z}_{t}),...,c_{Nt}^{*}(\mathbf{z}_{t})\right)$
defined by $U^{*}(c_{t}^{*}(\mathbf{z}_{t}),\mathbf{z}_{t})=1$. Therefore,
the log utility function is identified up to the location normalization
of $\varphi_{t}^{-1}(\cdot)$. The identified utility function is
a monotonic transformation of the true utility function, which implies
both utility functions represent the same consumer preference. 
\end{proof}

\subsection{Discrete Firm Characteristics $z_{t}$ \label{subsec:Discrete_z}}

This section proves Propositions \ref{P-step1} and \ref{P-step3}
for the case that $z_{it}$ is a discrete variable and have finite
support $\mathcal{Z}:=\{z^{1},...,z^{J}\}$. 

The following assumption modifies Assumption \ref{A-1} for discrete
$z_{it}$.
\begin{assumption}
\label{A-1-discrete} (a) $f_{t}(\cdot)$ is continuously differentiable
with respect to $(m,k,l)$ on $\mathcal{M}\times\mathcal{K}\times\mathcal{L}$
and strictly increasing in $m$. (b) For every $z\in\mathcal{Z}$,
$\varphi_{t}(\cdot,z)$ is strictly increasing and invertible with
its inverse $\varphi_{t}^{-1}(\bar{r},z)$, which is continuously
differentiable with respect to $\bar{r}$ on $\mathcal{\bar{R}}$.
(c) For every $(k,l,z)\in\mathcal{K}\times\mathcal{L}\times\mathcal{Z}$,
$\mathbb{M}_{t}(\cdot,k,l,z)$ is strictly increasing and invertible
with its inverse $\mathbb{M}_{t}^{-1}(m,k,l,z)$, which is continuously
differentiable with respect to $(m,k,l)$ on $\mathcal{M}\times\mathcal{K}\times\mathcal{L}$.
(d) $\varepsilon_{t}$ is mean independent of $x_{t}$ and $z_{t}$
with $E\left[\varepsilon_{t}|x_{t},z_{t}\right]=0$.
\end{assumption}
The following assumption modifies Assumption \ref{A-3} for discrete
$z_{it}$.
\begin{assumption}
\label{A-3-discrete} (a) The distribution $G_{\eta}(\cdot)$ of $\eta$
is absolutely continuous with a density function $g_{\eta}(\cdot)$
that is continuous on its support. (b) $\eta_{t}$ is independent
of $v_{t}:=(k_{t},l_{t},z_{t},x_{t-1},z_{t-1})'\in\mathcal{V}:=\mathcal{K}\times\mathcal{L}\times\mathcal{Z}\times\mathcal{X}\times\mathcal{Z}$.
(c) $x$ is continuously distributed on $\mathcal{X}$. (d) The support
$\varOmega$ of $\omega$ is an interval $[\text{\ensuremath{\underbar{\ensuremath{\omega}}}},\bar{\omega}]\subset\mathbb{R}$
where $\text{\ensuremath{\underbar{\ensuremath{\omega}}}}<0$ and
$1<\bar{\omega}$. (e) $h(\cdot)$ is continuously differentiable
with respect to $\omega$ on $\Omega$. (f) The set $\mathcal{A}_{q_{t-1}}:=\{(x_{t-1},z_{t-1})\in\mathcal{X}\times\mathcal{Z}:\partial G_{m_{t}|v_{t}}(m_{t}|v_{t})/\partial q_{t-1}\neq0\text{ for all }(m_{t},k_{t},l_{t},z_{t})\in\mathcal{M}\times\mathcal{K}\times\mathcal{L}\times\mathcal{Z}\}$
is nonempty for some $q_{t-1}\in\{k_{t-1},l_{t-1},m_{t-1},z_{t-1}\}$.
(g) For each $(x_{t-1},z_{t-1})\in\mathcal{X}\times\mathcal{Z}$,
it is possible to find $(x_{t},z_{t})\in\mathcal{X}\times\mathcal{Z}$
such that $\partial G_{m_{t}|v_{t}}(m_{t}|k_{t},l_{t},z_{t},x_{t-1},z_{t-1})/\partial m_{t}>0$. 
\end{assumption}
A sufficient condition for Assumption \ref{A-3-discrete} (g) is $g_{\eta}(\eta)>0$
for all $\eta\in\mathbb{R}$, under which (\ref{eq:m2-dis}) below
shows $\partial G_{m_{t}|v_{t}}(m_{t}|k_{t},l_{t},z_{t},x_{t-1},z_{t-1})/\partial m_{t}>0$
holds for all $(x_{t},z_{t})$.

The following proposition establishes the identification of $\mathbb{M}_{t}^{-1}(\cdot)$.
\begin{prop}
\label{P-step1-discrete} Suppose that Assumptions \ref{A-data},
\ref{A-2}, \ref{A-1-discrete}, and \ref{A-3-discrete} hold. Then,
we can identify $\mathbb{M}_{t}^{-1}(m_{t},k_{t},l_{t},z_{t})$ up
to scale and location, and identify $G_{\eta}(\cdot)$ up to scale.
\end{prop}
\begin{proof}
Choose normalization points $(m_{t1}^{*},k_{t}^{*},l_{t}^{*})$ and
$(m_{t0}^{*},k_{t}^{*},l_{t}^{*})$ in Assumption \ref{A-2} as well
as $x_{t-1}^{*}\in\mathcal{X}$ such that, for $z_{t},z_{t-1}\in\mathcal{Z}$,
\begin{equation}
\mathbb{M}_{t}^{-1}(m_{t0}^{*},k_{t}^{*},l_{t}^{*},z_{t})=c_{0}(z_{t}),\,\mathbb{M}_{t}^{-1}(m_{t1}^{*},k_{t}^{*},l_{t}^{*},z_{t})=c_{1}(z_{t})\text{, }\text{and }\ensuremath{\bar{h}(x_{t-1}^{*},z_{t-1})=c_{2}(z_{t-1}),}\label{eq:norm_dis}
\end{equation}
where $\{c_{0}(z_{t}),c_{1}(z_{t}),c_{2}(z_{t-1})\}_{z_{t},z_{t-1}\in\mathcal{Z}}$
are unknown constants. Without loss of generality, let $z_{t}^{*}$
in Assumption \ref{A-2} be $z_{t}^{*}=z^{1}$. Thus, the normalization
in Assumption \ref{A-2} is imposed as 
\[
c_{0}(z^{1})=0\text{ and }c_{1}(z^{1})=1.
\]

From $\eta_{t}\perp v_{t}$, the conditional distribution of $m_{t}$
given $v_{t}$ satisfies 
\begin{align*}
G_{m_{t}|v_{t}}(m_{t}|v_{t}) & =G_{\eta}\left(\mathbb{M}_{t}^{-1}(m_{t},k_{t},l_{t},z_{t})-\bar{h}_{t}\left(x_{t-1},z_{t-1}\right)\right).
\end{align*}
Taking the derivatives of $G_{m_{t}|v_{t}}(m_{t}|v_{t})$ with respect
to $q_{t}\in\{m_{t},k_{t},l_{t}\}$ and $q_{t-1}\in\{k_{t-1},l_{t-1},m_{t-1}\}$
. The derivatives of $G_{m_{t}|v_{t}}(m|v)$ are
\begin{align}
\frac{\partial G_{m_{t}|v_{t}}\left(m_{t}|v_{t}\right)}{\partial q_{t}} & =\frac{\partial\mathbb{M}_{t}^{-1}(m_{t},k_{t},l_{t},z_{t})}{\partial q_{t}}g_{\eta}\left(\mathbb{M}_{t}^{-1}(m_{t},k_{t},l_{t},z_{t})-\bar{h}_{t}\left(x_{t-1},z_{t-1}\right)\right),\label{eq:m2-dis}\\
\frac{\partial G_{m_{t}|v_{t}}\left(m_{t}|v_{t}\right)}{\partial q_{t-1}} & =-\frac{\partial\bar{h}\left(x_{t-1},z_{t-1}\right)}{\partial q_{t-1}}g_{\eta}\left(\mathbb{M}_{t}^{-1}(m_{t},k_{t},l_{t},z_{t})-\bar{h}_{t}\left(x_{t-1},z_{t-1}\right)\right).\label{eq:m3-dis}
\end{align}
Using Assumption \ref{A-3-discrete} (f), we can choose $q_{t-1}\in\{k_{t-1},l_{t-1},m_{t-1},z_{t-1}\}$
and $(\tilde{x}_{t-1},\tilde{z}_{t-1})\in\mathcal{A}_{q_{t-1}}$ such
that $\partial G_{m_{t}|v_{t}}\left(m_{t}|k_{t},l_{t},z_{t},\tilde{x}_{t-1},\tilde{z}_{t-1}\right)/\partial q_{t-1}\neq0$
for all $(m_{t},k_{t},l_{t},z_{t})\in\mathcal{M}\times\mathcal{K}\times\mathcal{L}\times\mathcal{Z}$.
Dividing (\ref{eq:m2-dis}) by (\ref{eq:m3-dis}), respectively, we
obtain for $q_{t}\in\{m_{t},k_{t},l_{t}\}$
\begin{align}
\frac{\partial\mathbb{M}_{t}^{-1}(m_{t},k_{t},l_{t},z_{t})}{\partial q_{t}} & =-\frac{\partial\bar{h}\left(\tilde{x}_{t-1},\tilde{z}_{t-1}\right)}{\partial q_{t-1}}\frac{\partial G_{m_{t}|v_{t}}\left(m_{t}|k_{t},l_{t},z_{t},\tilde{x}_{t-1},\tilde{z}_{t-1}\right)/\partial q_{t}}{\partial G_{m_{t}|v_{t}}\left(m_{t}|k_{t},l_{t},z_{t},\tilde{x}_{t-1},\tilde{z}_{t-1}\right)/\partial q_{t-1}}.\label{eq:dM/dq_dis}
\end{align}
Then, from (\ref{eq:norm_dis}) and (\ref{eq:dM/dm_id_dis}), we have
\begin{align*}
1 & =c_{1}(z^{1})-c_{0}(z^{1})\\
 & =\mathbb{M}_{t}^{-1}(m_{t1}^{*},k_{t}^{*},l_{t}^{*},z^{1})-\mathbb{M}_{t}^{-1}(m_{t0}^{*},k_{t}^{*},l_{t}^{*},z^{1})\\
 & =-\frac{\partial\bar{h}\left(\tilde{x}_{t-1},\tilde{z}_{t-1}\right)}{\partial q_{t-1}}\int_{m_{t0}^{*}}^{m_{t1}^{*}}\frac{\partial G_{m_{t}|v_{t}}\left(m|k_{t}^{*},l_{t}^{*},z^{1},\tilde{x}_{t-1},\tilde{z}_{t-1}\right)/\partial m_{t}}{\partial G_{m_{t}|v_{t}}\left(m|k_{t}^{*},l_{t}^{*},z^{1},\tilde{x}_{t-1},\tilde{z}_{t-1}\right)/\partial q_{t-1}}dm_{t}
\end{align*}
and therefore identify $\partial\bar{h}\left(\tilde{x}_{t-1},\tilde{z}_{t-1}\right)/\partial q_{t-1}$
 as
\begin{equation}
\frac{\partial\bar{h}\left(\tilde{x}_{t-1},\tilde{z}_{t-1}\right)}{\partial q_{t-1}}=-\tilde{S}_{q_{t-1}},\label{eq:h-bar_dis2}
\end{equation}
where 
\begin{align*}
\tilde{S}_{q_{t-1}} & :=\left(\int_{m_{t0}^{*}}^{m_{t1}^{*}}\frac{\partial G_{m_{t}|v_{t}}\left(m|k_{t}^{*},l_{t}^{*},z^{1},\tilde{x}_{t-1},\tilde{z}_{t-1}\right)/\partial m_{t}}{\partial G_{m_{t}|v_{t}}\left(m|k_{t}^{*},l_{t}^{*},z^{1},\tilde{x}_{t-1},\tilde{z}_{t-1}\right)/\partial q_{t-1}}dm_{t}\right)^{-1}.
\end{align*}
 By substituting (\ref{eq:h-bar_dis2}) into (\ref{eq:dM/dq_dis}),
we can identify $\partial\mathbb{M}_{t}^{-1}(m_{t},k_{t},l_{t},z_{t})/\partial m_{t}$
and $\partial\mathbb{M}_{t}^{-1}(m_{t},k_{t},l_{t},z_{t})/\partial q_{t}$
as
\begin{align}
\frac{\partial\mathbb{M}_{t}^{-1}(m_{t},k_{t},l_{t},z_{t})}{\partial m_{t}} & =\tilde{S}_{q_{t-1}}T_{m_{t}q_{t-1}}(x_{t},z_{t}),\nonumber \\
\frac{\partial\mathbb{M}_{t}^{-1}(m_{t},k_{t},l_{t},z_{t})}{\partial q_{t}} & =\tilde{S}_{q_{t-1}}T_{q_{t}q_{t-1}}(x_{t},z_{t}),\label{eq:dM/dm_id_dis}
\end{align}
where
\begin{align*}
T_{m_{t}q_{t-1}}(x_{t},z_{t}) & :=\text{\ensuremath{\frac{\partial G_{m_{t}|v_{t}}\left(m_{t}|k_{t},l_{t},z_{t},\tilde{x}_{t-1},\tilde{z}_{t-1}\right)/\partial m_{t}}{\partial G_{m_{t}|v_{t}}\left(m_{t}|k_{t},l_{t},z_{t},\tilde{x}_{t-1},\tilde{z}_{t-1}\right)/\partial q_{t-1}}}},\\
T_{q_{t}q_{t-1}}(x_{t},z_{t}) & :=\frac{\partial G_{m_{t}|v_{t}}\left(m_{t}|k_{t},l_{t},z_{t},\tilde{x}_{t-1},\tilde{z}_{t-1}\right)/\partial q_{t}}{\partial G_{m_{t}|v_{t}}\left(m_{t}|k_{t},l_{t},z_{t},\tilde{x}_{t-1},\tilde{z}_{t-1}\right)/\partial q_{t-1}}.
\end{align*}

From (\ref{eq:norm_dis}) and (\ref{eq:dM/dm_id_dis}), $\mathbb{M}_{t}^{-1}(x_{t},z_{t})$
is written as

\begin{equation}
\mathbb{M}_{t}^{-1}(x_{t},z_{t})=c_{0}(z_{t})+\Lambda_{m}(x_{t},z_{t}),\label{eq:M_dis}
\end{equation}

where
\begin{align*}
\Lambda_{m}(x_{t},z_{t}) & :=\tilde{S}_{q_{t-1}}\left\{ \int_{m_{t0}^{*}}^{m_{t}}T_{m_{t}q_{t-1}}(s,k_{t},l_{t},z_{t})ds\right.\\
 & \left.+\int_{k^{*}}^{k_{t}}T_{k_{t}q_{t-1}}(m_{t0}^{*},s,l_{t},z_{t})ds+\int_{l^{*}}^{l_{t}}T_{l_{t}q_{t-1}}(m_{t0}^{*},k_{t}^{*},s,z_{t})ds\right\} .
\end{align*}

From Assumption \ref{A-3-discrete} (g), for a given point $(x_{t-1},z_{t-1})\in\mathcal{X}\times\mathcal{Z}$,
we can find some point $(\tilde{m}_{t},\tilde{k}_{t},\tilde{l}_{t},\tilde{z}_{t})\in\mathcal{X}\times\mathcal{Z}$
such that $\partial G_{m_{t}|v_{t}}\left(\tilde{m}_{t}|\tilde{k}_{t},\tilde{l}_{t},\tilde{z}_{t},x_{t-1},z_{t-1}\right)/\partial m>0$.
Dividing (\ref{eq:m3-dis}) by (\ref{eq:m2-dis}) identifies $\partial\bar{h}\left(x_{t-1},z_{t-1}\right)/\partial q_{t-1}$
as
\begin{align*}
\frac{\partial\bar{h}\left(x_{t-1},z_{t-1}\right)}{\partial q_{t-1}} & =-\frac{\partial G_{m_{t}|v_{t}}\left(\tilde{m}_{t}|\tilde{k}_{t},\tilde{l}_{t},\tilde{z}_{t},x_{t-1},z_{t-1}\right)/\partial q_{t-1}}{\partial G_{m_{t}|v_{t}}\left(\tilde{m}_{t}|\tilde{k}_{t},\tilde{l}_{t},\tilde{z}_{t},x_{t-1},z_{t-1}\right)/\partial m}\frac{\partial\mathbb{M}_{t}^{-1}(\tilde{m}_{t},\tilde{k}_{t},\tilde{l}_{t},\tilde{z}_{t})}{\partial m}.
\end{align*}
Repeating this, we can identify $\partial\bar{h}\left(x_{t-1},z_{t-1}\right)/\partial q_{t-1}$
for all $\left(x_{t-1},z_{t-1}\right)\in\mathcal{X}\times\mathcal{Z}$.
From (\ref{eq:norm_dis}) and (\ref{eq:h-bar_dis2}), we can write
$\bar{h}_{t}(x_{t-1},z_{t-1})$ as

\begin{equation}
\bar{h}_{t}(x_{t-1},z_{t-1})=c_{2}(z_{t-1})+\Lambda_{\bar{h}}(x_{t-1},z_{t-1})\label{eq:hbar_dis}
\end{equation}

with
\begin{align*}
\Lambda_{\bar{h}}(x_{t-1},z_{t-1}) & :=\int_{m_{t-1}^{*}}^{m_{t-1}}\frac{\partial\bar{h}_{t}(s,k_{t-1},l_{t-1},z_{t-1})}{\partial m_{t-1}}ds\\
 & +\int_{k_{t-1}^{*}}^{k_{t-1}}\frac{\partial\bar{h}_{t}(m_{t-1}^{*},s,l_{t-1},z_{t-1})}{\partial k_{t-1}}ds+\int_{l_{t-1}^{*}}^{l_{t-1}}\frac{\partial\bar{h}_{t}(m_{t-1}^{*},k_{t-1}^{*},s,z_{t-1})}{\partial l_{t-1}}ds.
\end{align*}
Therefore, we can identify $\mathbb{M}_{t}^{-1}(m,k_{t},l_{t},z_{t})$
and $\bar{h}_{t}\left(x_{t-1},z_{t-1}\right)$ up to $\left\{ c_{0}(z),c_{2}(z)\right\} _{z\in\mathcal{Z}}$. 

Define $\widetilde{H}_{t}(z_{t},z_{t-1}):=E[\Lambda_{m}(m_{t},k_{t},l_{t},z_{t})-\Lambda_{\bar{h}}(x_{t-1},z_{t-1})|z_{t},z_{t-1}]$.
To determine $\left\{ c_{0}(z),c_{2}(z)\right\} _{z\in\mathcal{Z}}$,
we evaluate 
\begin{align*}
0 & =E\left[\eta_{t}|z_{t},z_{t-1}\right]\\
 & =E\left[\mathbb{M}_{t}^{-1}(m,k_{t},l_{t},z_{t})-\bar{h}_{t}\left(x_{t-1},z_{t-1}\right)|z_{t},z_{t-1}\right]\\
 & =\widetilde{H}_{t}(z_{t},z_{t-1})+c_{0}(z_{t})-c_{2}(z_{t-1})
\end{align*}
 at different values of $(z_{t},z_{t-1})\in\mathcal{Z}^{2}$. First,
evaluating $E\left[\eta_{t}|z_{t},z_{t-1}\right]=0$ at $z_{t}=z^{1}$,
and noting that $c_{0}(z^{1})=0$, we have
\[
c_{2}(z_{t-1})=\widetilde{H}_{t}(z^{1},z_{t-1}).
\]
Therefore, $c_{2}(z)$ is identified for all $z\in\mathcal{Z}$. Second,
evaluating $E\left[\eta_{t}|z_{t},z_{t-1}\right]=0$ at $z_{t-1}=z^{1}$,
we identify $c_{0}(z)$ as
\begin{align*}
c_{0}(z_{t}) & =c_{2}(z^{1})-\widetilde{H}_{t}(z_{t},z^{1})\\
 & =\widetilde{H}_{t}(z^{1},z^{1})-\widetilde{H}_{t}(z_{t},z^{1}).
\end{align*}
Given that $\left\{ c_{0}(z),c_{2}(z)\right\} _{z\in\mathcal{Z}}$
are identified, we can identified $\mathbb{M}_{t}^{-1}(m_{t},k_{t},l_{t},z_{t})$
and $\bar{h}_{t}(x_{t-1},z_{t})$ from (\ref{eq:M_dis}) and (\ref{eq:hbar_dis}).

Each firm's TFP $\omega_{it}=\mathbb{M}_{t}^{-1}(m_{it},k_{it},l_{it},z_{it})$
is identified up to scale and location normalization. From $E\left[\eta_{it}|x_{t-1},z_{t-1}\right]=0$,
we can identify $\bar{h}_{t}(x_{t-1},z_{t-1})=E\left[\omega_{it}|x_{t-1},z_{t-1}\right]$
and $\eta_{it}=\omega_{it}-\bar{h}_{t}(x_{it-1},z_{it-1})$. Thus,
we obtain the distribution of $\eta_{t}$, $G_{\eta_{t}}(\eta)$. 
\end{proof}
Note that the proofs for Lemma \ref{lem:step2} and Proposition \ref{P-step3}
do not rely on the continuity of $z_{t}$. Therefore, the exactly
same proof proves the following proposition. 
\begin{prop}
\label{P-step3-discrete} Suppose that Assumptions \ref{A-data},
\ref{A-2}, \ref{A-1-discrete}, \ref{A-3-discrete}, and \ref{A-FOC}
hold. Then, we can identify $\varphi_{t}^{-1}(\cdot)$ and $f_{t}(\cdot)$
up to scale and location and each firm's markup $\partial\varphi_{t}^{-1}(\bar{r}_{it},z_{it})/\partial r_{t}$
up to scale.
\end{prop}

\subsection{Demand Function with Unobservable Demand Shifter\label{subsec:quality}}

We derive the demand function (\ref{eq:quality}) form a representative
consumer's maximization problem. Suppose there are $I$ products.
Let $Y_{it}=\exp(y_{it})$ and $P_{it}=\exp(p_{it})$ be the output
and price levels of firm $i$. Consider a representative consumer's
utility maximization problem: 

\[
\max_{\{Y_{it}\}_{i=1}^{I}}U\left(u\left(\exp(\xi_{1t})Y_{1t},z_{1t}\right),....,u\left(\exp(\xi_{It})Y_{It},z_{It}\right)\right)\text{ s.t. }\sum_{i=1}^{I}P_{it}Y_{it}=Y_{t},
\]
where $Y_{t}$ is income, the upper tier utility $U(\cdot)$ is symmetric
in its arguments and the lower tier $u(\cdot)$ is common for all
products. Using $p_{it}^{\dagger}:=p_{it}-\xi_{it}$ and $y_{it}^{\dagger}:=y_{it}+\xi_{it}$,
the utility maximization problem is rewritten as
\[
\max_{\{y_{it}^{\dagger}\}_{i=1}^{I}}U\left(u\left(\exp\left(y_{1t}^{\dagger}\right),z_{1t}\right),....,u\left(\exp\left(y_{It}^{\dagger}\right),z_{It}\right)\right)\text{ s.t. }\sum_{i=1}^{I}\exp\left(p_{it}^{\dagger}\right)\exp\left(y_{it}^{\dagger}\right)=Y_{t}.
\]
The first-order condition for maximization is 
\[
U'\frac{\partial u\left(\exp\left(y_{it}^{\dagger}\right),z_{it}\right)}{\partial\exp\left(y_{it}^{\dagger}\right)}=\lambda_{t}\exp\left(p_{it}^{\dagger}\right),
\]
where $\lambda_{t}$ is the Lagrange multiplier and each firm takes
$\lambda_{t}$ and $U'$ as given under monopolistic competition.
The inverse demand function for firm $i$ is written as: 
\[
p_{it}^{\dagger}=\psi_{t}(y_{it}^{\dagger},z_{it}).
\]

\subsection{IID Productivity Shock\label{subsec:IID-Productivity-Shock}}

A firm receives an i.i.d. shock $e_{it}$ to output after choosing
inputs: 
\[
y_{it}=f_{t}(x_{it})+\omega_{it}+e_{it}.
\]
We suppose that firm's revenue $r_{it}$ is given by 
\begin{equation}
r_{it}=\varphi_{t}(y_{it},z_{it})=\varphi_{t}(f_{t}(x_{it})+\omega_{it}+e_{it},z_{it}).\label{eq:rev2}
\end{equation}
A firm chooses $m_{it}$ at time $t$ by maximizing the expected profit
conditional on the information available at the time denoted by $\mathcal{I}_{it}$
that includes all past variables and all time $t$ variables except
$e_{it}$:
\begin{align}
m_{it} & =\mathbb{M}_{t}(\omega_{it},k_{it},l_{it},z_{it})\nonumber \\
 & :=\argmax_{m\in\mathcal{M}}\ E\left[\exp\left(\varphi_{t}(f_{t}(m,k_{it},l_{it})+\omega_{it}+e_{it},z_{it})\right)|\mathcal{I}_{it}\right]-\exp(p_{t}^{m}+m)\nonumber \\
 & =\argmax_{m\in\mathcal{M}}\ E_{e}\left[\exp\left(\varphi_{t}(f_{t}(m,k_{it},l_{it})+\omega_{it}+e_{it},z_{it})\right)\right]-\exp(p_{t}^{m}+m),\label{eq:profit_maximization_2}
\end{align}
where $E_{e}$ is the expectation operator with respect to $e$.

The identification of $\varphi_{t}^{-1}(\cdot)$ and $f_{t}(\cdot)$
in the second step uses the conditional \emph{distribution} of $r_{t}$
given $w_{t}:=\left(x_{t},z_{t}\right)$, beyond the conditional \emph{expectation}
in Assumption \ref{A-data}. 
\begin{assumption}
\label{A-data-IID} The following information at time $t$ is known:
(a) the conditional distribution $G_{m_{t}|v_{t}}(\cdot)$ of $m_{t}$
given $v_{t}$; (b) the conditional distribution $G_{r_{t}|w_{t}}\left(r|w\right)$
of $r_{t}$ given $w_{t}:=\left(x_{t},z_{t}\right)$; (c) firm's expenditure
on material $\exp(p_{t}^{m}+m_{it})$. 
\end{assumption}

\subsubsection{Identification of Control Function and TFP}

Since $\mathbb{M}_{t}^{-1}(m_{it},k_{it},l_{it},z_{it})$ remains
a function of the same set of variables, Proposition \ref{P-step1}
holds with the same proof.
\begin{prop}
\label{P-step1-iid} Suppose that Assumptions \ref{A-1}, \ref{A-2},
\ref{A-3}, and \ref{A-data-IID} hold. Then, we can identify $\mathbb{M}_{t}^{-1}(m_{t},k_{t},l_{t},z_{t})$
up to scale and location for all $(m_{t},k_{t},l_{t},z_{t})\in\mathcal{M}\times\mathcal{K}\times\mathcal{L}\times\mathcal{Z}$
and identify $G_{\eta}(\cdot)$ up to scale.
\end{prop}

\subsubsection{Identification of Production Function}

We make the following assumption that corresponds to Assumption A1\textendash A3
and A5\textendash A6 in \citet{chiappori2015nonparametric}. (Assumption
\ref{A-1} (b) corresponds to Assumption A4 in \citet{chiappori2015nonparametric}.) 
\begin{assumption}
\label{A-3-iid} (a) The distribution $G_{e_{t}}(\cdot)$ of $e_{t}$
is absolutely continuous with a density function $g_{e_{t}}(\cdot)$
that is continuous on its support. (b) $e_{t}$ is independent of
$w_{t}:=\left(x_{t},z_{t}\right)'$ with $med(e_{t}|w_{t})=0$. (c)
$w_{t}$ is continuously distributed on $\mathcal{W}:=\mathcal{X}\times\mathcal{Z}$.
(d) The support $\mathcal{Y}$ of $y_{t}$ is an interval on $\mathbb{R}$
that contains $0$. (e) The set $\mathcal{B}_{q_{t}}:=\{x_{t}\in\mathcal{X}:\partial G_{r_{t}|w_{t}}(r|w_{t})/\partial q_{t}\neq0\text{ for every }(r_{t},z_{t})\in\mathcal{R}\times\mathcal{Z}\}$
is nonempty for some $q_{t}\in\{m_{t},k_{t},l_{t}\}$.
\end{assumption}
The conditional median restriction in Assumption \ref{A-3-iid}(b)
is location normalization. We continue to use the first-order condition
with respect to material as a restriction for identification.
\begin{assumption}
\label{A-4-iid}The first-order condition with respect to material
for the profit maximization problem (\ref{eq:profit_maximization_2})
holds for all firms as follows:
\begin{equation}
E_{e}\left[\exp\left(\varphi_{t}(\tilde{y}_{it}+e_{it},z_{it})\right)\frac{\partial\varphi_{t}(\tilde{y}_{it}+e_{it},z_{it})}{\partial\tilde{y}_{it}}\right]\frac{\partial f_{t}(x_{it})}{\partial m_{it}}=\exp(p_{t}^{m}+m_{it}),\label{eq:foc2}
\end{equation}
where $\tilde{y}_{it}:=f_{t}(x_{it})+\omega_{it}$ and the expectation
$E_{e}$ is taken with respect to $e_{it}$.
\end{assumption}
\begin{prop}
\label{P-step2-iid} Suppose that Assumptions \ref{A-1}, \ref{A-2},\ref{A-3},
\ref{A-data-IID}, \ref{A-3-iid}, and \ref{A-4-iid} hold. Then,
we can identify $\varphi_{t}^{-1}(\cdot)$, $f_{t}(\cdot)$, and $G_{e_{t}}(\cdot)$
up to scale and location.
\end{prop}
\begin{proof}
Because $\varphi_{t}$ is strictly increase in its first argument,
from $med(e_{t}|w_{t})=0$, we can identify 
\begin{align*}
\phi_{t}(x_{t},z_{t}) & :=\varphi_{t}(f_{t}(x_{t})+\mathbb{M}_{t}^{-1}\left(x_{t},z_{t}\right),z_{t})\\
 & =med(r_{t}|x_{t},z_{t}).
\end{align*}
 From 
\begin{align}
\varphi_{t}^{-1}(\phi_{t}(x_{t},z_{t}),z_{t}) & =f_{t}(x_{t})+\mathbb{M}_{t}^{-1}\left(x_{t},z_{t}\right),\label{eq:varphi_median}
\end{align}
the error term $e_{t}$ is expressed as 
\begin{equation}
e_{t}=\varphi_{t}^{-1}(r_{t},z_{t})-\varphi_{t}^{-1}(\phi_{t}(x_{t},z_{t}),z_{t}).\label{eq:error}
\end{equation}

From $e_{t}\perp w_{t}$ and $w_{t}:=(x_{t},z_{t})$, the conditional
distribution function $G_{r_{t}|w_{t}}(r_{t}|w_{t})$ satisfies
\begin{align}
G_{r_{t}|w_{t}}(r_{t}|w_{t}) & =G_{e_{t}|w_{t}}(\varphi_{t}^{-1}(r,z_{t})-f_{t}(x_{t})-\mathbb{M}_{t}^{-1}\left(x_{t},z_{t}\right)|w_{t})\nonumber \\
 & =G_{e_{t}}\left(\varphi_{t}^{-1}(r,z_{t})-f_{t}(x_{t})-\mathbb{M}_{t}^{-1}\left(x_{t},z_{t}\right)\right).\label{eq:CDF_epsilon}
\end{align}
For $q_{t}\in\{m_{t},k_{t},l_{t}\}$, the derivatives of (\ref{eq:CDF_epsilon})
are
\begin{align}
\frac{\partial G_{r_{t}|w_{t}}(r_{t}|w_{t})}{\partial r} & =\frac{\partial\varphi_{t}^{-1}(r_{t},z_{t})}{\partial r}g_{e_{t}}\left(\varphi_{t}^{-1}(r_{t},z_{t})-f_{t}(x_{t})-\mathbb{M}_{t}^{-1}\left(x_{t},z_{t}\right)\right),\label{eq:deriv1}\\
\frac{\partial G_{r_{t}|w_{t}}(r_{t}|w_{t})}{\partial q_{t}} & =-\left(\frac{\partial f_{t}(x_{t})}{\partial q_{t}}+\frac{\partial\mathbb{M}_{t}^{-1}\left(x_{t},z_{t}\right)}{\partial q_{t}}\right)g_{e_{t}}\left(\varphi_{t}^{-1}(r_{t},z_{t})-f_{t}(x_{t})-\mathbb{M}_{t}^{-1}\left(x_{t},z_{t}\right)\right),\label{eq:deriv2}\\
\frac{\partial G_{r_{t}|w_{t}}(r_{t}|w_{t})}{\partial z_{t}} & =\left(\frac{\partial\varphi_{t}^{-1}(r_{t},z_{t})}{\partial z_{t}}-\frac{\partial\mathbb{M}_{t}^{-1}\left(x_{t},z_{t}\right)}{\partial z_{t}}\right)g_{e_{t}}\left(\varphi_{t}^{-1}(r_{t},z_{t})-f_{t}(x_{t})-\mathbb{M}_{t}^{-1}\left(x_{t},z_{t}\right)\right).\label{eq:deriv3}
\end{align}
Using Assumption \ref{A-3-iid}(e), choose $q_{t}\in\{m_{t},k_{t},l_{t}\}$
and $\tilde{x}_{t}\in\mathcal{B}_{q_{t}}$ such that $\partial G_{r_{t}|w_{t}}\left(r_{t}|\tilde{x}_{t},z_{t}\right)/\partial q_{t}\neq0$
for all $(r_{t},z_{t})\in\mathcal{R}\times\mathcal{Z}$. Dividing
(\ref{eq:deriv1}) by (\ref{eq:deriv2}) and (\ref{eq:deriv3}) by
(\ref{eq:deriv2}), respectively, we obtain
\begin{align}
\frac{\partial\varphi_{t}^{-1}(r_{t},z_{t})}{\partial r} & =-\left(\frac{\partial f_{t}(\tilde{x}_{t})}{\partial q_{t}}+\frac{\partial\mathbb{M}_{t}^{-1}\left(\tilde{x}_{t},z_{t}\right)}{\partial q_{t}}\right)\frac{\partial G_{r_{t}|w_{t}}(r_{t}|\tilde{x}_{t},z_{t})/\partial r}{\partial G_{r_{t}|w_{t}}(r_{t}|\tilde{x}_{t},z_{t})/\partial q_{t}},\label{eq:varphi_r}\\
\frac{\partial\varphi_{t}^{-1}(r_{t},z_{t})}{\partial z_{t}}-\frac{\partial\mathbb{M}_{t}^{-1}\left(\tilde{x}_{t},z_{t}\right)}{\partial z_{t}} & =-\left(\frac{\partial f_{t}(\tilde{x}_{t})}{\partial q_{t}}+\frac{\partial\mathbb{M}_{t}^{-1}\left(\tilde{x}_{t},z_{t}\right)}{\partial q_{t}}\right)\frac{\partial G_{r_{t}|w_{t}}(r_{t}|\tilde{x}_{t},z_{t})/\partial z_{t}}{\partial G_{r_{t}|w_{t}}(r_{t}|\tilde{x}_{t},z_{t})/\partial q_{t}},\label{eq:varphi_z}
\end{align}
for all $r_{t}\in\mathcal{{R}}$.

Let $x_{t0}^{*}:=(m_{t0}^{*},k_{t}^{*},l_{t}^{*})$ and $r_{t}^{*}:=\phi_{t}(x_{t0}^{*},z_{t}^{*})$.
Then, the normalization Assumption \ref{A-2} implies:
\begin{align*}
\varphi_{t}^{-1}(r_{t}^{*},z_{t}^{*}) & =\varphi_{t}^{-1}(\phi_{t}(x_{t0}^{*},z_{t}^{*}),z_{t}^{*}).\\
 & =f_{t}(x_{t0}^{*})+\mathbb{M}_{t}^{-1}\left(x_{t0}^{*},z_{t}^{*}\right)\\
 & =0.
\end{align*}
Integrating (\ref{eq:varphi_r}) with respect to $r$ and using $\varphi_{t}^{-1}(r_{t}^{*},z_{t}^{*})=0$,
we obtain 
\begin{align}
\varphi_{t}^{-1}(r_{t},z_{t}^{*}) & =\int_{r_{t}^{*}}^{r_{t}}\frac{\partial\varphi_{t}^{-1}(s,z_{t}^{*})}{\partial r}ds\nonumber \\
 & =\left(\frac{\partial f_{t}(\tilde{x}_{t})}{\partial q_{t}}+\frac{\partial\mathbb{M}_{t}^{-1}\left(\tilde{x}_{t},z_{t}^{*}\right)}{\partial q_{t}}\right)S_{q_{t}}(r_{t}),\label{eq:varphi_1}
\end{align}
where
\begin{equation}
S_{q_{t}}(r_{t}):=-\int_{r_{t}^{*}}^{r_{t}}\frac{\partial G_{r_{t}|w_{t}}(s|\tilde{x}_{t},z_{t}^{*})/\partial r}{\partial G_{r_{t}|w_{t}}(s|\tilde{x}_{t},z_{t}^{*})/\partial q_{t}}ds>0\label{eq:S_q}
\end{equation}
is well-defined under Assumption \ref{A-3-iid}(e).

Define
\begin{equation}
c_{m}:=\frac{\partial f_{t}(\tilde{x}_{t})}{\partial q_{t}}+\frac{\partial\mathbb{M}_{t}^{-1}\left(\tilde{x}_{t},z_{t}^{*}\right)}{\partial q_{t}}.\label{eq:c_m}
\end{equation}
 From (\ref{eq:varphi_1}) and (\ref{eq:error}), $\varphi_{t}^{-1}(r_{t},z_{t}^{*})$
and $e_{t}$ are identified up to $c_{m}$ as: 
\begin{align}
\varphi_{t}^{-1}(r_{t},z_{t}^{*}) & =c_{m}S_{q_{t}}(r_{t})\label{eq:varphi_inv}\\
e_{t} & =c_{m}\left[S_{q_{t}}(r_{t})-S_{q_{t}}(\phi(x_{t},z_{t}^{*}))\right].\label{eq:e_t}
\end{align}
Because $e_{t}$ is independent of $z_{t}$ and $x_{t}$, we can identify
the distribution of $\tilde{e}_{t}:=e_{t}/c_{m}$ as $G_{\tilde{e}_{t}}(t)=\Pr(S_{q_{t}}(r_{t})-S_{q_{t}}(\phi(x_{t},z_{t}^{*}))\leq t|x_{t},z_{t}^{*})$
from (\ref{eq:e_t}).

Let $y_{t}:=\varphi_{t}^{-1}(r_{t},z_{t}^{*})=f(x_{t})+\mathbb{M}_{t}^{-1}\left(x_{t},z_{t}^{*}\right)+e_{t}$.
Then, (\ref{eq:varphi_inv}) implies
\begin{align}
\frac{y_{t}}{c_{m}} & =\frac{\varphi_{t}^{-1}(r_{t},z_{t}^{*})}{c_{m}}=S_{q_{t}}(r_{t}).\label{eq:ycm0}
\end{align}
Since $S_{q_{t}}(\cdot)$ is an increasing function, there exists
its inverse function $D(\cdot):=S_{q_{t}}^{-1}(\cdot)$ such that:
\begin{align}
r_{t}=\varphi_{t}(y_{t},z_{t}^{*}) & =D_{t}\left(\frac{y_{t}}{c_{m}}\right)\text{ and }\frac{\partial\varphi_{t}(y_{t},z_{t}^{*})}{\partial y_{t}}=\frac{1}{c_{m}}D'_{t}\left(\frac{y_{t}}{c_{m}}\right)\label{eq:varphi}
\end{align}
From $y_{t}-e_{t}=f(x_{t})+\mathbb{M}_{t}^{-1}\left(x_{t},z_{t}^{*}\right)=\varphi_{t}^{-1}(\phi_{t}(x_{t},z_{t}^{*}),z_{t}^{*})$,
(\ref{eq:ycm0}) implies 
\begin{align}
\frac{y_{t}}{c_{m}} & -\tilde{e}_{t}=\frac{\varphi_{t}^{-1}(\phi_{t}(x_{t},z_{t}^{*}),z_{t}^{*})}{c_{m}}=S_{q_{t}}(\phi_{t}(x_{t},z_{t}^{*})).\label{eq:ycm}
\end{align}

From (\ref{eq:varphi}) and (\ref{eq:ycm}), the expectation term
in the first-order condition (\ref{eq:varphi_inv}) for a firm with
$(x_{t},z_{t}^{*})$ times $c_{m}$ can be written as: 
\begin{align}
 & c_{m}E_{e}\left[\exp\left(\varphi_{t}(y_{t},z_{t}^{*})\right)\frac{\partial\varphi_{t}(y_{t},z_{t}^{*})}{\partial y_{t}}\right]\nonumber \\
= & c_{m}E_{e}\left[\exp\left(D_{t}\left(\frac{y_{t}}{c_{m}}\right)\right)\frac{1}{c_{m}}D_{t}'\left(\frac{y_{t}}{c_{m}}\right)\right]\text{ from }(\ref{eq:varphi})\nonumber \\
= & E_{e}\left[\exp\left(D_{t}\left(S_{q_{t}}(\phi(x_{t},z_{t}^{*}))+\tilde{e}_{t}\right)\right)D_{t}'\left(S_{q_{t}}(\phi(x_{t},z_{t}^{*}))+\tilde{e}_{t}\right)\right]\text{ from }(\ref{eq:ycm})\nonumber \\
= & \int\exp\left(D_{t}\left(S_{q_{t}}(\phi(x_{t},z_{t}^{*}))+\tilde{e}_{t}\right)\right)D_{t}'\left(S_{q_{t}}(\phi(x_{t},z_{t}^{*}))+\tilde{e}_{t}\right)dG_{\tilde{e}_{t}}(s)\nonumber \\
=: & \Upsilon(x_{t})\label{eq:upsilon}
\end{align}
where $\Upsilon(x_{t})$ is identified because $D_{t}(\cdot)$, $S_{q_{t}}\left(\cdot\right)$,
$\phi(\cdot)$, and $G_{\tilde{e}_{t}}(\cdot)$ are already identified.

From (\ref{eq:upsilon}), the first-order condition (\ref{eq:varphi_inv})
for a firm with $(x_{t},z_{t}^{*})$ becomes
\begin{equation}
\frac{\Upsilon(x_{t})}{c_{m}}\frac{\partial f_{t}(x_{t})}{\partial m_{t}}=\exp(p_{t}^{m}+m_{t}).\label{eq:foc_exp}
\end{equation}
Evaluating (\ref{eq:foc_exp}) at $(\tilde{x}_{t},z_{t}^{*})$ and
substituting it into (\ref{eq:c_m}), we identify $c_{m}$ as

\[
c_{m}=\frac{\Upsilon(\tilde{x}_{t})}{\Upsilon(\tilde{x}_{t})-\exp(p_{t}^{m}+\tilde{m}_{t})}\frac{\partial\mathbb{M}_{t}^{-1}\left(\tilde{x}_{t},z_{t}^{*}\right)}{\partial m_{t}}.
\]
Given that $c_{m}$ is identified, we identify $\varphi_{t}^{-1}(r_{t},z_{t}^{*})$
from (\ref{eq:varphi_inv}), $e_{t}$ from (\ref{eq:e_t}) and $f_{t}(\cdot)$
as
\[
f(x_{t})=\varphi_{t}^{-1}(\phi(x_{t},z_{t}^{*}),z_{t}^{*})-\mathbb{M}_{t}^{-1}\left(x_{t},z_{t}^{*}\right).
\]
Finally, we identify $\partial\varphi_{t}^{-1}(r_{t},z_{t})/\partial z_{t}$
from (\ref{eq:varphi_z}) as
\[
\frac{\partial\varphi_{t}^{-1}(r_{t},z_{t})}{\partial z_{t}}=-\left(\frac{\partial f_{t}(\tilde{x}_{t})}{\partial q_{t}}+\frac{\partial\mathbb{M}_{t}^{-1}\left(\tilde{x}_{t},z_{t}\right)}{\partial q_{t}}\right)\frac{\partial G_{r_{t}|w_{t}}(r_{t}|\tilde{x}_{t},z_{t})/\partial z_{t}}{\partial G_{r_{t}|w_{t}}(r_{t}|\tilde{x}_{t},z_{t})/\partial q_{t}}+\frac{\partial\mathbb{M}_{t}^{-1}\left(\tilde{x}_{t},z_{t}\right)}{\partial z_{t}}.
\]
and $\varphi_{t}^{-1}(r_{t},z_{t})$ as:
\[
\varphi_{t}^{-1}(r_{t},z_{t})=\varphi_{t}^{-1}(r_{t},z_{t}^{*})+\int_{z_{t}^{*}}^{z_{t}}\frac{\partial\varphi_{t}^{-1}(r_{t},s)}{\partial z_{t}}ds.
\]
\end{proof}

\subsubsection{Identification of Markup}

Because of the i.i.d. shock $e_{it}$, the first-order condition (\ref{eq:foc2})
includes the expectation with respect to $e_{it}$. Thus, the identified
value of $\partial\varphi_{t}^{-1}(r_{it},z_{it})/\partial r_{it}$
no longer equals the markup. Instead, we obtain the markup from the
cost minimization, following \citet{hall1988relation} and \citet{de2012markups}.

Consider a cost minimization problem of producing $\exp(\tilde{y}_{it})$
unit of output: 
\begin{equation}
C_{t}(\tilde{y}_{it},k_{it},l_{it}):=\min_{m}\exp(p_{t}^{m}+m)\text{ s.t. }\exp\left(f_{t}(m,k_{it},l_{it})+\omega_{it}\right)\geq\exp\left(\tilde{y}_{it}\right).\label{eq:cost_minimization}
\end{equation}
The first-order condition is 
\begin{equation}
\lambda_{it}\exp\left(\tilde{y}_{it}\right)\frac{\partial f_{t}(x_{t})}{\partial m_{t}}=\exp(p_{t}^{m}+m_{it})\label{eq:foc_cost}
\end{equation}
where $\lambda_{it}$ is the Lagrange multiplier and interpreted as
the marginal costs. Using the cost function (\ref{eq:cost_minimization}),
we write the profit maximization problem:
\begin{equation}
\max_{\tilde{y}_{jt}}E\left[\exp\left(\varphi_{t}(\tilde{y}_{it}+e_{it},z_{it})\right)|\mathcal{I}_{it}\right]-C_{t}(\tilde{y}_{it},k_{it},l_{it}).\label{eq:reduced problem}
\end{equation}
The first-order condition for (\ref{eq:reduced problem}) is 
\begin{equation}
E_{e}\left[\exp\left(\varphi_{t}(\tilde{y}_{it}+e_{it},z_{it})\right)\frac{\partial\varphi_{t}(\tilde{y}_{it}+e_{it},z_{it})}{\partial\tilde{y_{it}}}\right]=\frac{\partial C_{t}(\tilde{y}_{it},k_{it},l_{it})}{\partial\tilde{y}_{it}}=\lambda_{it}\exp\left(\tilde{y}_{it}\right).\label{eq:foc1}
\end{equation}
Substituting (\ref{eq:foc1}) into (\ref{eq:foc_cost}) obtains the
first-order condition (\ref{eq:foc2}) for the profit maximization
problem (\ref{eq:profit_maximization_2}). Therefore, the problem
(\ref{eq:reduced problem}) and the problem (\ref{eq:profit_maximization_2})
achieve the identical maximized profit.

From (\ref{eq:foc2}) and (\ref{eq:foc1}), the marginal cost $\lambda_{it}$
is expressed as 
\begin{align*}
\lambda_{it} & =\frac{E_{e}\left[\exp\left(\varphi_{t}(\tilde{y}_{it}+e_{it},z_{it})\right)\frac{\partial\varphi_{t}(\tilde{y}_{it}+e_{it},z_{it})}{\partial\tilde{y_{it}}}\right]}{\exp\left(y_{it}-e_{it}\right)}\\
 & =\frac{\exp(p_{t}^{m}+m_{it})/\frac{\partial f_{t}(x_{it})}{\partial m_{it}}}{\exp\left(y_{it}-e_{it}\right)}.
\end{align*}
Then, the markup becomes 
\[
\frac{\exp(p_{it})}{\lambda_{it}}=\frac{\partial f_{t}(x_{it})/\partial m_{it}}{\exp(p_{t}^{m}+m_{it})/\exp\left(r_{it}-e_{it}\right)},
\]
which is identified given our identification of $\partial f_{t}(x_{it})/\partial m_{it}$
and $e_{it}$.
\end{document}